\documentclass[sigconf,nonacm]{acmart} 

\AtBeginDocument{%
  \providecommand\BibTeX{{%
    \normalfont B\kern-0.5em{\scshape i\kern-0.25em b}\kern-0.8em\TeX}}}

\copyrightyear{2024}
\acmYear{2024}
\setcopyright{acmlicensed}\acmConference[KDD '24]{Proceedings of the 30th ACM SIGKDD Conference on Knowledge Discovery and Data Mining}{August 25--29, 2024}{Barcelona, Spain}
\acmBooktitle{Proceedings of the 30th ACM SIGKDD Conference on Knowledge Discovery and Data Mining (KDD '24), August 25--29, 2024, Barcelona, Spain}
\acmDOI{10.1145/3637528.3672048}
\acmISBN{979-8-4007-0490-1/24/08}

\usepackage[labelformat=simple]{subcaption}

\graphicspath{{figures/}}

\makeatletter
  \renewcommand\@acmplainheadfont{\scshape\bfseries}
  \renewcommand\@acmdefinitionheadfont{\scshape\bfseries}
\makeatother

\let\vec\boldsymbol%

\DeclareMathOperator*{\argmax}{argmax}
\DeclareMathOperator*{\rank}{rank}

\usepackage{mathtools}

\DeclarePairedDelimiter{\ceil}{\lceil}{\rceil}
\DeclarePairedDelimiter{\innerp}{\langle}{\rangle}
\newcommand\norm[1]{\left\lVert#1\right\rVert}

\usepackage{letltxmacro}
\LetLtxMacro{\oldsqrt}{\sqrt}
\renewcommand\sqrt[1][]{\oldsqrt[\leftroot{-1}\uproot{3}#1]}

\newcommand{\LL}[1]{\textcolor{black}{#1}} 
 
\newcommand{\eat}[1]{}
\usepackage{soul}

\usepackage{algorithm}
\usepackage[noend]{algpseudocode}

\usepackage{enumitem}
\setitemize{labelindent=0pt,leftmargin=*}

\usepackage{tabularx}
\newcolumntype{Y}{>{\centering\arraybackslash}X}

\newcommand{\cE}{\ensuremath{\mathcal{E}}}
\newcommand{\cK}{\ensuremath{\mathcal{K}}}
\newcommand{\cS}{\ensuremath{\mathcal{S}}}
\newcommand{\cT}{\ensuremath{\mathcal{T}}}
\newcommand{\cV}{\ensuremath{\mathcal{V}}}

\newcommand{\algoname}{HCF}

\makeatletter
\def\author@bx@sep{1pt}
\makeatother

\begin{document}

\title{Predicting Cascading Failures with a Hyperparametric Diffusion Model}

\author{Bin Xiang}
\affiliation{%
  \institution{CNRS@CREATE}
  \city{}
  \country{}
}
\email{bin.xiang@cnrsatcreate.sg}

\author{Bogdan Cautis}
\affiliation{%
  \institution{University of Paris-Saclay, CNRS LISN}
  \city{}
  \country{}
}
\email{bogdan.cautis@universite-paris-saclay.fr}

\author{Xiaokui Xiao}
\affiliation{%
  \institution{National University of Singapore}
  \city{}
  \country{}
}
\email{xkxiao@nus.edu.sg}

\author{Olga Mula}
\affiliation{%
  \institution{Eindhoven University of Technology}
  \city{}
  \country{}
}
\email{o.mula@tue.nl}

\author{Dusit Niyato}
\affiliation{%
  \institution{Nanyang Technological University}
  \city{}
  \country{}
}
\email{dniyato@ntu.edu.sg}

\author{Laks V.S. Lakshmanan}
\affiliation{%
  \institution{University of British Columbia}
  \city{}
  \country{}
}
\email{laks@cs.ubc.ca}

\begin{abstract}
In this paper, we study cascading failures in power grids through the lens of information diffusion models. Similar to the spread of rumors or influence in an online social network, it has been observed that failures (outages) in a power grid can spread \emph{contagiously}, driven by viral spread mechanisms. We employ a stochastic diffusion model that is Markovian (memoryless) and local (the activation of one node, i.e., transmission line, can only be caused by its neighbors). Our model integrates viral diffusion principles with physics-based concepts, by correlating the diffusion weights (contagion probabilities between transmission lines) with the \emph{hyperparametric} Information Cascades (IC) model. 
 We show that this diffusion model can be learned from traces of cascading failures, enabling  accurate modeling and prediction of failure propagation. This approach facilitates actionable information through well-understood and efficient graph analysis methods and graph diffusion simulations. Furthermore, by leveraging the hyperparametric model, we can predict diffusion and mitigate the risks of cascading failures even in unseen grid configurations, whereas existing methods falter due to a lack of training data. Extensive experiments based on a benchmark power grid and simulations therein show that our approach effectively captures the failure diffusion phenomena and guides decisions to strengthen the grid, reducing the risk of large-scale cascading failures. Additionally, we characterize our model's sample complexity, improving upon the existing bound. 
\end{abstract}

\ccsdesc[500]{Computing methodologies~Learning in probabilistic graphical models}
\ccsdesc[500]{Theory of computation~Sample complexity and generalization bounds}

\keywords{Cascading failures, diffusion graphs, independent cascades, hyperparametric model, learnability, power grids.}

\maketitle

\section{Introduction}%
\label{sec-introduction}
Interest in information diffusion and influence, initially motivated by social networks and viral marketing \cite{kempe2003maximizing,borgs2014maximizing}, has expanded significantly in recent years. Researchers have recognized the potential applicability of these methods beyond the social media sphere, as diffusion phenomena and algorithms for \emph{understanding} viral spread, \emph{leveraging} its reach or \emph{mitigating} its risks, are now being applied in public health and epidemiology \cite{fraser2007tracking}, transportation and logistics \cite{li2017optimizing}, internet virus propagation and cyber-security \cite{PhysRevLett.86.3200}, complex biological systems  \cite{https://doi.org/10.1046/j.1461-0248.2000.00130.x}, or ad-hoc communication networks \cite{6125993}.  

Besides these diverse fields, one critical domain stands out, pertaining to \emph{cascading failures} in power transmission networks, where viral spread phenomena can exert a significant impact, potentially disrupting our society. Just as rumors or influence may spread in an online social network, it has been observed and  confirmed by research that failures (outages) in a power grid are \emph{contagious} and can be described by viral spread mechanisms \cite{Kinney_2005, hines2009cascading}. Indeed, due to recent events like the 2021 Texas blackouts \cite{texas_blackouts_2021}, the issue of cascading grid failures has made  headlines and gained a lot of attention \cite{roth2022investigating}.
 In short, a cascading failure in a power grid represents the event of successive interdependent failures of components in the system. Usually initiated by one or a few source outages, due to endogenous or exogenous disturbances, they  propagate in a relatively short lapse of time, potentially leading to large blackouts~\cite{hines2009cascading}.

The analysis and prediction of such events is challenging because cascading failures are in general non-contiguous / non-local  with respect to the \textit{physical topology} of the power grid. This means that a failure of a specific transmission line may cause the failure of another one that is geographically distant and not directly coupled~\cite{hines2016cascading}, e.g., may be even located hundreds of kilometers away, as documented in the Western US blackout of July 2, 96~\cite{wscc}. This motivated researchers to study analytical frameworks that are based on graphs whose topology is not necessarily close to that the physical grid \cite{hines2010topological,cotilla2012comparing}. In particular, \emph{data-driven} approaches were proposed, leading to the development of  graph models of the observed interactions among components of the power grid \cite{hines2009cascading, hines2010topological, wu2021predicting, ghasemi2022higher}.   
See~\cite{nakarmi2020interaction} for a \LL{recent} survey on cascading failure analysis in power grids, based on  \emph{interaction graphs}, \LL{which model power grids with the nodes representing electrical components of interest such as buses or  transmission lines and the edges representing  interactions observed in known cascades of failures. }

However, constructing interaction / diffusion graphs that accurately capture cascading patterns hinges on two crucial factors: sufficient data quantity and high data quality. Only then can network analysis effectively reveal these patterns and ultimately guide decisions for mitigating cascading blackouts (e.g., by upgrading \LL{selected} transmission lines). Yet real-world historical data of cascading failures is hard to obtain, and the alternative to it, coming from grid simulators (e.g., quasi-steady-state or dynamic power network models~\cite{hines2016cascading,song2015dynamic}) may introduce artificial biases and errors. 
Furthermore, a significant limitation of these methods is the inability to adapt to \textit{unseen} grid configurations. They rely on past data, potentially missing rare or unique interactions  between transmission lines, \LL{which are latent}. This hinders their effectiveness in novel grid configurations and poses a significant challenge, as \LL{decisions involving grid configuration provisioning} may be time-critical. 

In this paper, we study cascading failures in power grids through the lens of information diffusion models. We  \LL{make use of} a classic stochastic diffusion model, known as Independent Cascades (IC), which is \emph{Markovian} (memoryless) and \emph{local} (the activation of one node, i.e., transmission line, can only be caused by its neighbors). Our model intentionally embraces such simplifying assumptions to exploit the efficiency and interpretability of established graph analysis tools and simulations for diffusion processes, thereby facilitating more readily actionable insights. Potential limitations in capturing intricate dependencies are counterbalanced by integrating \emph{physics-based concepts} into the viral diffusion model, aiming to (i) enhance model fidelity and (ii) reduce the sample complexity of the learning task. More precisely, we impose correlations on the diffusion weights (contagion probabilities between transmission lines) with the \emph{hyperparametric} IC  model of \cite{kalimeris2018learning}.  The model -- initially motivated by diffusion scenarios in social media -- assumes that each node in the diffusion graph has \emph{features} encapsulating its diffusion-relevant properties. Then, the diffusion weight between a pair of nodes is a function of their features and a global, low-dimensional hyperparameter.  Indeed, existing studies have demonstrated strong correlations between features of the transmission lines, which may be physics-based  or topology / connectivity-based, and their stability w.r.t. failures in the  grid~\cite{witthaut2016critical,titz2022predicting}. 
 Furthermore, since the diffusion probabilities are correlated, \emph{each observation} provides information about \emph{all
edges} in the network, thereby minimizing training data requirements for robust predictions.

\noindent  The main contributions of our paper can be summarized as follows: 
\begin{itemize}[topsep=5pt]
\item We revisit the hyperparametric IC model of \cite{kalimeris2018learning}, adapting it to the spread of failures in a diffusion graph having as nodes the transmission lines  of the power grid. 

\item We present a learning algorithm for this model, based on \LL{known cascades of failure traces}. Initializing with a complete graph representing all possible connections (accounting for non-local propagation), the  \LL{algorithm} learns diffusion probabilities over edges, effectively  sparsifying the network to reflect the most influential interactions.

\item Extensive experiments, based on a benchmark power grid and simulations therein, show that our approach can accurately model and predict how failures may propagate, and that its predictions can guide the decisions to consolidate the grid and thus reduce the risk of large cascading failures.  

\item Furthermore,  by leveraging the hyperparametric model, we show that we can make diffusion predictions and mitigate the risks of cascading failures even in \textit{unseen} grid configurations, for which existing methods cannot apply due to lack of training data.

\item Finally, we characterize the sample-complexity of our  \LL{learning algorithm}, improving upon the \LL{best known} bound of \cite{kalimeris2018learning}.  

\end{itemize}

\section{Related Work}%
\label{sec-related-work}
The state of the art for analyzing and learning cascading failures \LL{in power grids} can be mainly categorized into \emph{deterministic} and \emph{stochastic} ones. We note that while deep learning has shown promise in predicting cascading failures  (e.g.,~\cite{zhou2021power}), these approaches  lack interpretability and rigorous theoretical performance guarantees.

\emph{Deterministic methods} are mainly based on the OPA (ORNL-PSerc-Alaska) model~\cite{dobson2007complex,dobson2001initial}, which simulates the power system's response after contingencies of transmission line failures. They can provide detailed processes, tracing the cascading failures in the system, but they usually incur performance issues due to the computation of optimal power flow.

\emph{Stochastic approaches} adopt different types of models based on Markov chains or statistical learning. 
The main idea of these approaches is to estimate an influence (or contagion) probability matrix. Specifically, traditional Markov chain-based approaches~\cite{zhou2020markovian,rahnamay2016cascading} view a cascading process as a sequence of  transitions among states, where each state encapsulates the status of a group of nodes. The effectiveness of such an approach is impacted by the design of the state space and by the combinatorial characteristics of states. Specifically, by design, the individual node-to-node transition probabilities  are not directly estimated in such approaches. To address this issue, branching processes (a variant of Markov chains) have been used in~\cite{hines2016cascading,dobson2012estimating, hines2009cascading, hines2010topological}, where \LL{each failure in one stage is assumed to generate some random failures in the subsequent stage}, following an offspring distribution such as Poisson. 
Finally, the more recent statistical learning approaches~\cite{ghasemi2022higher,wu2021predicting} apply non-parametric regression models to fit the cascading processes, based on given historical traces of cascading failures. 

As mentioned previously, a major limitation of the aforementioned approaches is their inability to adapt to unseen grid configurations. More precisely, such models can only be applied to the same fixed power system configuration that produced the cascading logs. As shown in our experiments, any small change or perturbation in the power system may drastically impact their effectiveness. This poses a significant challenge, since some decisions may be time-critical and cannot be delayed until new cascading traces are produced  (e.g., by grid simulations) to retrain upon.

\emph{Information diffusion and influence} study was pioneered by the seminal works of~\cite{kempe2003maximizing,DBLP:conf/kdd/DomingosR01}. Diffusion models such as Independent Cascades (IC), Linear Threshold (LT), and generalizations thereof were introduced in~\cite{kempe2003maximizing}.  The problem of \emph{Influence Maximization} (IM) -- selecting a set of  seed nodes that  maximize the expected spread in a diffusion network under a certain diffusion model -- was also introduced and  studied extensively. See  the recent survey~\cite{10.1145/3604559} for a detailed review of the IM literature.

While classic IM methods assume  that the underlying diffusion network, including the influence probabilities associated with each connection, is known, this assumption rarely holds in practice. \eat{Instead, cascade data can be collected, and we have also seen a significant interest in the analysis of such information diffusion traces (e.g., tweet propagation in a microblogging platform), under the formal scope of \emph{Influence Estimation} (IE) \cite{GomezRodriguezS12,GomezRodriguezLK10,Gomez-RodriguezS12a}.} 
\LL{A rich body of work has been devoted to learning the underlying diffusion network when historical cascades of propagation traces are available~\cite{GomezRodriguezLK10, saito-etal-2008, DBLP:conf/wsdm/GoyalBL10, DBLP:conf/icml/DuLBS14, DBLP:conf/nips/NarasimhanPS15}.} 

Often, the drawback of such methods is their high sample-comple\-xity, as discussed in~\cite{kalimeris2018learning}, which proposes an alternative approach, with lower sample-complexity, based on a hyperparametric IC model. They assume nodes / edges have features, which induce correlations between the diffusion probabilities of different edges. This allows to minimize training data requirements for robust predictions, as also shown in \cite{kalimeris2019robust}. 

Our study starts with the thesis that methods of information diffusion analysis  from the rich IM literature may be applicable to other fields such as cascading failures in power grids, enabling us to leverage well-established techniques, thus leading more readily to actionable insights. Nevertheless, important differences have to be taken into account, such as the fact that  failures may propagate \emph{non-locally}, in a \emph{physical system} such as the power grid, while information in a social network may spreed to the connected ``neighbors'' in a sparse online network. By integrating \emph{physics-based concepts} into a 
diffusion model and by designing an adapted learning approach, we aim to address the specificity of this application scenario, in this way improving the model's effectiveness and its training data requirements such as sample-complexity. The merging of viral diffusion principles with physics-based
notions is supported by recent studies showing that  the physical features of a power grid exhibit strong correlations with the system's stability~\cite{witthaut2016critical,titz2022predicting}. This suggests that such features can be used to train a cascading failure model, one having a better potential for generalization and robustness to grid changes.  In comparison with the works discussed here, the hyperparametric approach allows us to (i) enhance model fidelity and (ii) reduce the sample complexity of the learning task. Moreover, its predictions can guide the decisions to consolidate the grid and thus reduce the risk of large CFs.  Importantly,  the hyperparametric approach enables us to make diffusion predictions and mitigate the risks of CFs even in \textit{unseen} grid configurations, for which existing methods cannot apply due to lack of training data.

\vspace{-3pt}

\section{Diffusion Model}%
\label{sec-system-model}
\subsection{Classic Diffusion Models} 

We start with the premise that CFs bear resemblance to the well-studied processes of information diffusion in generic settings, particularly using the classic Independent Cascade (IC) model~\cite{kempe2003maximizing}, which we briefly review next. 

\vspace{1mm}
\noindent 
\textbf{IC Model:} In the classic IC model, we have a  graph $\mathcal{G(V,E)}$ and a probability $p_{uv}$ associated with every edge $(u,v) \in \cE$. Diffusion proceeds in discrete time steps. At time $t=0$, only the seed nodes, the initiators of a cascade, are active. Once activated, a node remains active. Every node $u$ that became active at time $t>0$ has one chance to activate each of its inactive neighbors $v$, with probability $p_{uv}$. The propagation terminates when a fixpoint is reached and no more nodes can become active. The classic IM problem aims at finding $k$ seed nodes that lead to the maximum number of activated nodes in expectation, also known as \textit{(expected) spread}. 

\vspace{1mm}
\noindent 
\textbf{Hyperparametric IC Model:} An instance of IC model is characterized by the edge probabilities. Since these parameters may not be known or learned exactly, there has been an investigation of IM over a set of model instances corresponding to the uncertainty in our knowledge of edge probabilities \cite{chen2017robust, he2018stability, kalimeris2018learning}. A particularly elegant approach among those is the so-called hyperparametric IC model \cite{kalimeris2018learning}. It postulates a vector of features associated with each node which are relevant to the node exerting and experiencing influence. Examples of such features include age, gender, profession, degree, pagerank, etc. The influence between a pair of nodes is a function of the node features and a global low dimensional hyperparameter $\vec{\theta}$. More precisely, the hyperparametric model $H: \Theta \times X \rightarrow [0,1]$ restricts the IC model by imposing correlations among edge probabilities. Each edge $(u,v) \in \cE$ is associated with a $d$-dimensional vector $\vec{x}_{uv} \in X$, where $X \subseteq [0,1]^d$ which encodes the features associated with its endpoints. The probability of $u$ influencing $v$ is a function of $\vec{x}_{uv}$ and a low-dimensional hyperparameter, i.e., $H(\vec{x}_{uv}, \vec{\theta})$, where $\vec{\theta} \in \Theta \subseteq [-B,B]^d$, for some constant $B>0$.

\eat{ 
In this section, we first describe influence graphs in  Section \ref{sec-influence-graph} and then our adaptation of the hyperparametric Independent Cascades (IC) model, initially proposed in \cite{kalimeris2018learning}, for cascading failures in power grids (Section  \ref{sec-hp-model}). 
} 

\vspace{-1mm}

\subsection{Adaptation to Power Grids}  
\noindent 
\textbf{Diffusion graph:} 
Earlier attempts to model power grids as graphs, with nodes representing generators or buses and edges representing transmission lines were found to be ineffective in predicting failures as these can propagate in a manner that transcends the grid topology \cite{dobson2012estimating}. As such, we do not discuss these approaches further. 

Network models where nodes represent transmission lines and edges are based on observed cascades of failures have been found to be more successful in modeling and predicting failure cascades \cite{hines2009cascading}. However, 
unlike in applications such as social networks, the diffusion graph underlying cascading failures in power grids is not explicitly available and must  be learned from available cascades. In these applications, a set of cascades \cK\ is available, where for each $k \in \cK$, the cascade consists of a sequence of sets of nodes (i.e., transmission lines) $(\cV_0, \cV_1, ..., \cV_{\cT_k})$, where $\cV_t$ is the set of nodes that failed at time $t \in [0,\cT_k]$. It is assumed that any node that failed at time $t$ could influence any node that failed at the next time step $t+1$. With no further information available on failure propagation, we allow the possibility that the diffusion graph contains an edge $(u,v)$, for all $u\in \cV_t$ and $v\in \cV_{t+1}$. Given a set of such cascades $\cK$, we learn the most likely diffusion graph that explains all observed failure cascades using a learning algorithm, which tries to sparsify the above graph.

\eat{ 
\subsection{Influence graphs}
\label{sec-influence-graph}
A power grid can be modeled as a graph, where the nodes represent either generators or load buses, and edges correspond to transmission lines. The complex nature of cascading failures, as discussed in Section~\ref{sec-introduction}, transcends the physical topology of this initial graph representation, requiring the introduction of a directed \emph{influence graph}, $\mathcal{G(V,E)}$ instead. In this graph, each node $v\in\mathcal{V}$ represents a transmission line within the power grid, and each edge $e\in\mathcal{E}$ (i) represents the potential interactions between its endpoints, and (ii) is associated with an influence probability. Conceptually, the influence graph is a complete one, due to the unknown hidden interactions among possibly distant and not directly coupled transmission lines.

Cascading failures (CFs) recorded in the event data can be modeled as a sequence of failures of transmission lines with time stamps, i.e., for any given time $t\in \bR_{\geqslant 0}$, a subset $\mathcal{V}_{t}\subseteq\mathcal{V}$ comprises the transmission lines that have failed at time $t$.
The onset of cascading is triggered by the failure of $n\geq 1$ lines, referred to as $N-n$ contingencies, where $N=|\mathcal{V}|$ is the total number of lines.

Given the cascading event data, there are several ways to describe the interactions or influence among  transmission lines. A straightforward way is to interpret successive (pairwise) line failures as interacting, but this overlooks the synergistic effects multiple lines might exert on other lines. A more nuanced approach involves segmenting the data into sequences of \emph{generations / groups}, based on a temporal threshold, assuming that failures within one generation precipitate the failures in the subsequent one~\cite{dobson2012estimating}.

Our representation of historical event data will be as follows. Let $\mathcal{K}$ denote the set of all the cascades recorded in the event log, let $\mathcal{T}_k$ represent the set of discrete time steps in cascade $k\in\mathcal{K}$ (where the discrete time is for the convenience of representation), and let $\mathcal{V}_{k,t}\subseteq\mathcal{V}$ denote the group of nodes failed at step $t\in\mathcal{T}_k$.
Here, $\mathcal{V}_{k,0}$ marks the initially failed nodes, and $\mathcal{V}_{k,|\mathcal{T}_k|+1}=\varnothing$ indicate the end of the cascade $k$.
Then, the influence graph can be constructed based on all the consecutive groups $\{\mathcal{V}_{k,t},\,\mathcal{V}_{k,t+1}\}_{k\in\mathcal{K},t\in\mathcal{T}_k}$.
} 
\eat{ 
\subsection{Hyperparametric IC model}
\label{sec-hp-model} 
} 

\vspace{1mm}
\noindent 
\textbf{Choice of Model:} 
As discussed in Section \ref{sec-related-work}, most prior research on cascading failures (CFs) uses statistical or probabilistic frameworks, such as branching processes or Markov chains, to analyze interactions. However, these methods often work only for specific cases, rendering them less reliable for broader applications. We posit that hyperparametric modeling can offer a more effective solution, as evidenced by its potential in studies of information diffusion in social networks~\cite{kalimeris2018learning, kalimeris2019robust}. Nevertheless, CFs in power grids pose unique challenges: non-local cascades, complex networks, and intricate physical properties.
To address these issues, we leverage the physical and topological features of power grids to unravel the intricate dynamics of cascading failures. We propose an adaptation of the  hyperparametric model that quantifies influence probabilities using these features, coupled with data from observed CF events.
\eat{This model is particularly suited to capture the complex interplay between various factors influencing cascading failures.}

\vspace{1mm}
\noindent 
\textbf{Influence probability:}
Various functions can be used to model the relationship between features, the hyperparameter, and the likelihood of influence. We adopt here the logistic function, which is frequently used in the existing literature. Let $p_{uv}$ denote the influence probability from line $u$ to $v$, $\vec{x}_{uv}$  the vector of features of the endpoints (transmission lines) $u$ and $v$, and $\vec{\theta}$  the hyperparameter vector. Then, the influence probability function is 
\begin{equation} 
    p_{uv}=\frac{1}{1+e^{-\vec{\theta}^\mathsf{T}\vec{x}_{uv}}}.\label{eq_p}
\end{equation}

The feature values are assumed to be normalized, i.e., $\vec{x}_{uv}\in [-1,1]^d$, where $d$ is the dimensionality of the hyperparameter. While most features exhibit non-negative values, some, such as impedance, may span both negative and non-negative ranges.
The hyperparameter vector is confined within a specific hypothesis space, i.e., $\vec{\theta}\in\mathcal{H}=[-B, B]^d$, for some constant  $B>0$.

\eat{We start with the premise that CFs bear resemblance to the well-studied processes of information diffusion in generic settings, particularly to the classic independent cascading (IC) model~\cite{kempe2003maximizing}. To facilitate comparison, we briefly review this diffusion model here. The process begins with an initial set of active nodes, which then attempt to activate their inactive neighbors in a stochastic manner at the subsequent time step. Each \emph{newly} active node $u$ at time step $t$ \emph{independently} tries to activate each inactive neighbor $v$  once at time $t+1$, and succeeds with probability $p_{uv}$. An activated node remains active throughout the process, with the diffusion reaching a fixed point when no further activations occur.
} 

In the context of CFs, a line failure is analogous to the activation of a node in the IC model. Despite the complexity of failure propagations in power grids, it is often assumed for analytical purposes that the influence between elements is independent, simplifying the study of transmission line interactions \cite{hines2016cascading}.

In summary, the primary distinctions between CFs in a power grid and traditional IC settings include: (i) the influence graph for CFs, unlike the physical grid of a power system, is conceptual and initially viewed as a complete graph (due to non-contiguous events), whereas a social network's known topology serves as the influence medium, and (ii) the influence probability matrix for CFs will be linked to physical and topological features from the power grid, by a hyperparametric model. 

\eat{The analysis and representation of event data suggest that each diffusion step in a cascade $k\in\mathcal{K}$ is characterized by two consecutive sets of failed nodes, i.e., $\{\mathcal{V}_{k,t},\,\mathcal{V}_{k,t+1}\}$. Recall that $\mathcal{V}_{k,t}$ denotes the nodes failed at step $t$ in cascade $k$.
From the independent cascading assumption, a set of \emph{positive} samples $\mathcal{S}^+_{kt}$ can be defined to represent each diffusion step, written as:} 
Recall that a cascade of failures in a power grid consists of a sequence of sets of lines $(\cV_0, \cV_1, ..., \cV_{\cT_k})$ that fail at successive time steps. Recall that the influence graph $\mathcal{G(V,E)}$ starts as a complete one, where any node $u\in\mathcal{V}_{k,t}$ can potentially influence any node $v\in\mathcal{V}_{k,t+1}$. From the independent cascading assumption, a set of \emph{positive} samples $\mathcal{S}^+_{kt}$ can be associated with  each diffusion step, defined as:
\begin{equation}
    \mathcal{S}^+_{kt}=\{(u,v)\,|\, u \in \cV_{k,t}, v\in\mathcal{V}_{k,t+1}\},\;\forall k\in\mathcal{K},\forall t\in\mathcal{T}_k.
\end{equation}
That is, in cascade $k$, each $v\in\mathcal{V}_{k,t+1}$ at time $t+1$ is ``activated'' due to the set of nodes $\mathcal{V}_{k,t}$ \textcolor{black}{which became active} at time $t$.
\textcolor{black}{Note that \textit{any} node  $u\in \cV_{k,t}$ could have activated $v\in\mathcal{V}_{k,t+1}$, we have no information on exactly which one, and we use this pairwise notation  to adhere to the standard notation where every node is atomic and does not contain sets.}
For convenience, we sometimes abuse notation and denote the positive samples as $\mathcal{S}^+_{kt}=\{(\cV_{k,t},v)\,|\,  v\in\mathcal{V}_{k,t+1}\},\;\forall k\in\mathcal{K},\forall t\in\mathcal{T}_k.$ 

This setup also implicitly includes a set $\mathcal{S}^-_{kt}$ of \emph{negative} samples, representing node pairs that do not influence each other within a given time step: 
\begin{equation}
    \mathcal{S}^-_{kt}=\{(u,v)\,|\, u \in \cV_{k,t}, v\notin\textstyle\bigcup^{t+1}_{\tau=0}\cV_{k,\tau}\},\;\forall k\in\mathcal{K},\forall t\in\mathcal{T}_k.
\end{equation} 
Again, we sometimes abuse notation for convenience and write the negative sample as 
$\mathcal{S}^-_{kt} = \{(\cV_{k,t},v)\,|\,  v\notin\textstyle\bigcup^{t+1}_{\tau=0}\cV_{k,\tau}\},\;\forall k\in\mathcal{K},\forall t\in\mathcal{T}_k.$
\eat{ 
\begin{equation}
    \mathcal{S}^-_{kt}=\{(\mathcal{V}_{k,t},v)\,|\,v\in\mathcal{V}-\textstyle\bigcup^{t+1}_{\tau=0}\mathcal{V}_{k,\tau}\},\;\forall k\in\mathcal{K},\forall t\in\mathcal{T}_k.
\end{equation} 
} 
\LL{Let  $\vec{s}=(\mathcal{V}_{\vec{s}},v) \in \cS^+_{k,t}$ be a positive sample,} where $\mathcal{V}_{\vec{s}}$ is the set of nodes that may influence $v$ in sample $\vec{s}$. Based on Eq.~\eqref{eq_p}, the \LL{likelihood} of a positive sample in one diffusion step, \LL{for any} $\vec{s}\in\mathcal{S}^+_{kt}$, is:
\begin{equation}
    P(\vec{\theta}\,|\,\vec{x},\vec{s})=1-\textstyle\prod_{u\in\mathcal{V}_{\vec{s}}}(1-p_{uv}).\label{eq_p_theta}
\end{equation}
\LL{For any negative sample $\vec{s}=(\mathcal{V}_{\vec{s}},v) \in \cS^-_{k,t}$ , the likelihood is $1-P(\vec{\theta}\,|\,\vec{x},\vec{s})$.} 

\vspace{1mm}
\noindent 
\textbf{Estimator:}
Given this CF model, the traditional Maximum Likelihood Estimation (MLE) method can be applied to estimate the hyperparameter. The idea is to maximize the likelihood between the predictions of the hyperparametric model and the ground truth of event data.

Based on Eq.~\eqref{eq_p_theta}, the likelihood of one diffusion step can be formulated as $\prod_{\vec{s}\in\mathcal{S}^+_{kt}}P(\vec{\theta}\,|\,\vec{x},\vec{s})\prod_{\vec{s}\in\mathcal{S}^-_{kt}}(1-P(\vec{\theta}\,|\,\vec{x},\vec{s}))$.
Let $y=\mathbf{I}_{\mathcal{S}^+_{kt}}(\vec{s})$ denote an indicator function which equals 1 if $\vec{s}\in\mathcal{S}^+_{kt}$, 0 otherwise.
Then, the log-likelihood, known as the cross-entropy is: 
\begin{equation}
    L(\vec{\theta}\,|\,\vec{x},\vec{s},y)=y\log P(\vec{\theta}\,|\,\vec{x},\vec{s})+(1-y)\log(1-P(\vec{\theta}\,|\,\vec{x},\vec{s})). \label{eq_l_loss}
\end{equation}
Let $\mathcal{S}=\{(\vec{s},y)\,|\,\vec{s}\in\mathcal{S}^+_{kt}\cup\mathcal{S}^-_{kt}, y=\mathbf{I}_{\mathcal{S}^+_{kt}}(\vec{s}), k\in\mathcal{K},t\in\mathcal{T}_k\}$ denote all the samples of the cascade events.
\eat{For convenience, we also define the} \LL{Notice that the  set of positive samples  is $\mathcal{S}^+=\{\vec{s} \mid (\vec{s},1)\in\mathcal{S}\}$ and the set of negative samples is  $\mathcal{S}^-=\{\vec{s} \mid (\vec{s},0)\in\mathcal{S}\}$. } 

Then, the expected log-likelihood over $\mathcal{S}$ can be written as:
\begin{equation}
    \mathcal{L}_{\mathcal{S}}(\vec{\theta}\,|\,\vec{x})=\frac{1}{|\mathcal{S}|}\sum_{(\vec{s},y)\in\mathcal{S}}L(\vec{\theta}\,|\,\vec{x},\vec{s},y).\label{eq_ls}
\end{equation}

Finally, the empirical estimator can be written as:
\begin{equation}
    \hat{\vec{\theta}}=\argmax_{\vec{\theta}\in\mathcal{H}}\mathcal{L}_{\mathcal{S}}(\vec{\theta}\,|\,\vec{x}).\label{eq_theta_h}
\end{equation}

\section{Learnability}%
\label{sec-learnability}
In this section, we study the \emph{Probably Approximately Correct} (PAC) learnability~\cite{valiant1984theory,shalev2014understanding} of our model, drawing on the theory of sample complexity analysis~\cite{shalev2014understanding} for the MLE approach. The log-likelihood function family w.r.t. the cascading failure model, which is initially part of an infinite hypothesis space, is transformed into a finite hypothesis space using covering theory and Lipschitz continuity analysis for the diffusion probability function. This transformation allows us to examine the complexity more effectively. We examine the conditional concavity of the empirical log-likelihood function, which paves the way for applying Rademacher complexity to assess the model's \LL{sample} complexity. Rademacher complexity is crucial as it evaluates the expressiveness of a function class by its ability to fit a hypothesis set to a random distribution, which is closely linked to sample complexity~\cite{mohri2018foundations}.

Building on this foundation, we derive the sample complexity for our model. 
Detailed proofs of these theoretical findings can be found in 
 the appendix. 

\begin{definition}[Agnostic PAC learnability~\cite{valiant1984theory,shalev2014understanding}]
    A hypothesis class $\mathcal{H}$ is PAC learnable if, for any $\varepsilon,\delta\in(0,1)$, and for any distribution $\mathcal{D}$ over the product space  of examples $\mathcal{X}$ and labels $\mathcal{Y}$, there exists a polynomial function $m_{\mathcal{H}}$ and a learning algorithm $\mathcal{A}$, such that when $\mathcal{A}$ is run on $\geq m_{\mathcal{H}}(\varepsilon,\delta)$ i.i.d.\ samples from $\mathcal{D}$, \eat{of size $m\geqslant m_{\mathcal{H}}(\varepsilon,\delta)$,} it produces a hypothesis $h$ that, with probability $\geq 1-\delta$, achieves a loss $\mathcal{L}_\mathcal{D}(h)$ that is within $\varepsilon$ of the minimum possible loss over all hypotheses in $\mathcal{H}$.    
\end{definition}

\LL{Following~\cite{narasimhan2015learnability,kalimeris2018learning}, we assume that the influence probability $p_{uv}$ is restricted to the interval $[\lambda, 1-\lambda]$, where $\lambda\in(0,0.5)$ is a constant that controls the precision of our estimates. W.l.o.g., we also assume that each node in the network has at least one significant feature. Together, these assumptions allow us to  bound the magnitude of the influence weights, which \eat{, when combined with the assumption that each node in the network possesses at least one significant feature,} enables us to define the range of the hypothesis space $\mathcal{H}$. }

We next analyze the Lipschitz continuity of the log-likelihood function, based on its gradient:
\begin{equation}
    \nabla_{\vec{\theta}} L(\vec{\theta}\,|\,\vec{x},\vec{s},y)
    = \left(\frac{y}{P(\vec{\theta}\,|\,\vec{x},\vec{s})}-\frac{1-y}{1-P(\vec{\theta}\,|\,\vec{x},\vec{s})}\right)\nabla_{\vec{\theta}} P(\vec{\theta}\,|\,\vec{x},\vec{s}),\label{eq_dL}
\end{equation}
where the gradient of $P$ w.r.t.\ $\vec{\theta}$ is given by:
\begin{equation}
    \nabla_{\vec{\theta}} P(\vec{\theta}\,|\,\vec{x},\vec{s})
    = \sum_{u\in\mathcal{V}_{\vec{s}}}\nabla_{\vec{\theta}}p_{uv} \prod_{u'\in\mathcal{V}_{\vec{s}}\backslash \{u\}}(1-p_{u'v}).\label{eq_dP_o}
\end{equation}
Previous work~\cite{kalimeris2018learning,narasimhan2015learnability} had established  a relatively loose bound of the Lipschitz continuity of $L(\cdot)$ w.r.t.\ $\ell_q$-norm as follows:
    \[\norm{\nabla L}_q=\left|\frac{y}{P}-\frac{1-y}{1-P}\right|\norm{\nabla P}_q\leqslant\max\{\frac{1}{P},\frac{1}{1-P}\}\cdot\sup\norm{\nabla P}_q.\]
    Based on the fact that $\norm{\nabla p_{uv}}_{\infty}\leqslant 1$, $p_{uv}\in[\lambda,1-\lambda]$, and the definition of $\nabla P$ (see Eq.~\eqref{eq_dP_o}), we have $\norm{\nabla P}_{\infty}\leqslant N(1-\lambda)^{N-1}$, where $N\leqslant|\mathcal{V}|$ is the maximum number of  active neighbor nodes.
    \LL{From Eq.~\eqref{eq_p_theta}, $P(\vec{\theta}\,|\,\vec{x},\vec{s})\in[\lambda,1-\lambda^{V_{\vec{s}}}]$, from which we can show $\max\{\frac{1}{P},\frac{1}{1-P}\}=\frac{1}{\lambda^N}$.} 
    This leads to a loose Lipschitz bound of $\frac{\rho}{\lambda^N}$, where $\rho=N(1-\lambda)^{N-1}\sqrt[q]{d}$.
    Note that $\lambda>0$ is a small value defining the precision, e.g., $10^{-9}$. 
    We derive a tighter bound in what follows. 

\begin{lemma}\label{lemma_l_lip}
    The log-likelihood function $L(\vec{\theta}\,|\,\vec{x},\vec{s},y)$
is bounded and $(V_{\vec{s}}\sqrt[q]{d}\log\frac{1}{\lambda})$-Lipschitz w.r.t.\ $\ell_q$-norm  $\forall q\geqslant 1$, where $d=\mathsf{dim}(\vec{x})$.
\end{lemma}

To examine the sample complexity of our model, we first investigate whether the log-likelihood function is concave. The idea is that if it is concave, then we can derive the sample complexity based on optimization theory~\cite{shalev2014understanding}. If not, then a canonical approach is to utilize the Rademacher complexity framework. The following lemma settles the question.

\begin{lemma}\label{lemma_l_nonc}
    The log-likelihood for one sample, i.e., $L(\vec{\theta}\,|\,\vec{x},\vec{s},y)$, is concave in $\vec{\theta}$ if the sample is either negative with $y=0$, or positive with $y=1$ and $|\mathcal{V}_{\vec{s}}|=1,\forall\vec{s}\in\mathcal{S}^+$. 
     Otherwise, it is not concave.
\end{lemma}

\eat{By Lemma~\ref{lemma_l_nonc}, the expected log-likelihood $\mathcal{L}_{\mathcal{S}}(\vec{\theta}\,|\,\vec{x})$ depends on the dataset $\mathcal{S}$. For example, suppose that $\mathcal{S}$ contains only one specific type of positive samples: $\{((\mathcal{V}\backslash u,\,u), 1), \forall u\in\mathcal{V}\}$, where $|\mathcal{V}|>1$.} 
\LL{It follows from Lemma~\ref{lemma_l_nonc} that, in general, the expected log-likelihood $\mathcal{L}_{\mathcal{S}}(\vec{\theta}\,|\,\vec{x})$ over $\mathcal{S}$ is non-concave.} More results on the general and conditional concavity analysis can be found in  
 the appendix. 

\eat{Since $\mathcal{L}_{\mathcal{S}}(\vec{\theta}\,|\,\vec{x})$ is non-concave in general,} Given this, we resort to Rademacher bound theory to characterize the sample complexity of our model~\cite{shalev2014understanding}.
Let $\mathcal{F}$ be the family of log-likelihood functions $L(\vec{\theta}\,|\,\vec{x},\vec{s},y)$ on hypothesis $\vec{\theta}\in\mathcal{H}$, defined as:
\[\mathcal{F}\stackrel{\text{def}}{=}\{f:(\vec{s},y)\mapsto L(\vec{\theta}\,|\,\vec{x},\vec{s},y)\,|\,\vec{\theta}\in\mathcal{H}\}.\]
\LL{Let $\mathcal{F(S)}$ denote the set of vectors of likelihood values to which $\mathcal{F}$ evaluates each sample in $\mathcal{S}$, defined as:
\[\mathcal{F(S)}\stackrel{\text{def}}{=}\{[f(\vec{s},y)]_{\vec{s},y\in\mathcal{S}}\,|\,f\in\mathcal{F}\}.\]} 
We have the following lemmas.

\begin{lemma}[\textbf{Covering Number}]\label{lemma_c_num}
    Let $\mathcal{N}_{\infty}(\varepsilon,\mathcal{F(S)})$ be the $\ell_{\infty}$-norm $\varepsilon$-covering number of $\mathcal{F(S)}$, then
    \begin{equation}
        \mathcal{N}_{\infty}(\varepsilon,\mathcal{F(S)})\leqslant\ceil*{\frac{B V\log\frac{1}{\lambda}}{\varepsilon}}^d,
    \end{equation}
    where $V=\max_{\vec{s}\in\mathcal{S}}V_{\vec{s}}\leqslant|\mathcal{V}|$ is the maximum number of  activated nodes.
\end{lemma}

\begin{lemma}[\textbf{Rademacher Bound}]\label{lemma_r_bound}
    Let $\mathcal{R(F(S))}$ denote the Rademacher complexity of $\mathcal{F}$ w.r.t. $\mathcal{S}$ with \LL{sample} size $m=|\mathcal{S}|$, then
    \begin{equation}
        \mathcal{R(F(S))}\leqslant V\log{\frac{1}{\lambda}}\sqrt{\frac{2\log\mathcal{N}_{\infty}(\varepsilon,\mathcal{F(S)})}{m}}+\varepsilon.
    \end{equation}
\end{lemma}

Let $\mathcal{L}_{\mathcal{D}}(\vec{\theta}\,|\,\vec{x})$ denote the expected  log-likelihood on a hypothesis $\vec{\theta}\in\mathcal{H}$ w.r.t. the distribution $\mathcal{D}$ over data $(\vec{s},y)$ and a fixed feature $\vec{x}$, defined as:
\begin{equation}
    \mathcal{L}_{\mathcal{D}}(\vec{\theta}\,|\,\vec{x})\stackrel{\text{def}}{=}\mathbb{E}_{\vec{s},y\sim\mathcal{D}}[L(\vec{\theta}\,|\,\vec{x},\vec{s},y)].
\end{equation}

Based on Lemmas~\ref{lemma_c_num} and~\ref{lemma_r_bound}, we have the following result.
\begin{lemma}[\textbf{Sample Complexity}]\label{lemma_s_complx}
    For every $\varepsilon,\delta\in(0,1)$ and data distribution $\mathcal{D}$, when running the MLE algorithm on $\geqslant m(\varepsilon,\delta)$ i.i.d. samples drawn from $\mathcal{D}$, where
    \[m(\varepsilon,\delta)=\mathcal{O}\left(\frac{V^2\log^2{\frac{1}{\lambda}}}{\varepsilon^2}\left(d\log\frac{B V\log\frac{1}{\lambda}}{\varepsilon} + \log{\frac{1}{\delta}}\right)\right),\]
    the algorithm returns a hypothesis $\hat{\vec{\theta}}$ (i.e., Eq.~\eqref{eq_theta_h}) that satisfies,
    w.p. at least $1-\delta$,
    \begin{equation}
        \sup\nolimits_{\vec{\theta}\in\mathcal{H}}\mathcal{L}_{\mathcal{D}}(\vec{\theta}\,|\,\vec{x})-\mathcal{L}_{\mathcal{D}}(\hat{\vec{\theta}}\,|\,\vec{x})\leqslant\varepsilon.
    \end{equation}
\end{lemma}

    \textit{It immediately follows from  Lemma~\ref{lemma_s_complx} that the hypothesis for the cascading failure model is PAC learnable. }

\section{Overall Solution}%
\label{sec-solution}

Section~\ref{sec-learnability} establishes that the expected log-likelihood is in general non-concave. The standard solution strategies for this type of problem include stochastic gradient descent (SGD) or a quasi-Newton approach such as L-BFGS-B~\cite{byrd1995limited}. In our work, we adopt the latter, as it  offers \LL{faster convergence} and can be \LL{more} easily parallelized.

From a computational standpoint, the availability of an analytic gradient enables us to accurately compute values for either the entire dataset, or for large batches thereof in parallel. This computation is facilitated by a GPU-accelerated workflow, designed to leverage the structural properties of the gradient and of the dataset. Specifically, observe that the gradient of the log-likelihood functions mainly consists of terms such as  $P(\vec{\theta}\,|\,\vec{x},\vec{s})$ and $\hat{\vec{x}}_{\vec{s}}(\vec{\theta})$.
With $\vec{\theta}$ fixed in each iteration, the terms ${[e^{\vec{\theta}^\mathsf{T}\vec{x}_{uv}}]}_{uv\in\mathcal{V}^2}$ can be pre-computed. The subsequent computation of the terms $P(\vec{\theta}\,|\,\vec{x},\vec{s})$ and $\hat{\vec{x}}_{\vec{s}}(\vec{\theta})$ involves simple operations such as summation or multiplication with $u\in\mathcal{V}_{\vec{s}}$.
Consequently, the gradient $\nabla_{\vec{\theta}} L(\vec{\theta}\,|\,\vec{x},\vec{s},y)$ for individuals, as well as $\nabla_{\vec{\theta}}\mathcal{L}_{\mathcal{S}}(\vec{\theta}\,|\,\vec{x})$ for entire sample set, can be computed in parallel.

Moreover, as discussed in Sections~\ref{sec-system-model} and~\ref{sec-learnability}, the interactions within our model form a complete graph. For each cascade $k\in\mathcal{K}$, a given sample $((\mathcal{V}_{\vec{s}},v),1)\in\mathcal{S}$ implies the existence of samples $\{((\mathcal{V}_{\vec{s}},v'),0)\,|\,v'\in\mathcal{V}\backslash \mathcal{V}_{\vec{s}} \cup \{v\}\}$, for all $v'$ s.t. $((\mathcal{V}_{\vec{s}},v'),1)\notin \mathcal{S}$.  
Therefore, $|\mathcal{S}|\geqslant|\mathcal{K}||\mathcal{V}|$.
When accounting for the $|\mathcal{V}|-n$ contingencies cascading (see Section~\ref{sec-system-model}, $n\geqslant 1$), the size of $S$ is $\mathcal{O}(|\mathcal{V}|^{n+1})$.
Despite the relative sparsity of samples, they can be efficiently encoded and aggregated using the sample key $(\mathcal{V}_{\vec{s}},v)$ into different matrices for optimized storage and computation.

Algorithm~\ref{lst_lw} summarizes the overall workflow for learning a Hyperparametric diffusion model for predicting Cascading Failures (denoted \algoname).
\begin{algorithm}[H]
    \small
    \caption{Learning workflow of \algoname}
    \label{lst_lw}
    \begin{algorithmic}[1]
        \Statex \textbf{Input:} cascading failures, power grid topology and settings.
        \Statex \textbf{Output:} predicted cascading failures.
        \State Extract physical/topological features $\vec{x}$ from power grid
        \State Encode samples $\mathcal{S}$ and construct MLE optimization (see Eq.~\eqref{eq_theta_h})
        \State Run L-BFGS-B for MLE to learn the hyperparameter $\vec{\theta}$
        \State Compute diffusion probability matrix based on Eq.~\eqref{eq_p}
        \State Run Monte Carlo simulations to predict cascading failures
    \end{algorithmic}
\end{algorithm}

\section{Experiments}%
\label{sec-experiments}

Our approach was implemented in Python using GPUs, with the OpenCL standard\footnote{Source code is available at \url{https://github.com/unkux/Learn_CF}.}. The experiments were mainly performed on a laptop Intel(R) Core(TM) i7-8550U CPU @ 1.80GHz with 16 GB memory (denoted L) and a server with AMD EPYC 7713 processor and 128 GB memory (denoted S).

\noindentparagraph{\bf Datasets.} 
Acquiring real-world cascading failure datasets in the power systems community is challenging, due to privacy concerns and confidentiality. Available data, such as the one used in \cite{zhou2018can}, often lacks crucial physical settings like power supply / demand, while the topology can only be partially recovered  from the cascading traces. This limited information hinders the application of parametric models that depend on it for reasoning.

Therefore, for our analysis, we leverage the IEEE300 standard power grid (i.e, a 300Bus having $516$ transmission lines), which offers comprehensive details like power supply / demand, line impedance, capacities, and topology. This information facilitates the extraction of the relevant physical features used in our learning process. Specifically, we extract $22$ basic features based on the pre-cascading balanced power flow and grid topology.

To generate CF events, we leverage the widely used DCSepSim~\cite{eppstein2012random} simulator. We simulate power flow and cascading processes by strategically cutting lines, triggering power imbalances, and rerouting flow. If overloaded lines exceed their capacity, outages occur, cascading through the system until equilibrium is reached. We generate 100k Monte Carlo traces per power grid instance, by having each transmission line fail with a probability of $1/516$. 

A similar experiment on a dataset obtained from the larger \emph{Polish grid} is described in the appendix.

\eat{
For obtaining the cascading failure event data, we adapt a well-known open source simulator DCSepSim~\cite{eppstein2012random} to simulate the power flow and cascading process. The basis idea is that: each time breaking a selected initial lines, which will lead to the imbalance of the system, then, the power flow needs to be re-dispatched in order to host the missing flow and balance the power system. During this process, if the power flow re-dispatched exceeds the limit/capacity of others lines, then, outages is happening, leading to the failure cascading to other lines. This process will propagate to the rest of the system until reaching a balanced state. In experiments, for each instance of power grid, we run 100k Monte Carlo cascading failure traces.
}

\noindentparagraph{\bf Validation and testing approach.}
Due to the stochastic nature of our model, the traditional cross-validation and testing process can not be directly applied. Instead, we propose a different validation and testing process, relying on the IEEE300 standard power grid and the DCSepSim simulations. Given an initial power grid configuration, we apply perturbations on the power demand settings or on the physical topology, to obtain a set of different power grid instances. The model will be trained on only one of the power grid instances, and tested on the others.
 In some of the experiments, the initial power grid instance is close to saturation with respect to the power supply / demand -- and thus more inclined to exhibit cascading failures -- in order to observe the performance of the tested methods under extreme conditions.
Here we define $2$ subscript notations, in the context of a each figure: ${ii}$ for a model learned from the instance  $i$ and tested on that same instance, 
and ${i1}$ for a model learned on instance $1$ in that figure and tested on instance $i$.  The top 5\% most critical transmission lines w.r.t. cascading failure size were considered for the evaluation.

\noindentparagraph{\bf Evaluation metrics.}
Our evaluation focuses on two error types related to predicting cascading failures. The first metric, \textbf{Distribution Error} (DE), compares the expected number of failures per line across the actual DCSepSim simulations  and the model-generated cascades. For each line in the cascading failure dataset, we count the failures. Then, using the trained model and Monte Carlo diffusions, we generate new cascades and repeat the count on those simulated lines. Finally, we compare the distributions of failure counts using either mean absolute error or relative error. 

The second metric, \textbf{Probability Error} (PE), is based on the probability matrix, and can only be applied for our hyperparametric model. We stress that in the cascading failures data the interactions between lines are only partially observed, i.e., given a cascading sample $(\mathcal{V}_{\vec{s}},\,v)$, there is no information for which node $u\in\mathcal{V}_{\vec{s}}$ triggered the activation of node $v$. This means that simple frequency-based estimation would lead to biased results due to over-counting, hence ground truth influence probabilities are unavailable. To mitigate this issue, we define $3$ types of  mean absolute / relative errors (using the same subscript notation): $P_{ii}-P_{i1}$ for the prediction error, $P_{ii}-P_{11}$ for the probabilities' changes from instance $i$ to instance $1$, and $P_{i1}-P_{11}$ for the probabilities' changes  when applying the model learned on instance $1$ in the figure to instance $i$. 

\noindentparagraph{\bf Baseline methods.}
\eat{
The main goal of this work is to learn a hyperparametric diffusion model for cascading failures, that can be not only used for the power grid with a single fixed configuration, but also can be generalized for predicting the cascading failures for other power grids with different demands and topologies without the needs of cascading dataset or frequent retraining process, both of which are too expensive. To our knowledge, there is no existing works towards the latter direction, most of them focus on learning and predicting for the power grid with single setting.
}
Recall that existing non-parametric approaches, like ~\cite{hines2016cascading,wu2021predicting,ghasemi2022higher}, struggle with even slight changes in a power grid (e.g., power demand or topology). Trained on data from a specific grid instance, they lack generalizability and require retraining for new instances. Due to this reliance on specific grid data, our method offers an advantage under certain conditions. If the topology remains fixed while other factors like power demand change, we can directly compare our approach with the existing ones, by evaluating their pre-trained models on unseen instances. However, if the topology itself undergoes changes, adapting these models for comparison on new instances becomes unfeasible. 

In what follows, for a fixed topology incurring power demand changes, we compare the performance of our approach with the state-of-the-art Branching Process (in short, \textbf{BP}) method of~\cite{hines2016cascading}. 

A similar experiment involving another state-of-the-art method \cite{wu2021predicting} is described in the appendix.

\subsection{Distribution of Cascades}
We first compare the probability distribution and occurrence rate of CFs in the original dataset (standard IEEE300 with power demand factor $1$).

Fig.~\ref{fig_cas_pmf} shows the probability distribution of cascade sizes (number of failures) for (i) the cascades original CF data, (ii) the ones simulated by our hyperparametric diffusion model, (iii) the ones simulated by the BP approach. Note that the original data exhibits an abnormal distribution, beginning with many small size cascades, reaching a trough around size $10$, and following with several small peaks around size $20$. Regarding the simulated data, both models align well with the original data for small cascades, while for the moderate and larger cascade, our model shows better prediction.

Fig.~\ref{fig_occ_rate_1} compares the failure occurrence rate of lines (excluding initial failures) in the original cascade dataset (empirical) and in our hyperparametric diffusion model. Our prediction aligns well with the dataset in the range of occurrence rate larger than $10^{-3}$, and diverges moderately for the lower range, where are situated the small cascading failures, which are more difficult to be predicted.

\begin{figure}
    \captionsetup{skip=5pt}
    \captionsetup[subfigure]{skip=5pt}
    \centering
    \begin{subfigure}[t]{0.49\linewidth}\centering
        \includegraphics[width=\linewidth]{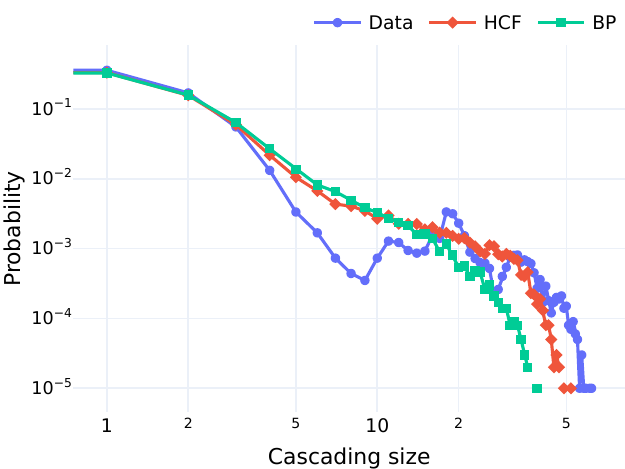}
        \caption{Probability distribution.}
        \label{fig_cas_pmf}
    \end{subfigure}
    \hfill
    \begin{subfigure}[t]{0.49\linewidth}\centering
        \includegraphics[width=\linewidth]{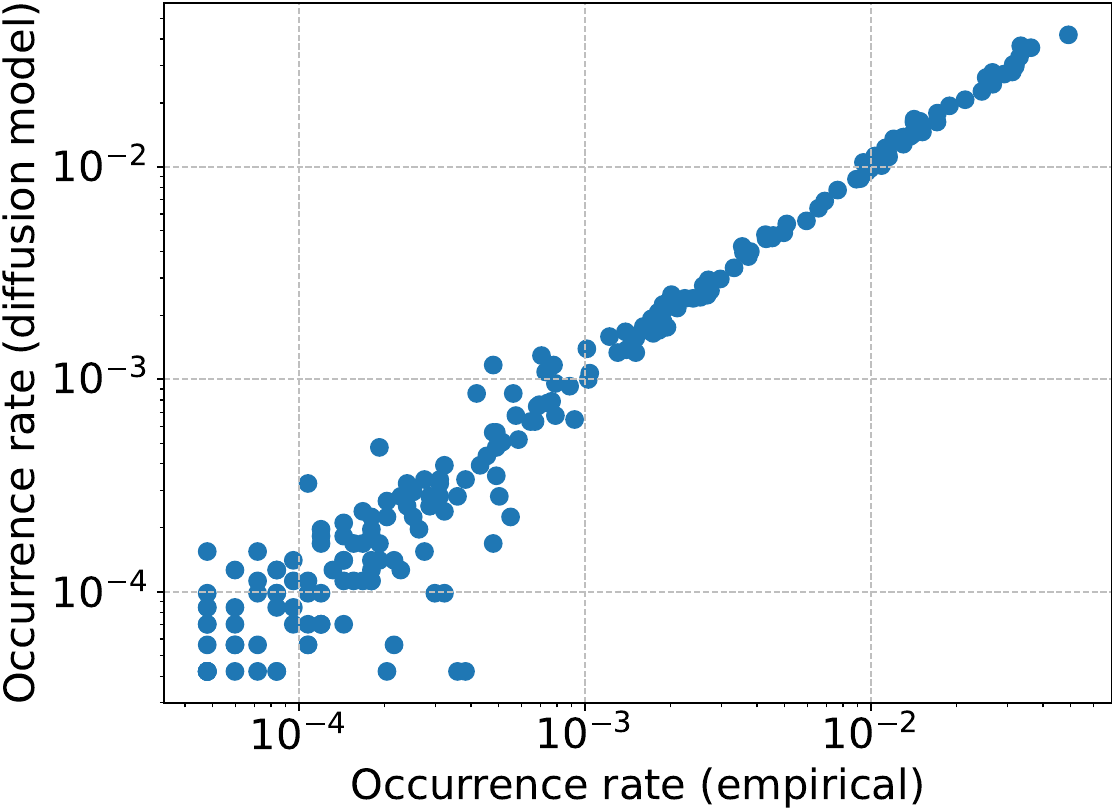}
        \caption{Failure occurrence rate.}
        \label{fig_occ_rate_1}
    \end{subfigure}
    \caption{Probability distribution of cascades and failure occurrence rate (excluding initial failures) for IEEE300.}
    \vspace{-3mm}
\end{figure}

\begin{figure*}[!t]
    \captionsetup{skip=5pt}
    \captionsetup[subfigure]{skip=5pt}
    \centering
    \begin{subfigure}[t]{0.24\linewidth}\centering
        \includegraphics[width=\textwidth]{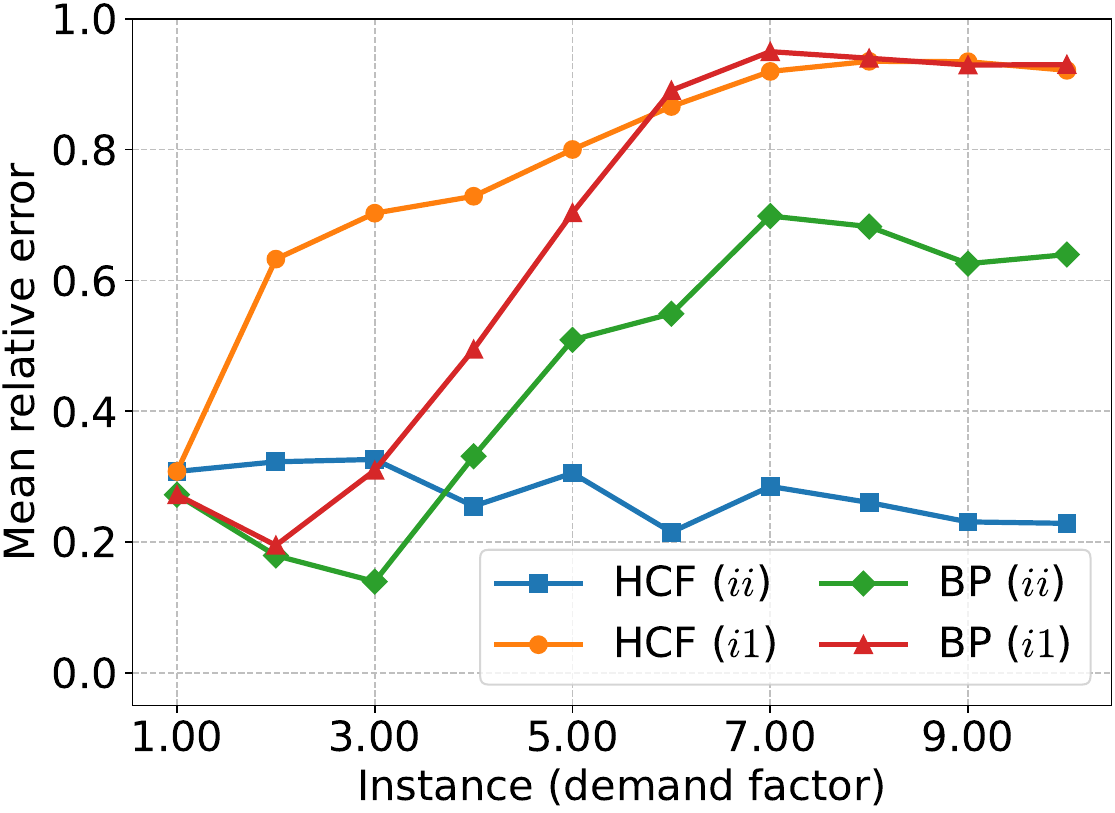}
        \caption{Failure distribution error}
        \label{fig_vcnt_err_r_dc_1.0}
    \end{subfigure}
    \hfill
    \begin{subfigure}[t]{0.24\linewidth}\centering
        \includegraphics[width=\textwidth]{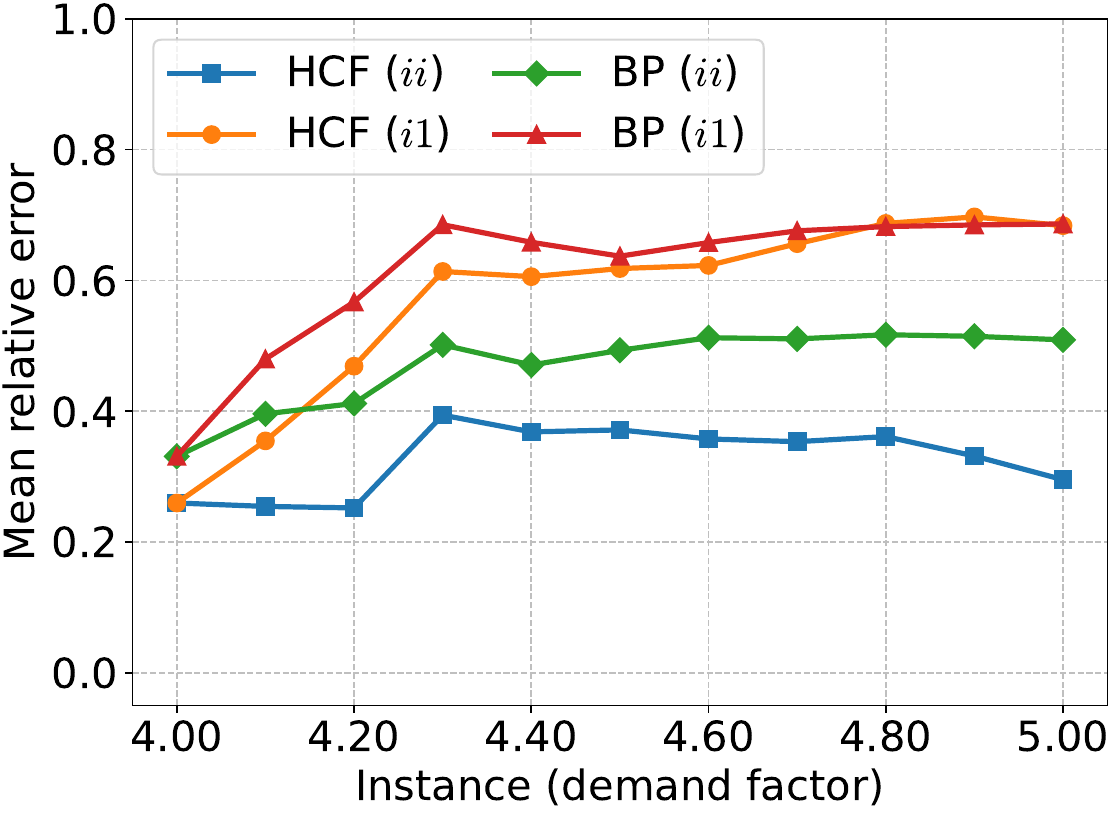}
        \caption{Failure distribution error}
        \label{fig_vcnt_err_r_dc_4.0}
    \end{subfigure}
    \hfill
    \begin{subfigure}[t]{0.24\linewidth}\centering
        \includegraphics[width=\textwidth]{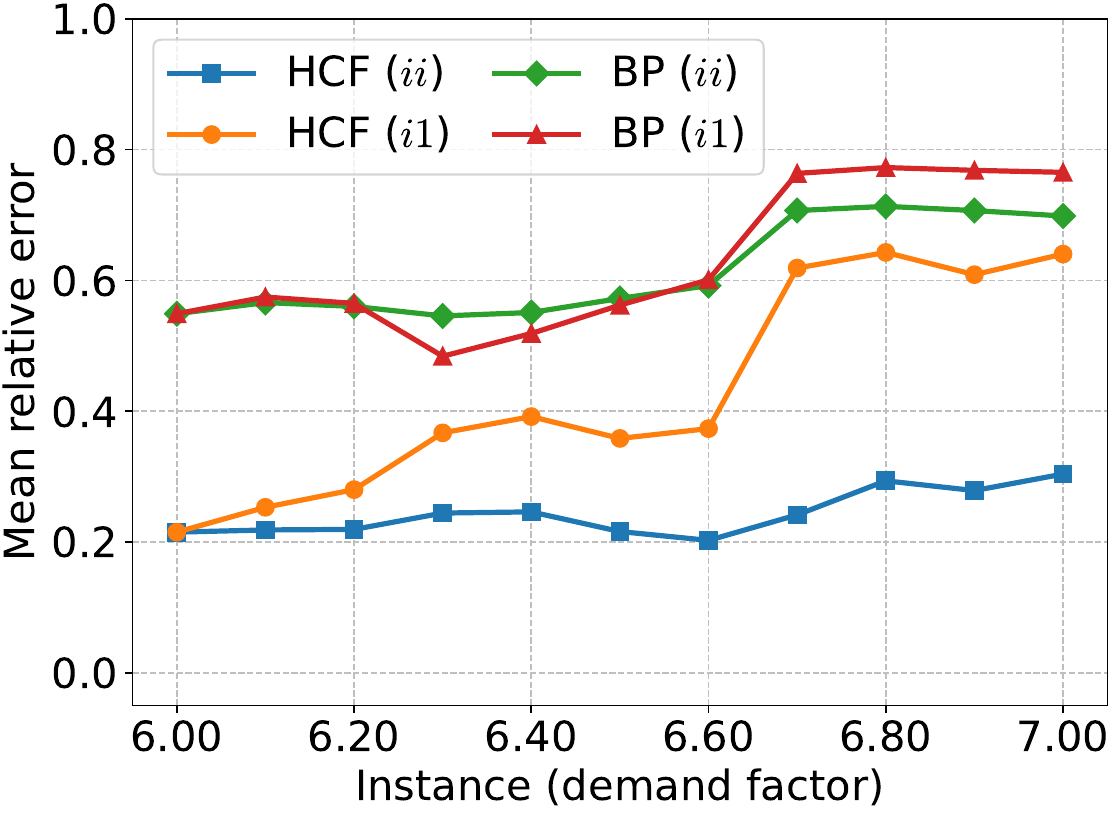}
        \caption{Failure distribution error}
        \label{fig_vcnt_err_r_dc_6.0}
    \end{subfigure}
    \hfill
    \begin{subfigure}[t]{0.24\linewidth}\centering
        \includegraphics[width=\textwidth]{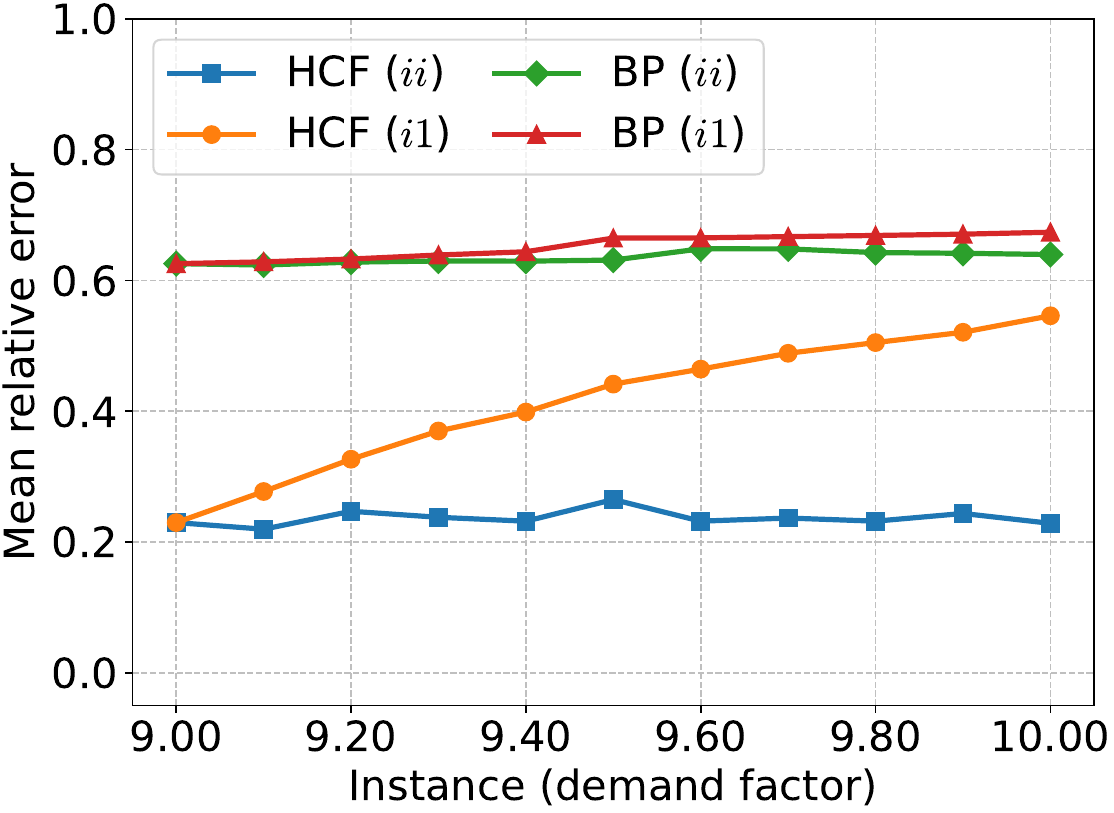}
        \caption{Failure distribution error}
        \label{fig_vcnt_err_r_dc_9.0}
    \end{subfigure}\\[1pt]
    \begin{subfigure}[t]{0.24\linewidth}\centering
        \includegraphics[width=\textwidth]{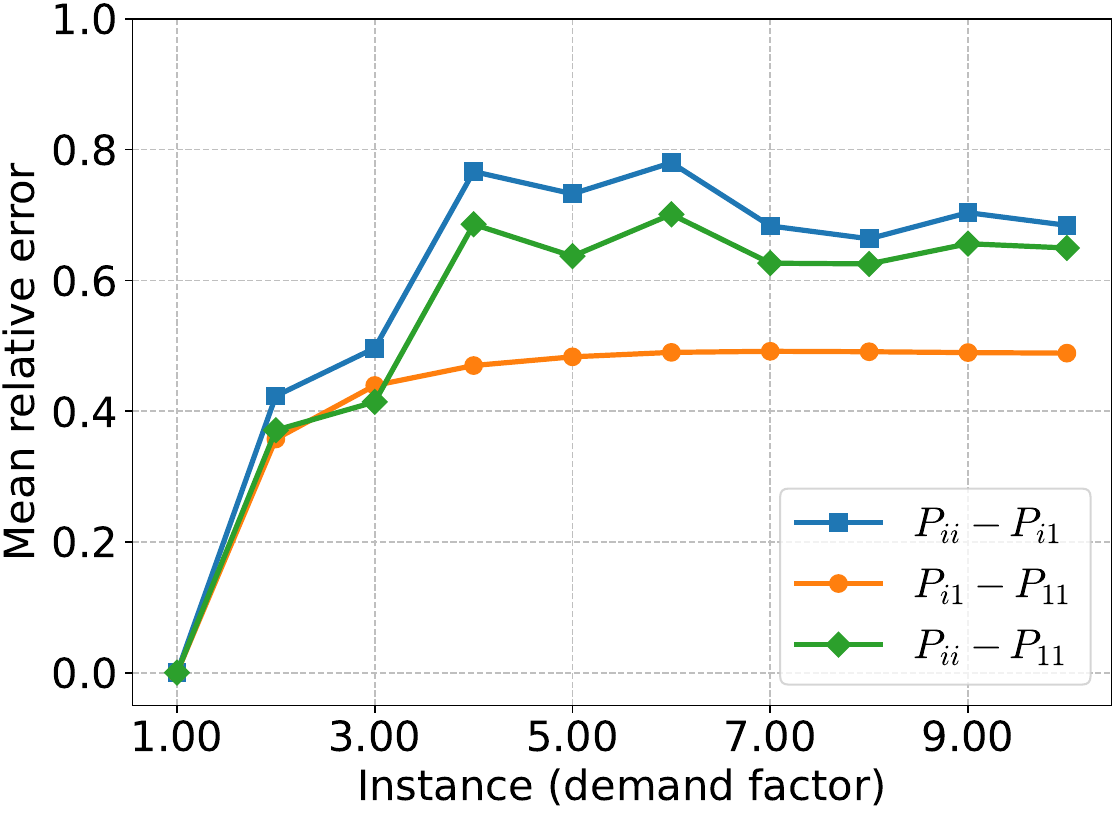}
        \caption{Probability error}
        \label{fig_pmat_err_r_dc_1.0}
    \end{subfigure}
    \hfill
    \begin{subfigure}[t]{0.24\linewidth}\centering
        \includegraphics[width=\textwidth]{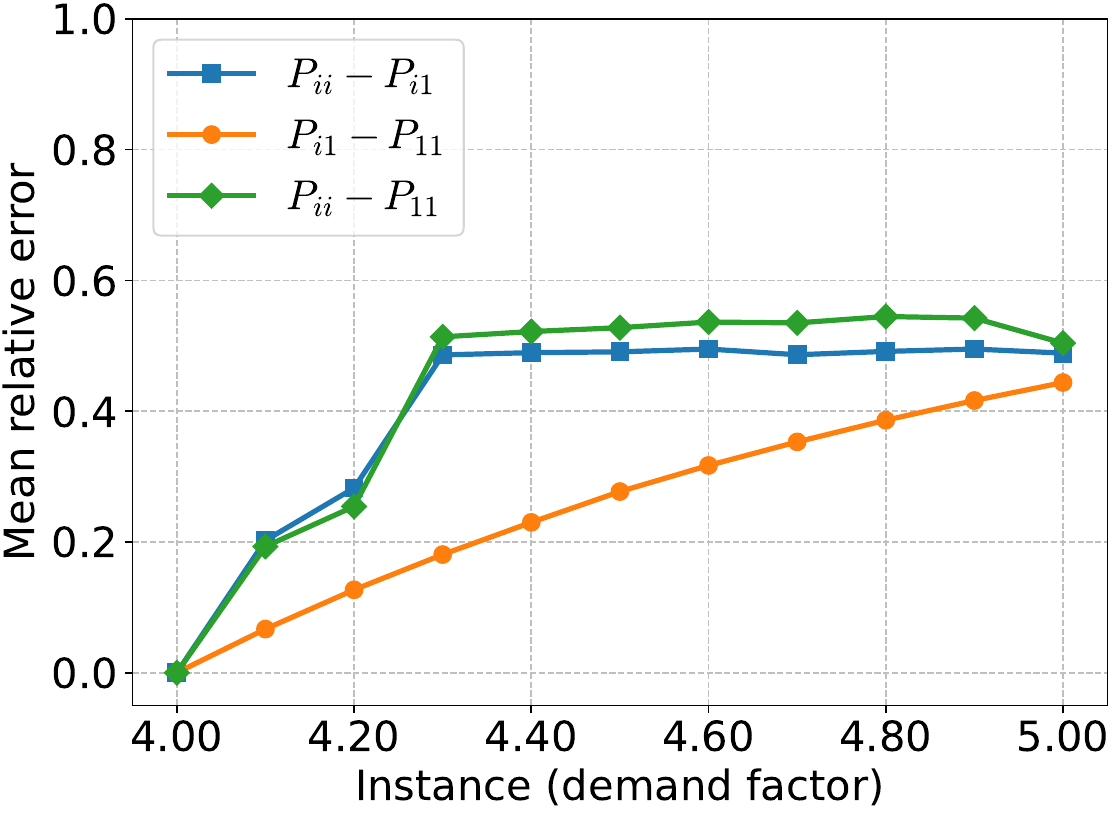}
        \caption{Probability error}
        \label{fig_pmat_err_r_dc_4.0}
    \end{subfigure}
    \hfill
    \begin{subfigure}[t]{0.24\linewidth}\centering
        \includegraphics[width=\textwidth]{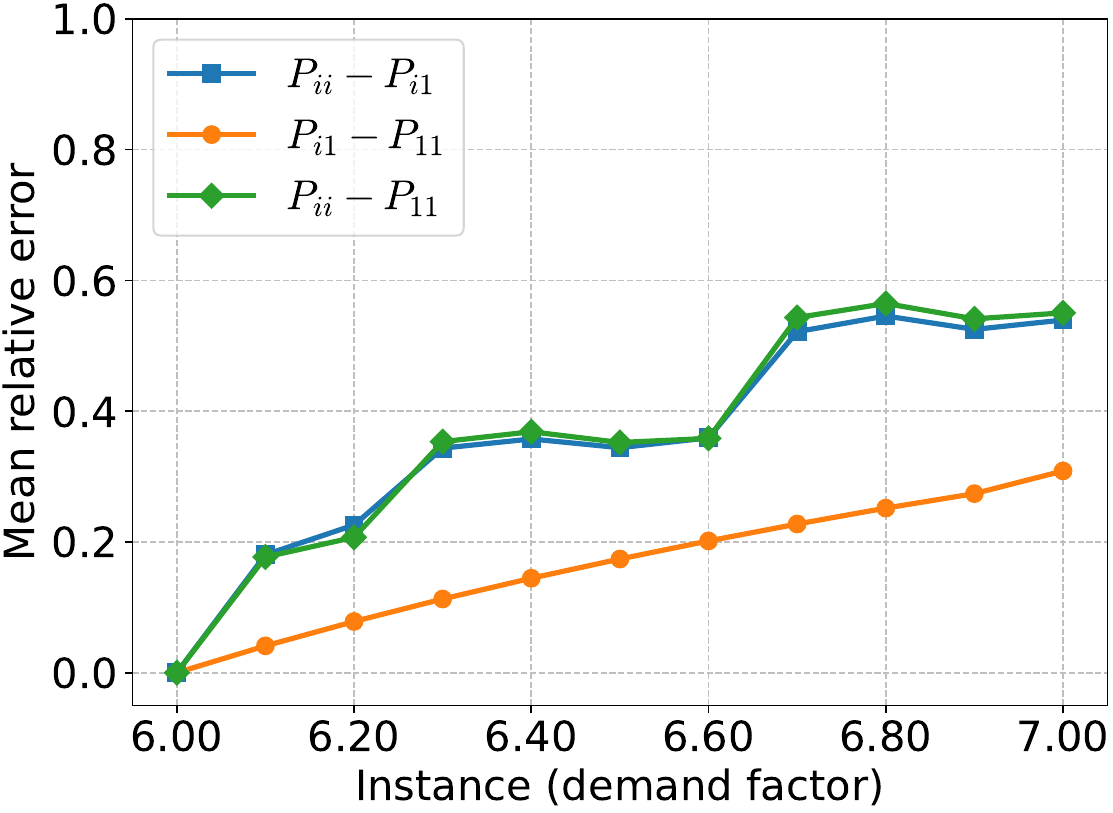}
        \caption{Probability error}
        \label{fig_pmat_err_r_dc_6.0}
    \end{subfigure}
    \hfill
    \begin{subfigure}[t]{0.24\linewidth}\centering
        \includegraphics[width=\textwidth]{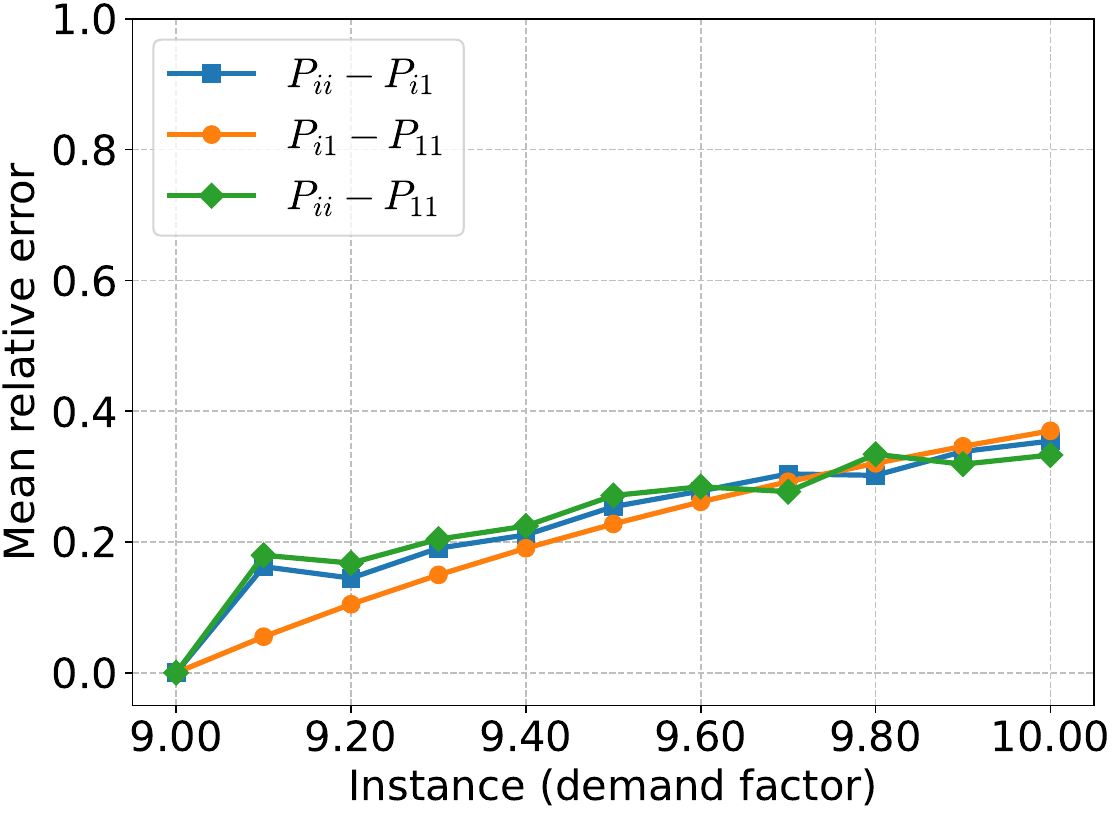}
        \caption{Probability error}
        \label{fig_pmat_err_r_dc_9.0}
    \end{subfigure}
    \caption{Mean relative error for line failure distribution and diffusion probability matrix (IEEE300 dataset).}
    \label{fig_err_r_dc_1.0}
\end{figure*}

\begin{figure*}[!t]
    \captionsetup{skip=5pt}
    \captionsetup[subfigure]{skip=5pt}
    \centering
    \begin{subfigure}[t]{0.19\linewidth}\centering
        \includegraphics[width=\textwidth]{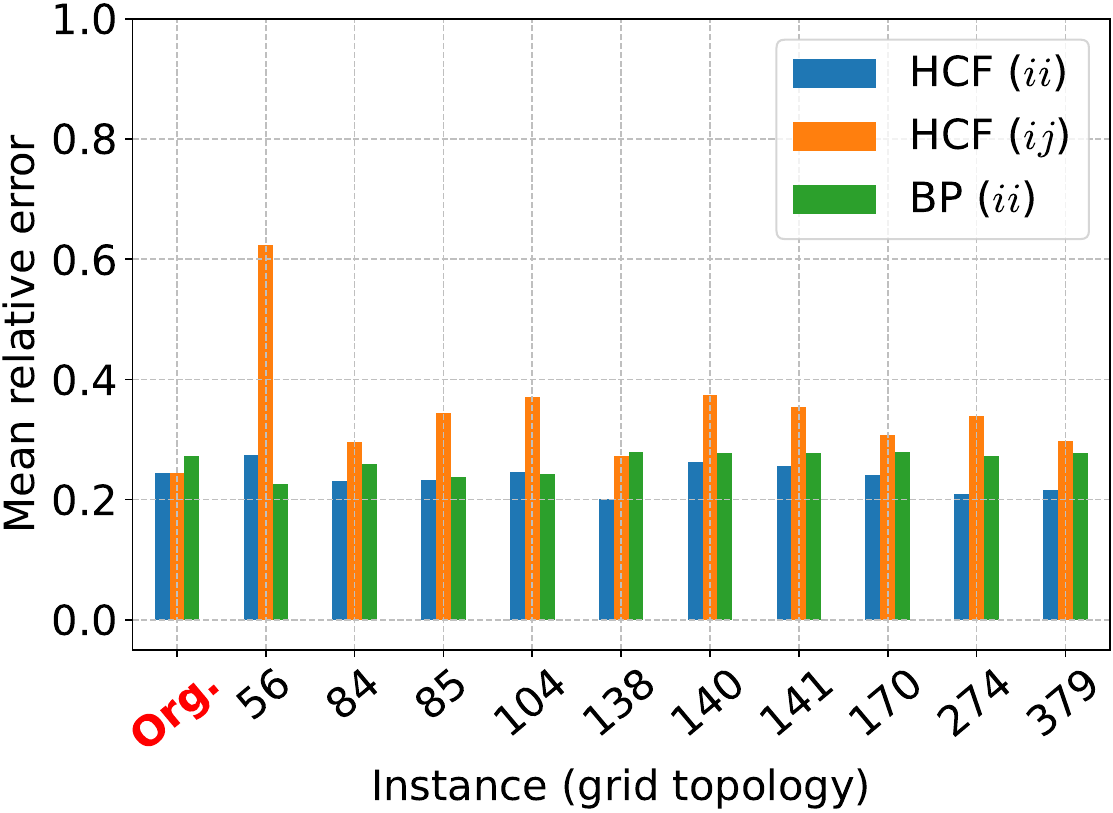}
        \caption{Original grid}
        \label{fig_vcnt_err_tc_0}
    \end{subfigure}
    \hfill
    \begin{subfigure}[t]{0.19\linewidth}\centering
        \includegraphics[width=\textwidth]{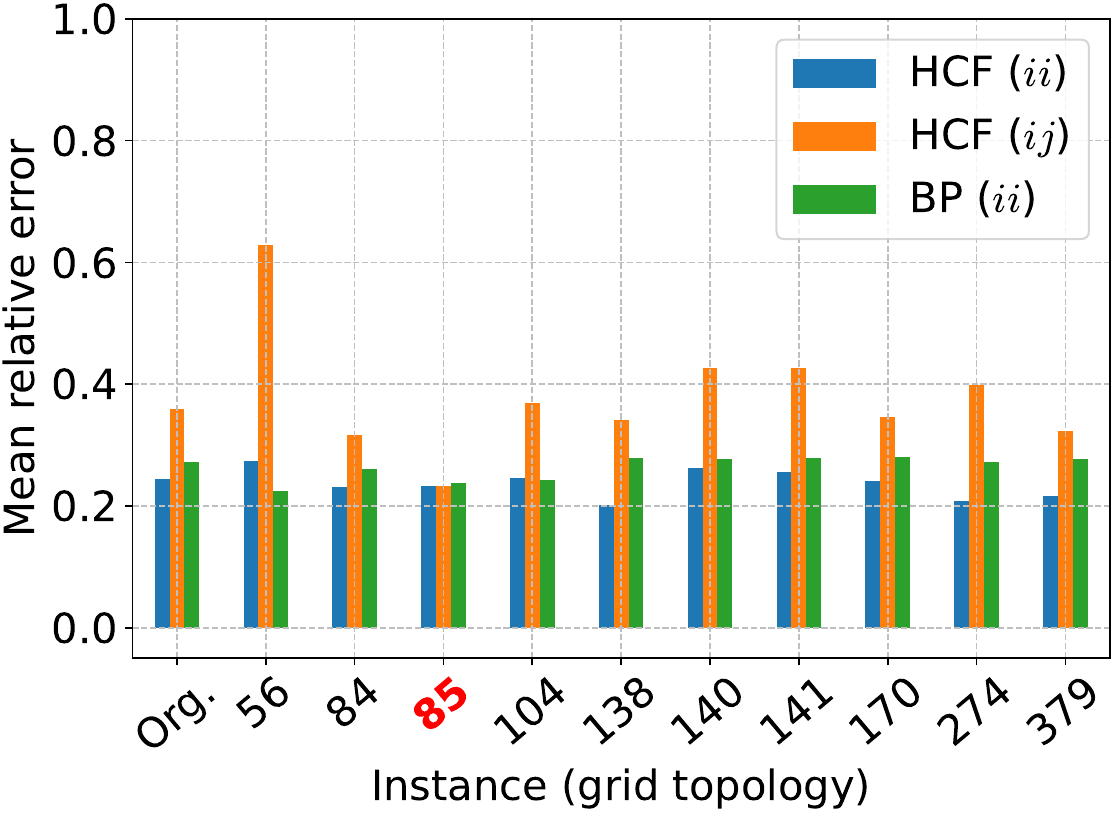}
        \caption{Grid (L-85 removed)}
        \label{fig_vcnt_err_tc_85}
    \end{subfigure}
    \hfill
    \begin{subfigure}[t]{0.19\linewidth}\centering
        \includegraphics[width=\textwidth]{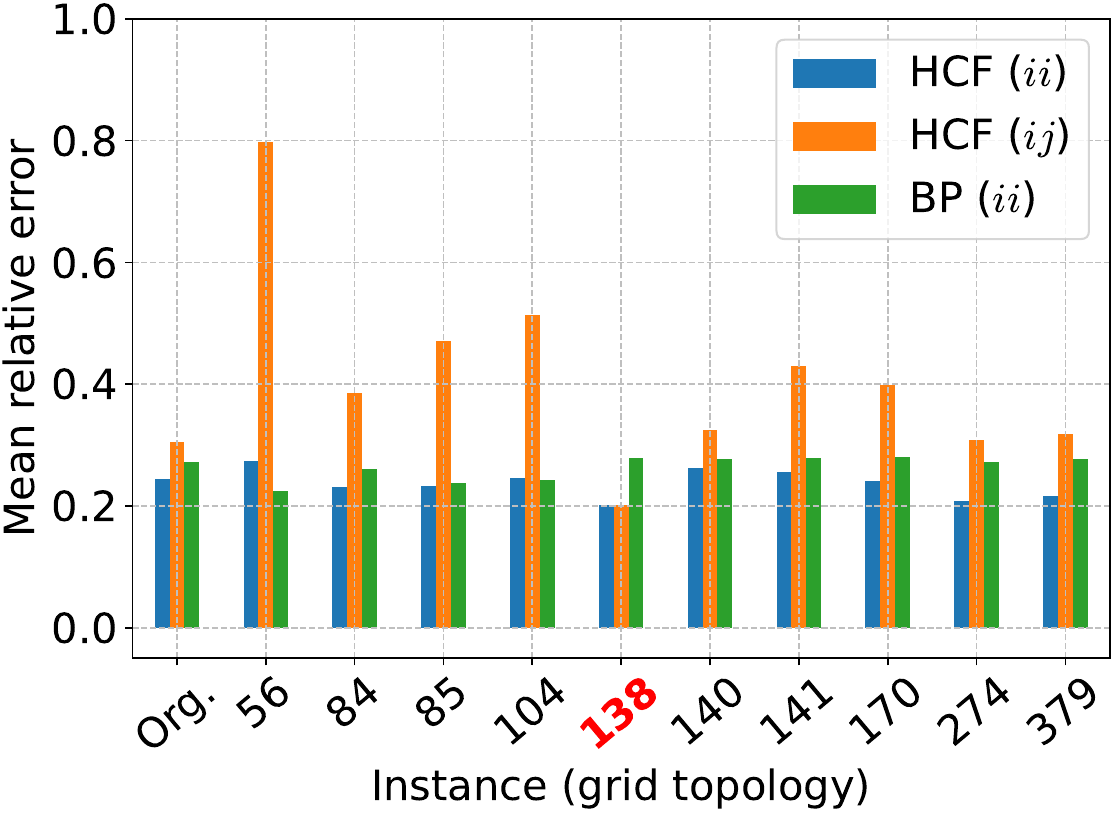}
        \caption{Grid (L-138 removed)}
        \label{fig_vcnt_err_tc_138}
    \end{subfigure}
    \hfill
    \begin{subfigure}[t]{0.19\linewidth}\centering
        \includegraphics[width=\textwidth]{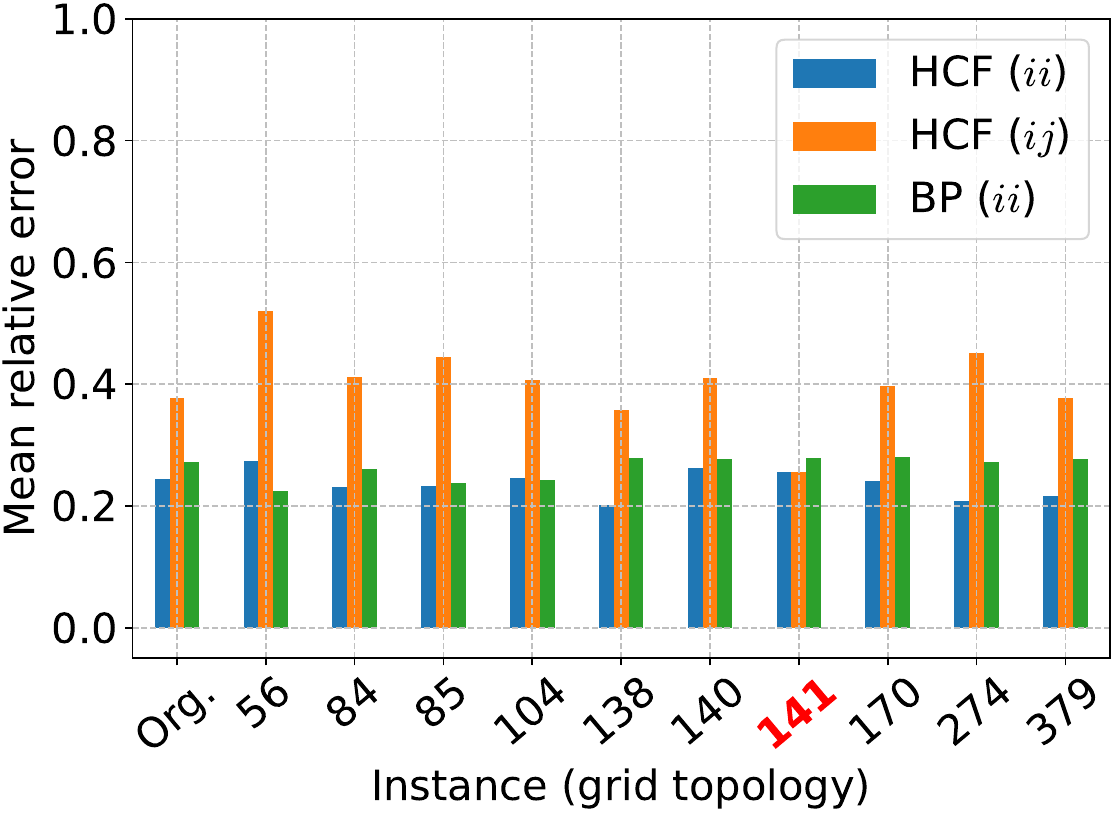}
        \caption{Grid (L-141 removed)}
        \label{fig_vcnt_err_tc_141}
    \end{subfigure}
    \hfill
    \begin{subfigure}[t]{0.19\linewidth}\centering
        \includegraphics[width=\textwidth]{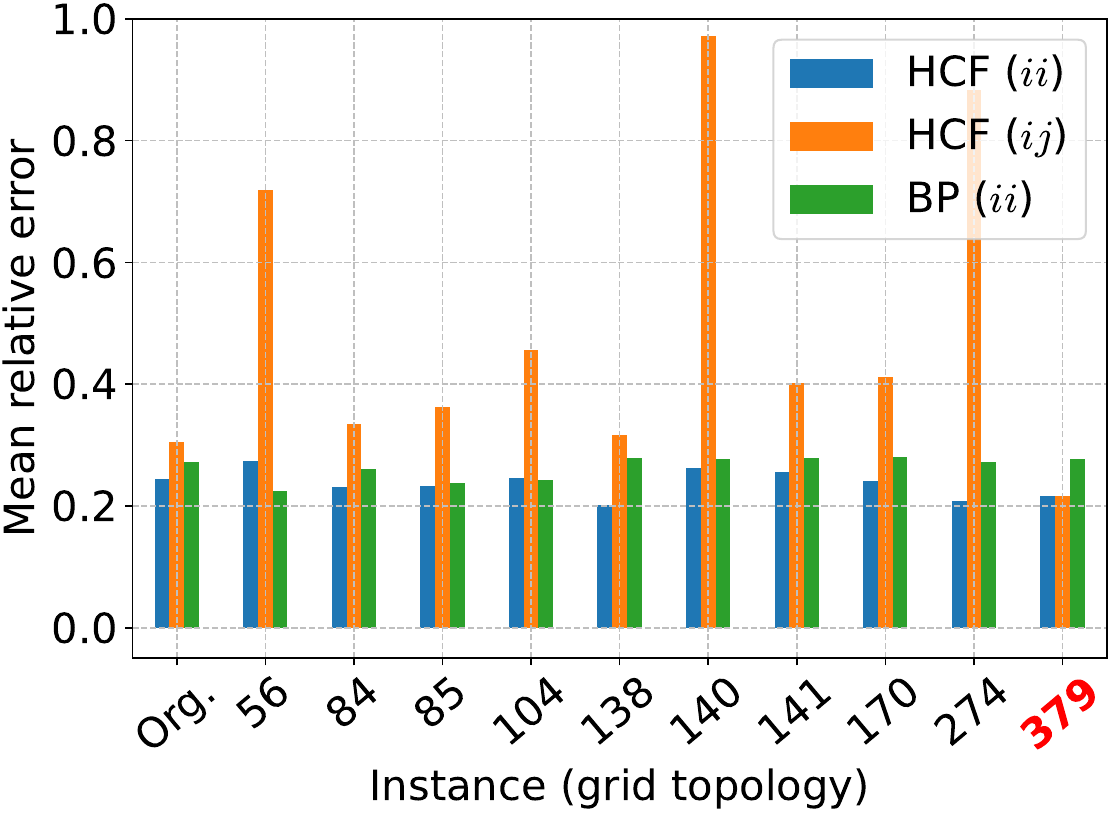}
        \caption{Grid (L-379 removed)}
        \label{fig_vcnt_err_tc_379}
    \end{subfigure} \\[1pt]
    \begin{subfigure}[t]{0.19\linewidth}\centering
        \includegraphics[width=\textwidth]{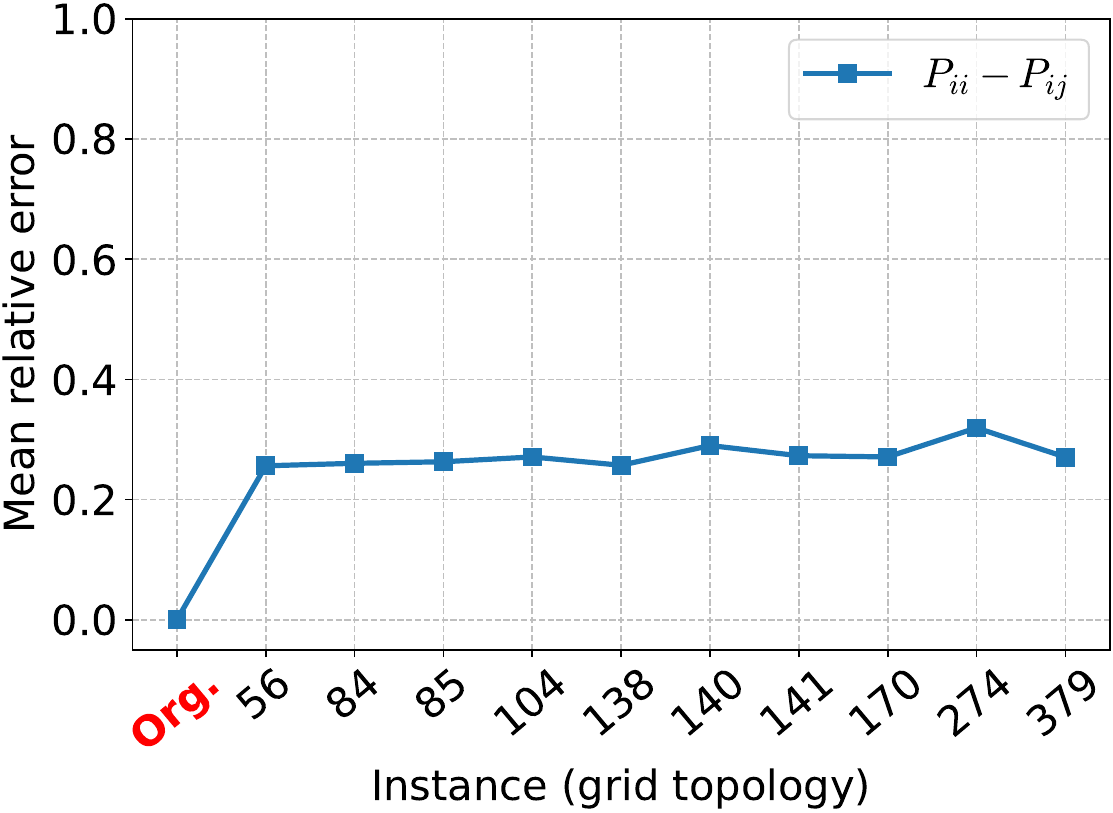}
        \caption{Original grid}
        \label{fig_pmat_err_tc_0}
    \end{subfigure}
    \hfill
    \begin{subfigure}[t]{0.19\linewidth}\centering
        \includegraphics[width=\textwidth]{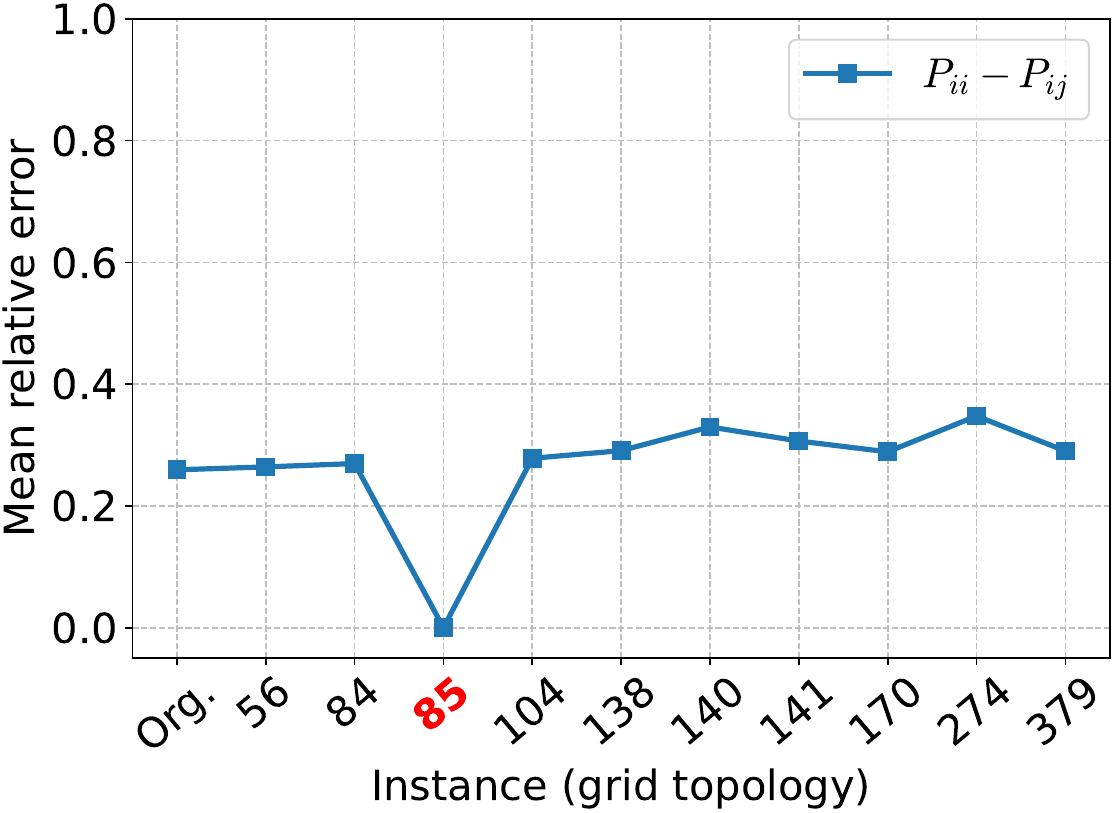}
        \caption{Grid (L-85 removed)}
        \label{fig_pmat_err_tc_85}
    \end{subfigure}
    \hfill
    \begin{subfigure}[t]{0.19\linewidth}\centering
        \includegraphics[width=\textwidth]{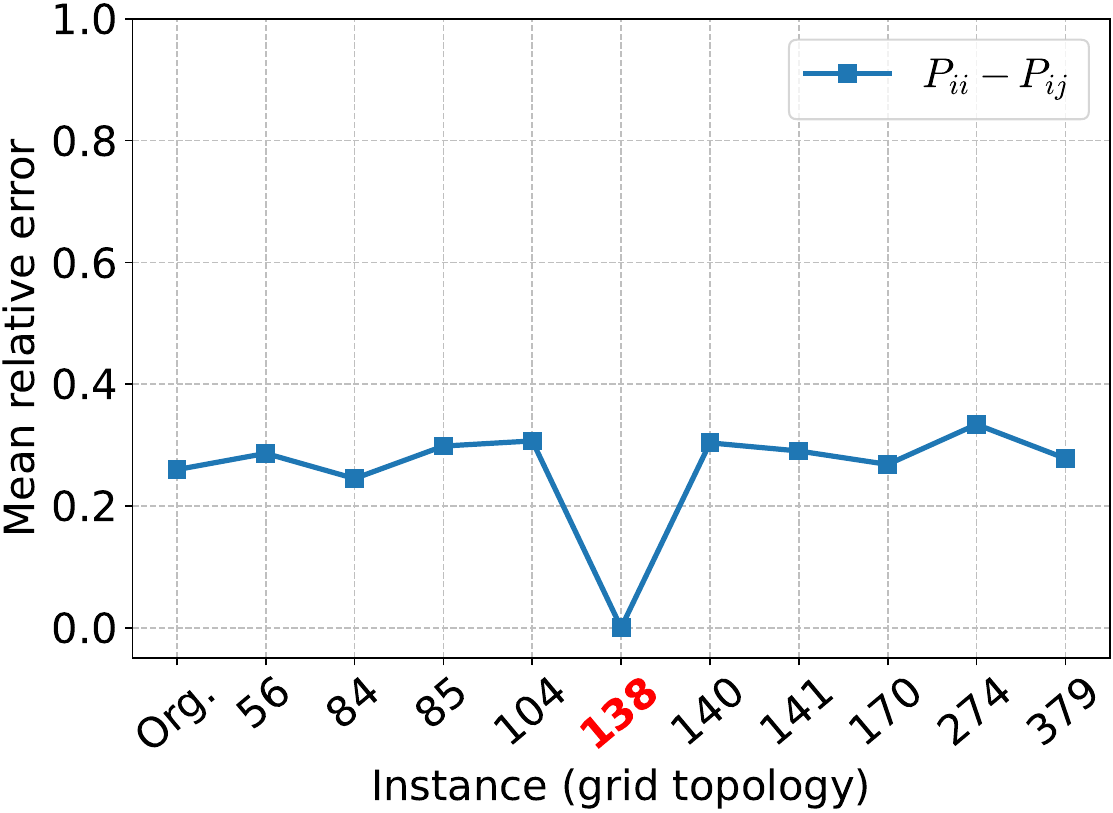}
        \caption{Grid (L-138 removed)}
        \label{fig_pmat_err_tc_138}
    \end{subfigure}
    \hfill
    \begin{subfigure}[t]{0.19\linewidth}\centering
        \includegraphics[width=\textwidth]{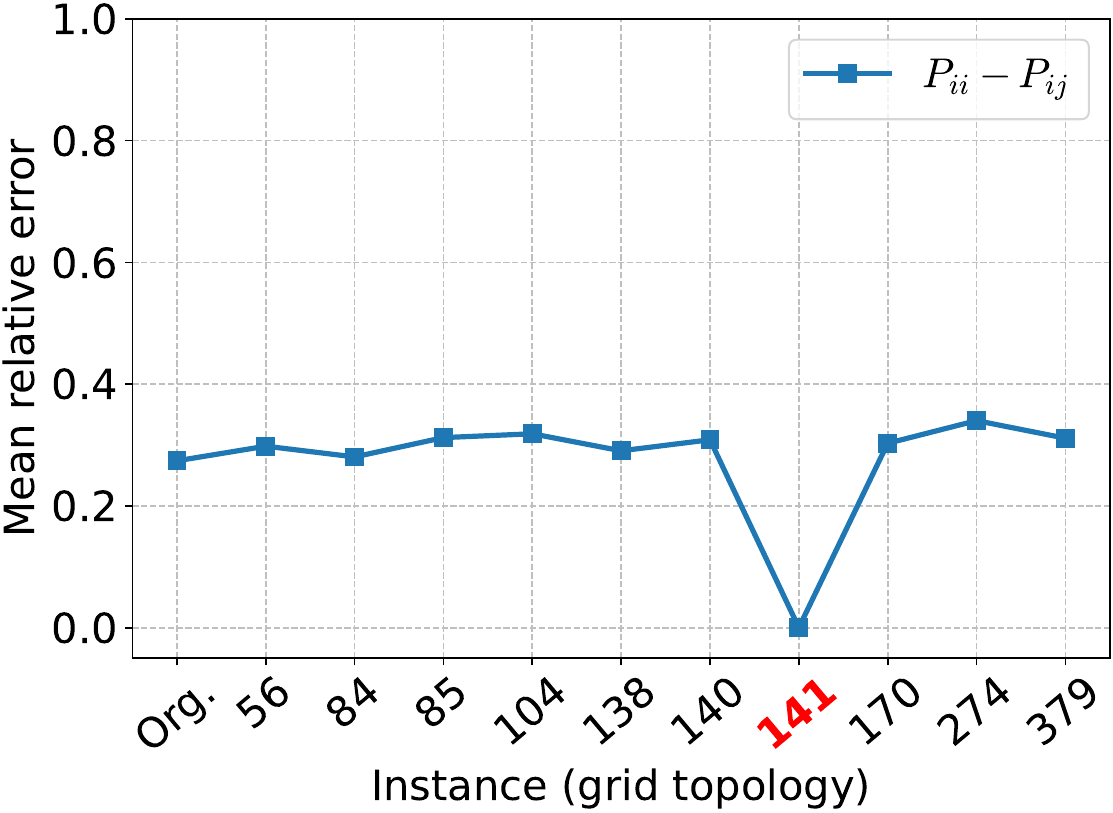}
        \caption{Grid (L-141 removed)}
        \label{fig_pmat_err_tc_141}
    \end{subfigure}
    \hfill
    \begin{subfigure}[t]{0.19\linewidth}\centering
        \includegraphics[width=\textwidth]{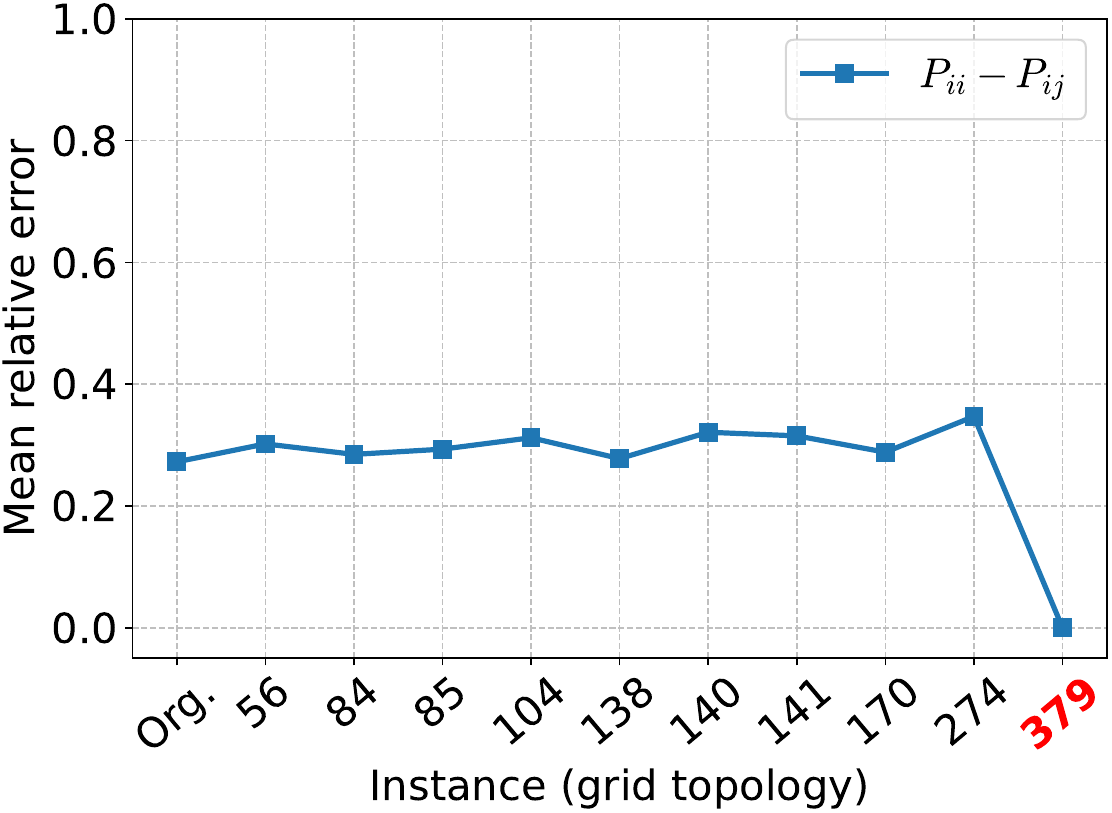}
        \caption{Grid (L-379 removed)}
        \label{fig_pmat_err_tc_379}
    \end{subfigure}
    \caption{Training the model with CF data from one grid, testing on the others (IEEE300 dataset).}
    \label{fig_err_tc}
    \vspace{-1pt}
\end{figure*}

\subsection{Model Generalization}

We study in this section the generalization ability of our model in the following directions. First, in Sec.~\ref{sec:change_power_demand},  given the original, standard IEEE300 power grid, and the CFs generated by the DCSepSim simulator, we increase the user power demands of a specific region in the grid to different levels -- moderate, heavy, or severe -- and regenerate corresponding CFs. Note that the increasing factors remain in a range that ensures the resulting power system is stable, i.e., works normally with a balanced power flow.  Second, in Sec.~\ref{sec:change_topology},  from the CFs of the original grid, we pick the $10$ most frequent lines therein, and remove each of them, one at a time, from the original grid, to generate 10 new grids having a slightly different physical topology; again, all these have a balanced power flow. The rationale is that these lines can be seen as highly critical.
We use a similar subscript notation as before, in the context of each figure: ${ii}$ for a model learned from the instance  $i$ and tested on that same instance, 
and ${ij}$ for a model learned on instance $j$ in that figure and tested on instance $i$. 

\subsubsection{Power demand changes}
\label{sec:change_power_demand}
In Fig.~\ref{fig_err_r_dc_1.0}, the power demands of the original power grid instance in the region $<20\%$ \emph{quantile} are scaled from $1$ to $10$ to produce $10$ different grid instances. When the increasing factor is higher than $10$, the power grid will not work normally. Therefore, the settings we selected provide a coarse-grained full-scope examination of performance by power demand. The plots show the mean relative DEs and PEs of all lines. 

\eat{
First, in the leftmost plot,  Fig.~\ref{fig_vcnt_err_r_dc_1.0}, we can observe mainly three regions of decrease in predictive accuracy, i.e., $[4,5]$, $[6,7]$ and $[9,10]$ (see Appendix~\ref{ap_add_exp} for the corresponding plots on absolute error). This motivated us to further investigate  these power settings regions, to which we refer as moderate, heavy, and severe changes. We re-evaluate the models by shifting the basis training point to $4.0$, $6.0$ and $9.0$, and testing on the other instances in the corresponding ranges. The results for the distribution errors are provided in Figs.~\ref{fig_vcnt_err_r_dc_1.0},~\ref{fig_vcnt_err_r_dc_4.0},~\ref{fig_vcnt_err_r_dc_6.0}, and~\ref{fig_vcnt_err_r_dc_9.0} respectively, and for the corresponding probability errors, they are in Figs.~\ref{fig_pmat_err_r_dc_1.0},~\ref{fig_pmat_err_r_dc_4.0},~\ref{fig_pmat_err_r_dc_6.0}, and~\ref{fig_pmat_err_r_dc_9.0} respectively.
}

Figs.~\ref{fig_vcnt_err_r_dc_1.0} and~\ref{fig_pmat_err_r_dc_1.0} show respectively the mean relative DEs and PEs for the full-scope demand change ($[1,10]$, with basis instance $1.0$). First, in Fig.~\ref{fig_vcnt_err_r_dc_1.0}, when the demand factor $<5.0$, our testing results show large relative errors; however,  the corresponding absolute errors in this region for all approaches are quite small (see Appendix~\ref{ap_add_exp}) and the relative error in this case is less indicative of performance as the values are relatively small. For the retraining results, the performance is quite constant.  Beyond the $<5.0$ region,  the testing results are very similar, but the errors become quite large. This is to be expected, since the power grid has experienced large power demand changes. 
 Indeed, it is challenging for the model that is trained on the initial instance ($1.00$) to capture a very ``distant'' instance such as $10.0$. As for the retraining results, when the demand factor is in the region $>5.0$, our model achieves lower relative errors (around $0.3$) than BP, with a gap around $0.35$. In Fig.~\ref{fig_pmat_err_r_dc_1.0} -- mean relative PEs -- we can notice that $P_{ii}-P_{i1}$, which is an indicator of our model generalization capability, follows a similar trend as with the DE values. Recall $P_{ii}-P_{11}$ indicates the changes in predicted probabilities when retraining on a different instance. We can notice that  the cascading behavior of the grid instances indeed experience significant changes when the demand factor is in the region $<5.0$. The model learned on the initial instance becomes less applicable as the power demand changes.

Regarding the \emph{moderate} demand changes (in the range $[4,5]$), zoomed upon in Fig.~\ref{fig_vcnt_err_r_dc_4.0} and~\ref{fig_pmat_err_r_dc_4.0}, the basis instance for training is moved to $4.0$. The results show a different pattern compared with the previous full-scope case. Now, prediction   on other instances becomes more accurate, as the  power changes are less drastic. 
In the region $>4.2$, all approaches have an error around $0.6$, and show a stable performance. For the retraining ($ii$) results, our model has relative errors around $0.2\sim 0.3$, and performs better than BP, by a margin of around $0.1$ (so around $50\%$). For the generalization testing results ($i1$), our model shows similar performance as BP (but outperforms it, by a margin of  around $18\%$, in terms of mean absolute error, see Appendix~\ref{ap_add_exp}). For the PEs, in Fig.~\ref{fig_pmat_err_r_dc_4.0}, the pattern is also rather different from the full-scope case. The results are aligned with the corresponding failure distribution errors. One conclusion we can draw here is that in practice it may be preferable to retrain in different regions in the spectrum of power demand changes, in order to have more robust / generalizable predictions. 
 
In the region of  \emph{heavy} power demand changes (range $[6,7]$), zoomed upon  in Figs.~\ref{fig_vcnt_err_r_dc_6.0} and~\ref{fig_pmat_err_r_dc_6.0}, the general trend is again quite different from the moderate case, showing a wave or step pattern. 
This is most likely indicative of the complex behaviors of CFs as the demand is changing. Regarding the DE metric, after retraining at $6.0$, all models have a longer applicable range from $6.0$ to $6.6$. In comparison, our model outperforms BP by a $0.2$ margin in average, when testing generalizability ($i1$), and a  $0.4$ margin when retraining ($ii$). Beyond point $6.6$, there is again a large increase of errors. Nevertheless, both our testing and retraining errors are lower than the ones of BP.  
The PE values shown in Fig.~\ref{fig_pmat_err_r_dc_6.0} follow a similar step pattern as the DE ones in  Fig.~\ref{fig_vcnt_err_r_dc_6.0}.

Finally, moving to the \emph{severe} demand changes (range $[9,10]$), shown in Figs.~\ref{fig_vcnt_err_r_dc_9.0} and~\ref{fig_pmat_err_r_dc_9.0}, in contrast to the \emph{heavy} case, the error trends for all approaches become smoother. For Fig.~\ref{fig_vcnt_err_r_dc_9.0}, the DEs present a Z-shape pattern. Our retraining ($ii$) results are stable around $0.28$, while the testing ones ($i1$) increase smoothly up to around $0.60$. Both the absolute and relative PEs of the BP model stay around $60\%$.
Our model outperforms BP with a margin of around $0.35$ when retraining, resp.  around $0.23$ when testing directly. 
The PEs in Fig.~\ref{fig_pmat_err_r_dc_9.0} remain lower than $0.4$ and lower than those observed in the previous three cases.
We believe that in this severe case, as the demands  increase towards the limit, the CF patterns converge as well,  leading to more stable performance.

One important conclusion from these three cases -- moderate, heavy, or severe power demand changes -- is that the distribution of CFs changes significantly; this leads to settings where the basic assumption of Poisson offspring distribution (on which BP relies) no longer holds, causing BP's decrease in performance. In contrast, since our hyperparametric diffusion model learns to predict the diffusion probabilities, and by extension the CFs, based on the physical and topological features, it can capture the underlying  physics dynamics better than BP's non-parametric approach.

\subsubsection{Topology changes}
\label{sec:change_topology}
When changing the topology of the power grid, by removing one transmission line, the state-of-the-art approaches cannot be applied without retraining (no $ij$ cross train-test experiment is possible). Therefore, only their retraining results ($ii$) can be obtained.  In Fig.~\ref{fig_err_tc}, we start from the original power grid (marked as ``Org.''), each time removing one line (shown in x-axis) and generate $10$ different power grids. In the figures, the value marked in red refers to the training instance. 
In terms of retraining results, our model slightly outperforms BP, with error around $0.2$-$0.3$.  For testing results ($ij$), our model has a general stable performance around $0.3$ on an average. There are a few peaks appearing for instances $56$, $140$, and $274$:  when the model is trained based  on instance ``Org.'', $85$, $138$, or $379$, the testing results on instance $56$ show a larger relative error, approximately in the range $0.6$-$0.8$.
When the model is trained on instance $379$, the test results show a large error on instances $140$ and $274$.

\begin{figure*}[t]
    \captionsetup{skip=3pt}
    \captionsetup[subfigure]{skip=3pt}
    \centering
    \begin{subfigure}[t]{0.33\linewidth}\centering
        \includegraphics[width=0.9\linewidth]{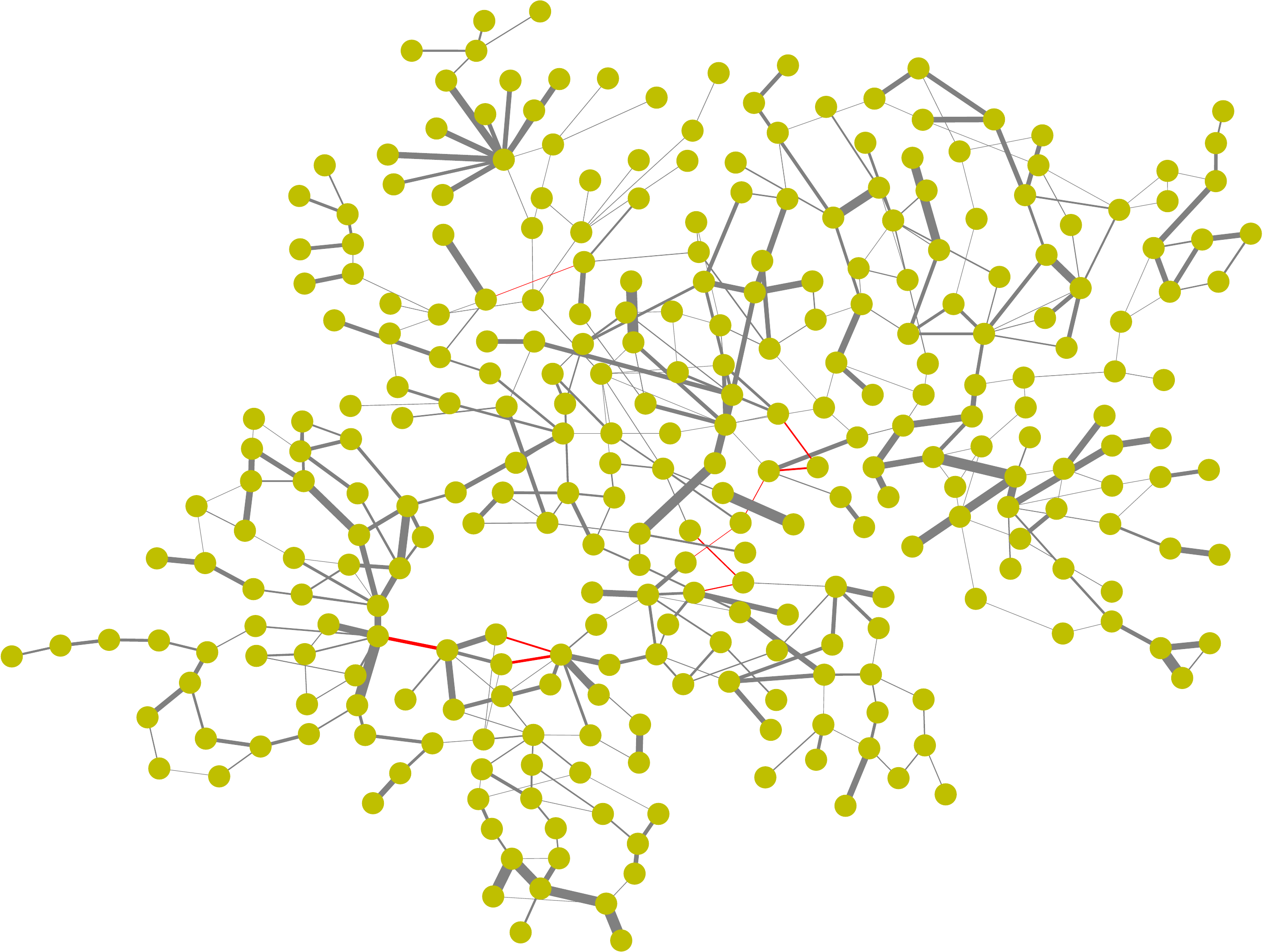}
        \caption{Physical graph.}
        \label{fig_phy_topo}
    \end{subfigure}
    \hfill
    \begin{subfigure}[t]{0.33\linewidth}\centering
        \includegraphics[width=0.9\linewidth]{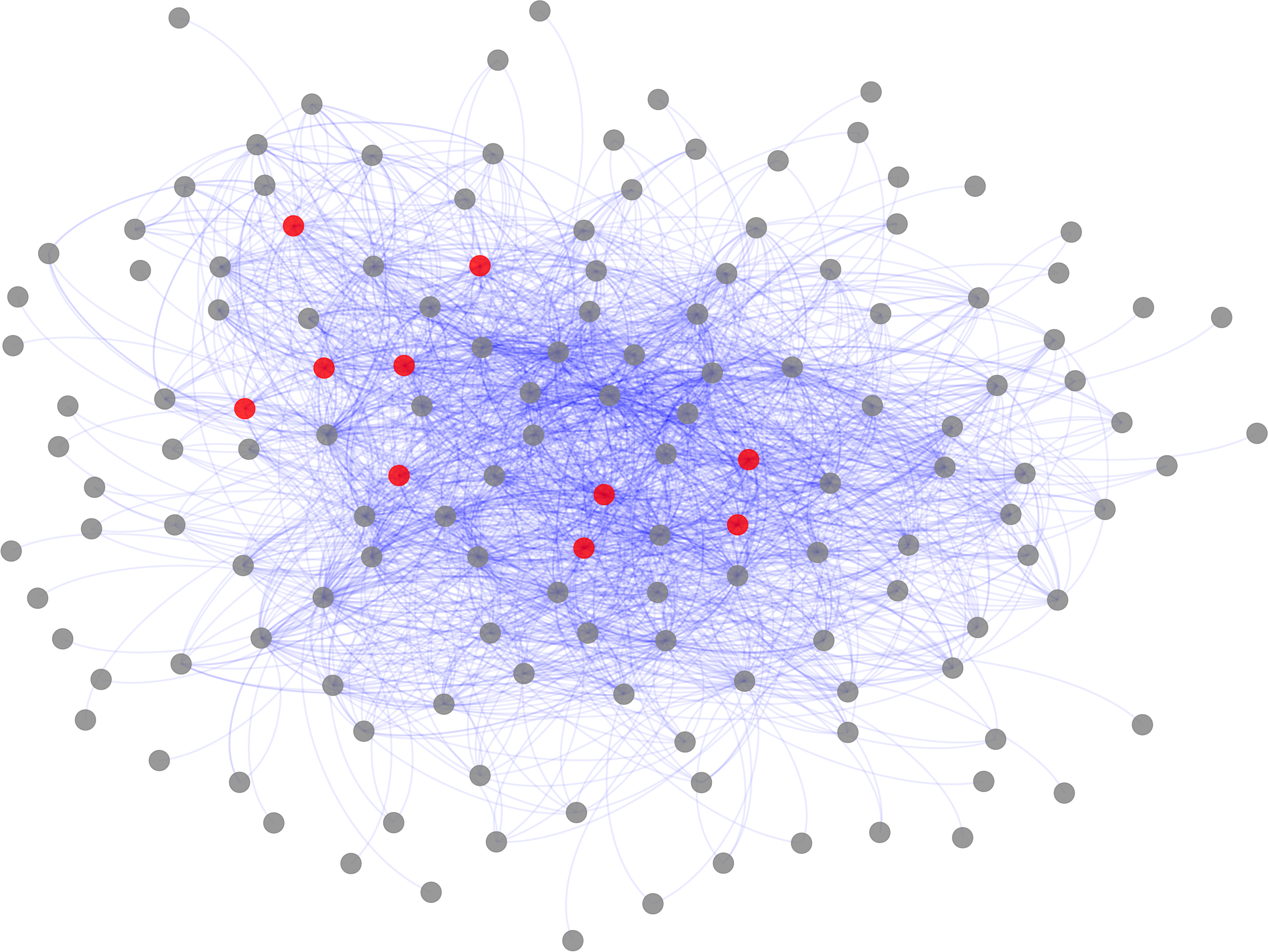}
        \caption{Diffusion graph.}
        \label{fig_diffusion}
    \end{subfigure}
    \hfill
    \begin{subfigure}[t]{0.33\linewidth}\centering
        \includegraphics[width=0.9\linewidth]{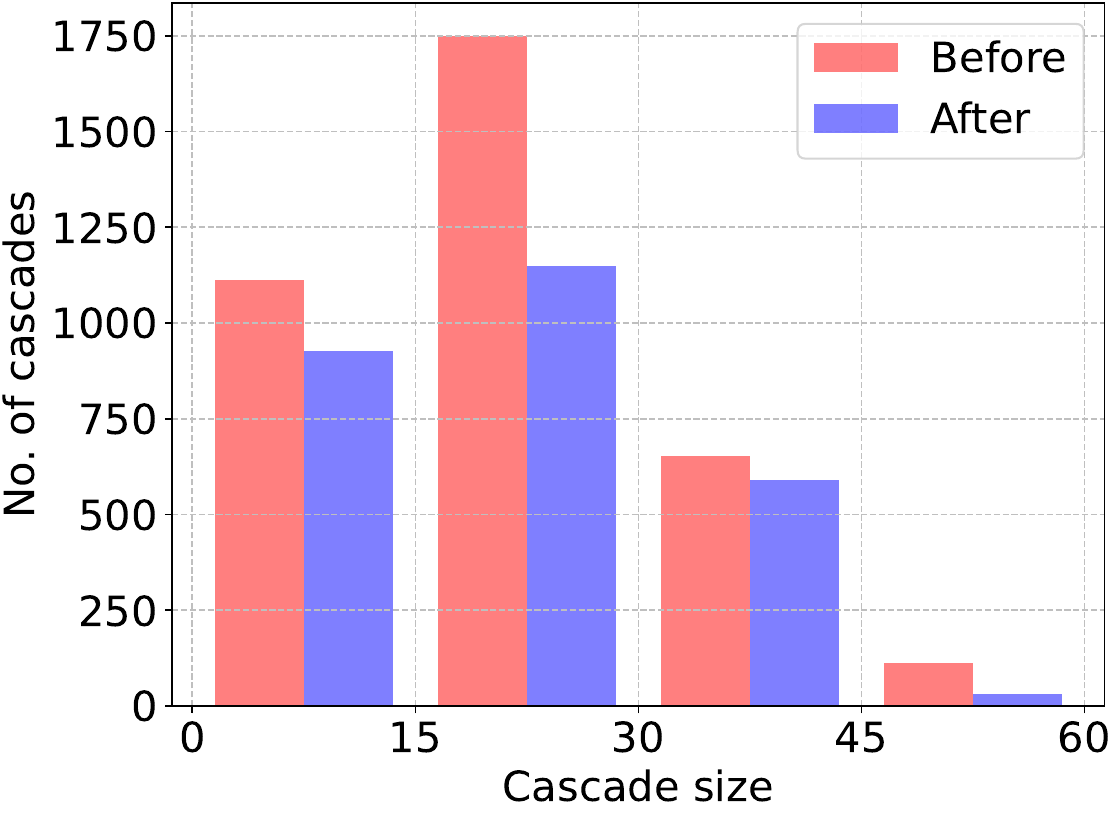}
        \caption{Cascading failures reduction.}
        \label{fig_cas_size_ba0}
    \end{subfigure}
    \caption{Physical graph of power grid, and the diffusion graph learned from the hyperparametric model, filtered with the probability value $p_{uv}\geqslant 0.01$. Cascade failures before and after increasing the capacities of ten most critical branches (IEEE300 dataset).}
    \vspace{-2mm}
\end{figure*}

\vspace{-3pt} 
\subsection{Mitigating CFs}
In this section, based on the diffusion probability matrix  we learned, we consider an initial attempt to mitigate the risks of CFs, by increasing the capacity of certain nodes in our diffusion graph (i.e., transmission lines in the original power grid). 
 In order to select the lines to be consolidated, we apply the  traditional IM algorithm CELF \cite{DBLP:conf/kdd/LeskovecKGFVG07} to retrieve the top $10$ seeds, corresponding in some sense to the $10$ most critical lines, their failure leading to the largest expected cascade. We  doubled the capacity of these lines, and re-simulate the CFs in order to observe to  which extent the CFs were mitigated. 
 
 Figures~\ref{fig_phy_topo} and~\ref{fig_diffusion} illustrate respectively the physical grid and the diffusion graph, along with the $10$ selected transmission lines.  They can be  viewed as ``dual'' graphs, where the edges in Fig.~\ref{fig_phy_topo} (marked in gray) correspond to  nodes (marked in red) in the diffusion graph. Recall that, in the power grid, nodes represent load or generator buses. We can observe that the grid is sparse, while the diffusion graph is rather dense (conceptually complete); in our illustration, we filtered out the edges with $p_{uv}\leq 0.01$.

We compare the distribution of cascade sizes, before and after doubling the capacity of the selected $10$ seed nodes, in Fig.~\ref{fig_cas_size_ba0}. There is a clear effect of reduction in CFs  for all size ranges (i.e., $[0,15]$, $[15,30]$, $[30,45]$, $[45,60]$, as the maximum cascade for this power grid is of size $60$). For each of these ranges, CFs are respectively decreased by around $18\%$, $36\%$, $9\%$, and $75\%$. We interpret this result as a promising one, and we intend to explore further options for reducing the potential spread in the resulting diffusion graph.

\subsection{Running Time}

Running time is mainly influenced by two factors: the specific grid instance and the size of the diffusion samples used for training.
Table~\ref{tab_time} shows both the training and testing time for 3 instances with power demand changes on the IEEE300 grid.
It is important to note that training for any method in this context is unlikely to be real-time.  \algoname's training is much longer compared with BP, when running on L. However, once our model is learned, it can be directly applied when the grid changes; the other approaches need to retrain the model on a new dataset, which requires a long time to collect (or, in some cases, a new dataset might not even be available). The ability to function effectively in new grid configurations is crucial, as decisions about grid provisioning are time-sensitive. This highlights once more the importance of model adaptation to unseen grid configurations.

\begin{table}[!th]
\center
\captionsetup{skip=3pt}
\caption{Training and testing time (seconds) per instance for the IEEE300 grid on laptop (L) and server (S).}
\label{tab_time}
\begin{tabularx}{\linewidth}{c|Y|Y|Y|Y|Y|Y}
\toprule
& \multicolumn{2}{c|}{1.0} & \multicolumn{2}{c|}{5.0} & \multicolumn{2}{c}{10.0} \\
\cmidrule(l{2pt}r{2pt}){2-3} \cmidrule(l{2pt}r{2pt}){4-5} \cmidrule(l{2pt}r{2pt}){6-7}
& Train. & Test. & Train. & Test. & Train. & Test. \\
\midrule
 \algoname\ (L) & 238  & 2  & 602 & 9  & 2521 & 16 \\
BP (L)  & 10   & 15 & 27 & 38 & 63  & 84 \\
 \algoname\ (S) & 28   & 4  & 57 & 4  & 161 &  5 \\
BP (S)  & 9    & 14 & 22 & 35 & 51  & 70 \\
\bottomrule
\end{tabularx}
\end{table}

\section{Conclusion}%
\label{sec-conclusion}
We present a new approach for predicting cascading failures in power grids by integrating physical and topological features from the grid into a hyperparametric model. Our methodology spans  the extraction of critical features,  construction of a Maximum Likelihood Estimation (MLE) optimization framework, and  application of the L-BFGS-B algorithm to learn the hyperparameters influencing the diffusion of failures. The resulting diffusion probability matrix is used to run Monte Carlo simulations, providing a robust predictive framework for cascading failures. In the future, we plan to explore how real-time data from fluctuating grid conditions can be integrated to dynamically adjust the predictive model.

\begin{acks}
This research is part of the programme DesCartes and is supported by the National Research Foundation, Prime Minister’s Office, Singapore under its Campus for Research Excellence and Technological Enterprise (CREATE) programme.
Lakshmanan's research was supported in part by a grant from NSERC (Canada).
\end{acks}

\bibliographystyle{ACM-Reference-Format}
\bibliography{refs}

\appendix
\section{Additional Experimental Results}\label{ap_add_exp}

\subsection{Other Experiments on the IEEE300 Dataset}

Fig.~\ref{fig_err_dc_1.0} shows the mean absolute error for line failure distribution and the diffusion probability matrix on the IEEE300 grid, complementing  the mean relative error results shown in Fig.~\ref{fig_err_r_dc_1.0}, once again evaluating the performance of the tested methods when the power demand changes. The variation trend of the error is consistent with the one in Fig.~\ref{fig_err_r_dc_1.0}, and in all cases, HCF exhibits a better generalization performance than the BP method.

\begin{figure*}[!t]
    \captionsetup{skip=3pt}
    \captionsetup[subfigure]{skip=3pt}
    \centering
    \begin{subfigure}[t]{0.24\linewidth}\centering
        \includegraphics[width=\textwidth]{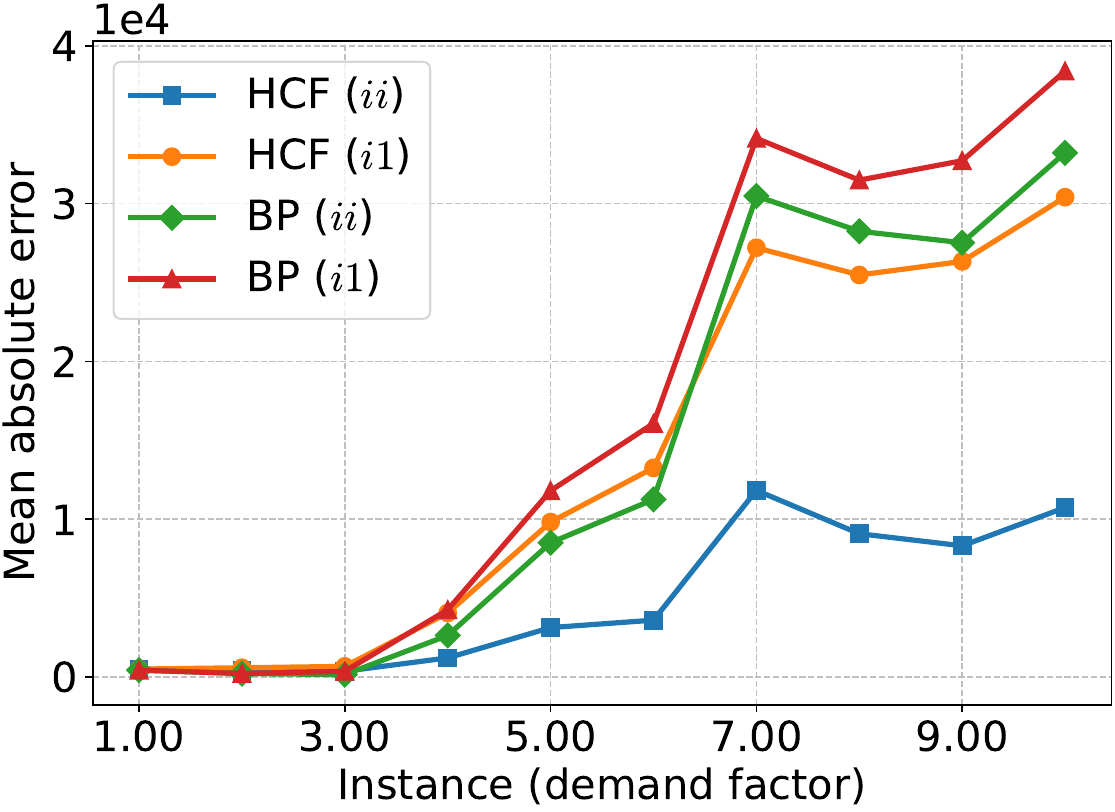}
        \caption{Distribution error}
        \label{fig_vcnt_err_dc_1.0}
    \end{subfigure}
    \hfill
    \begin{subfigure}[t]{0.24\linewidth}\centering
        \includegraphics[width=\textwidth]{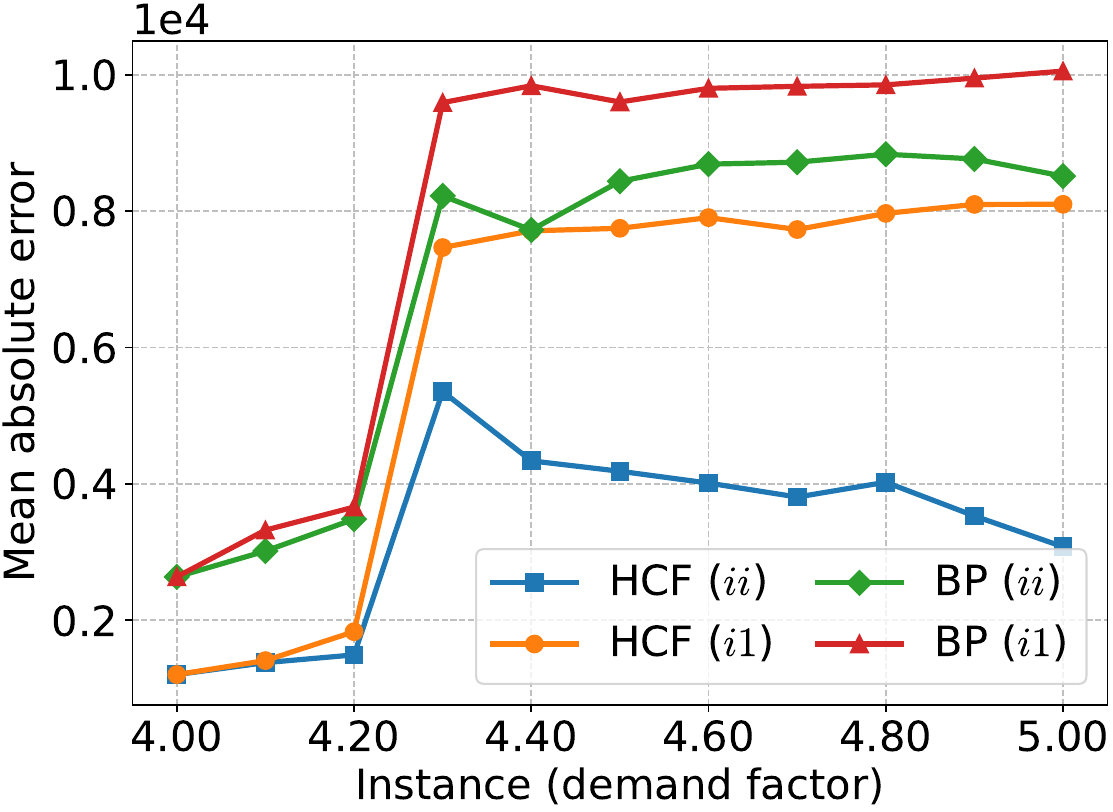}
        \caption{Distribution error}
        \label{fig_vcnt_err_dc_4.0}
    \end{subfigure}
    \hfill
    \begin{subfigure}[t]{0.24\linewidth}\centering
        \includegraphics[width=\textwidth]{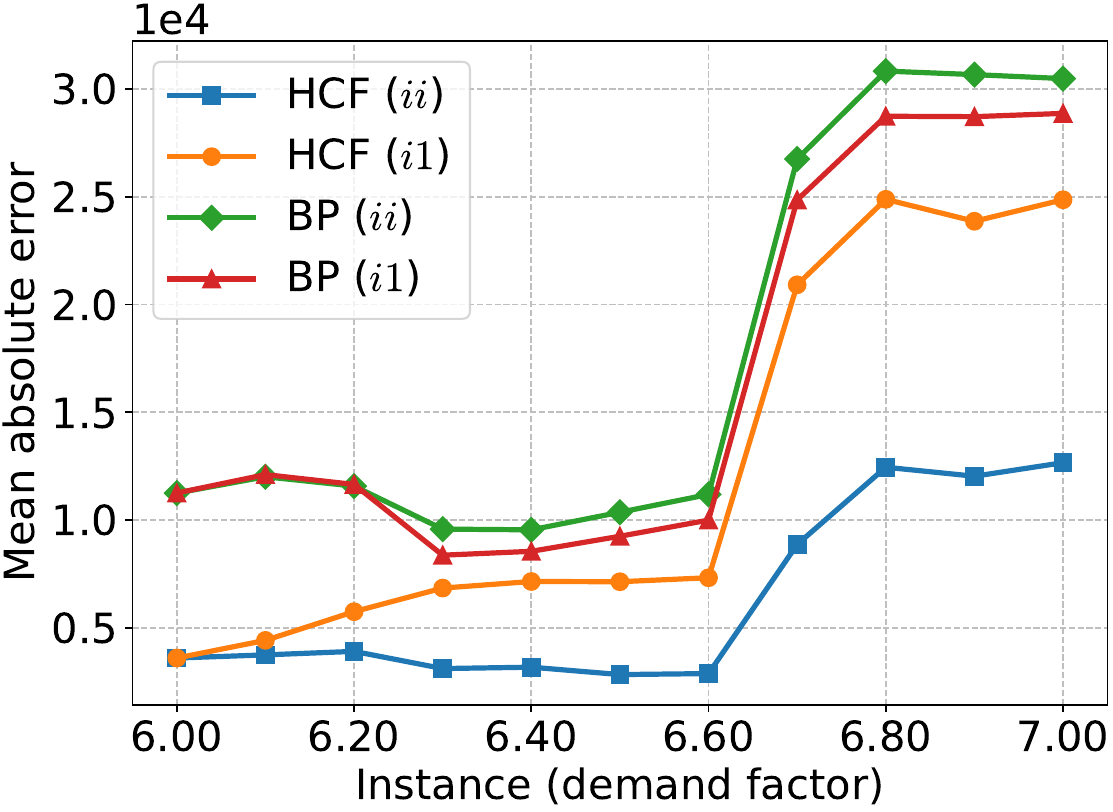}
        \caption{Distribution error}
        \label{fig_vcnt_err_dc_6.0}
    \end{subfigure}
    \hfill
    \begin{subfigure}[t]{0.24\linewidth}\centering
        \includegraphics[width=\textwidth]{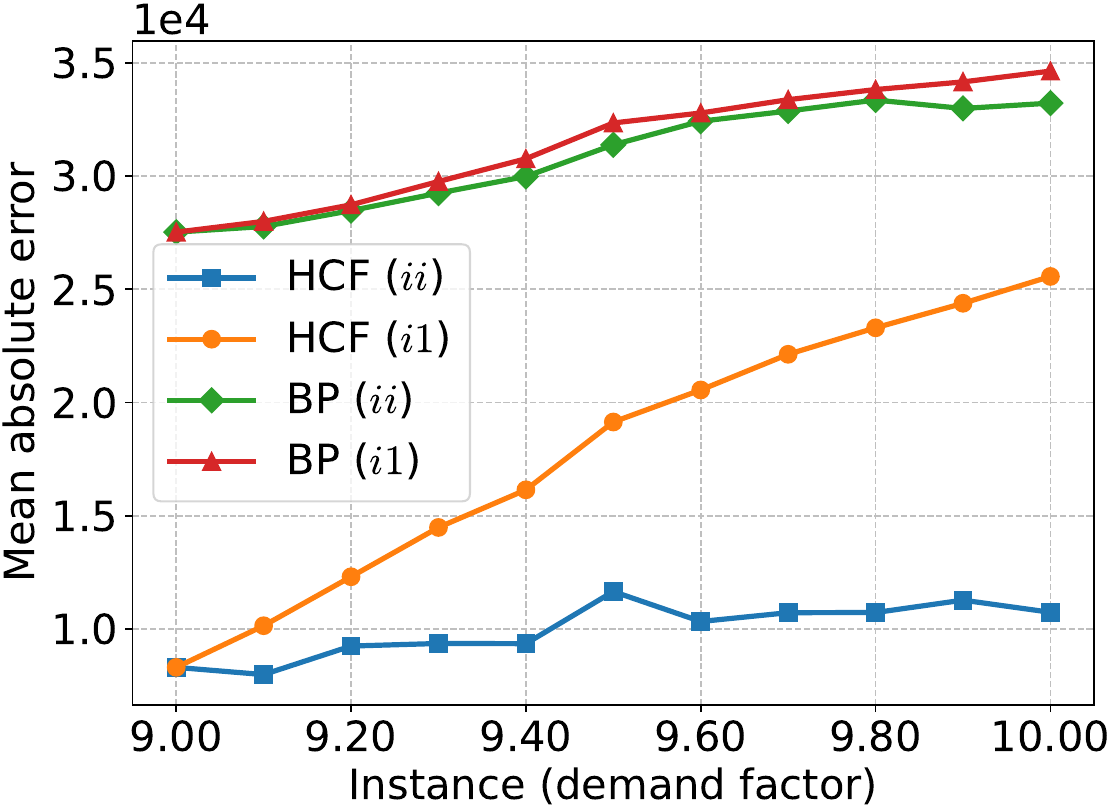}
        \caption{Distribution error}
        \label{fig_vcnt_err_dc_9.0}
    \end{subfigure}\\
    \begin{subfigure}[t]{0.24\linewidth}\centering
        \includegraphics[width=\textwidth]{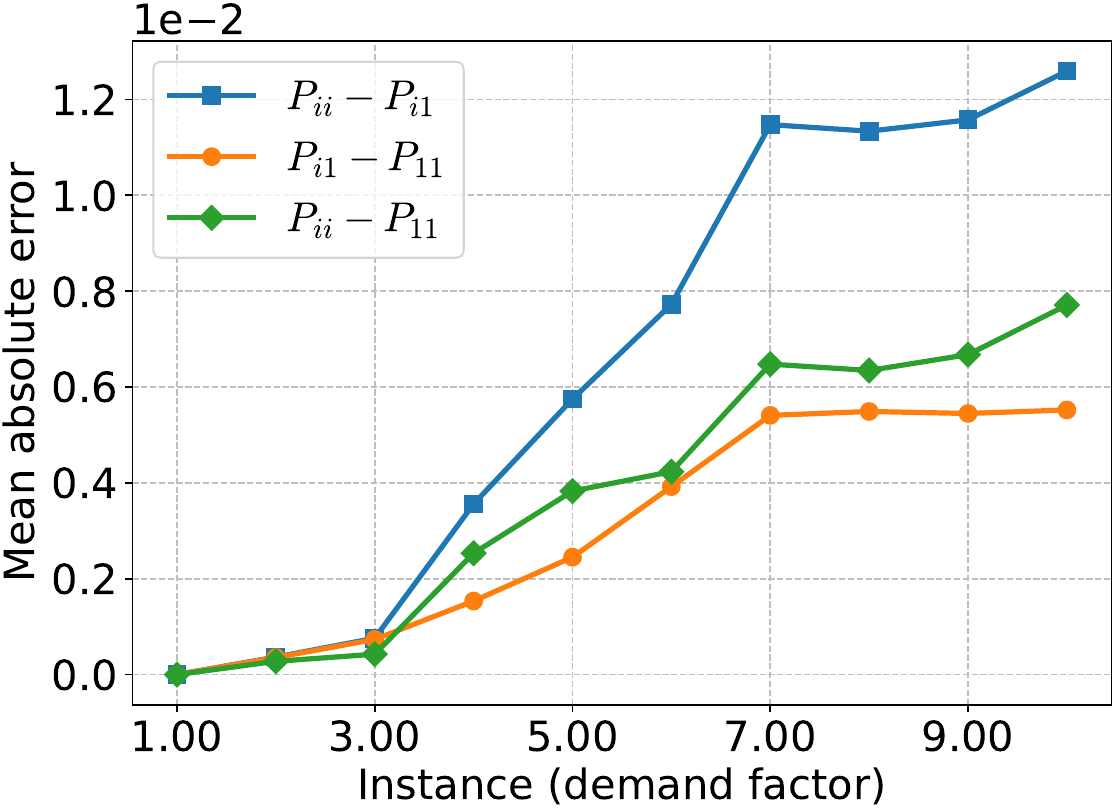}
        \caption{Probability error}
        \label{fig_pmat_err_dc_1.0}
    \end{subfigure}
    \hfill
    \begin{subfigure}[t]{0.24\linewidth}\centering
        \includegraphics[width=\textwidth]{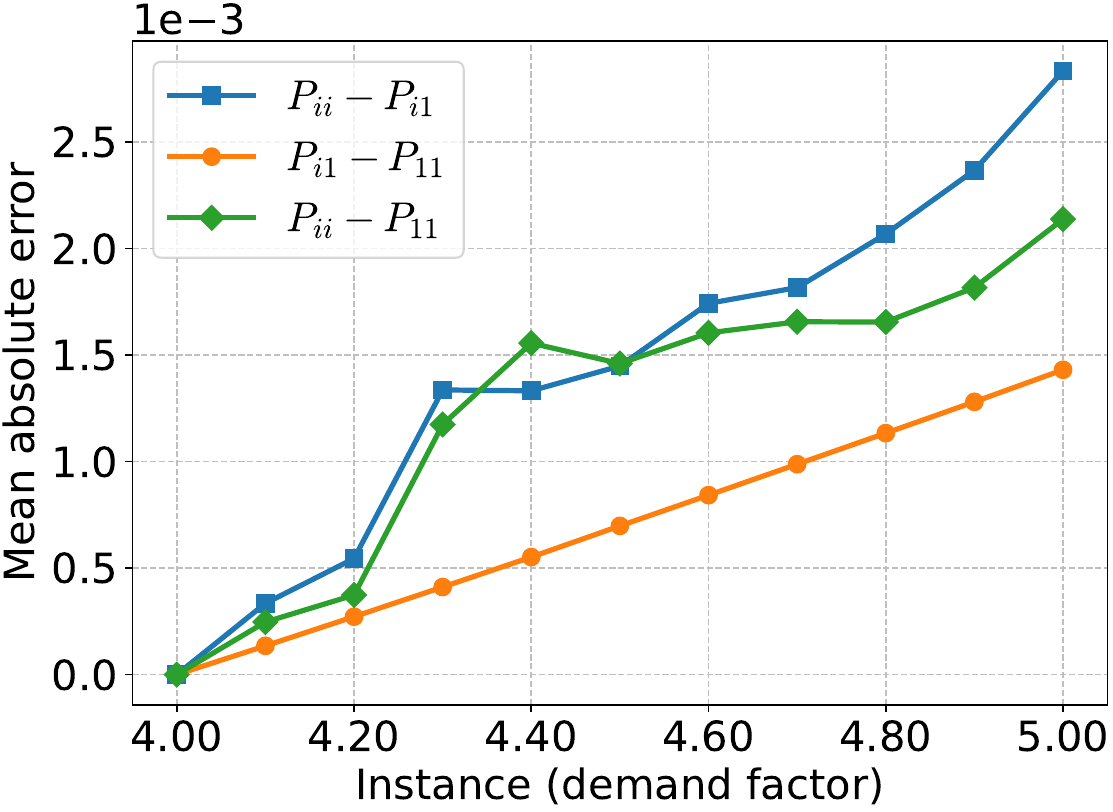}
        \caption{Probability error}
        \label{fig_pmat_err_dc_4.0}
    \end{subfigure}
    \hfill
    \begin{subfigure}[t]{0.24\linewidth}\centering
        \includegraphics[width=\textwidth]{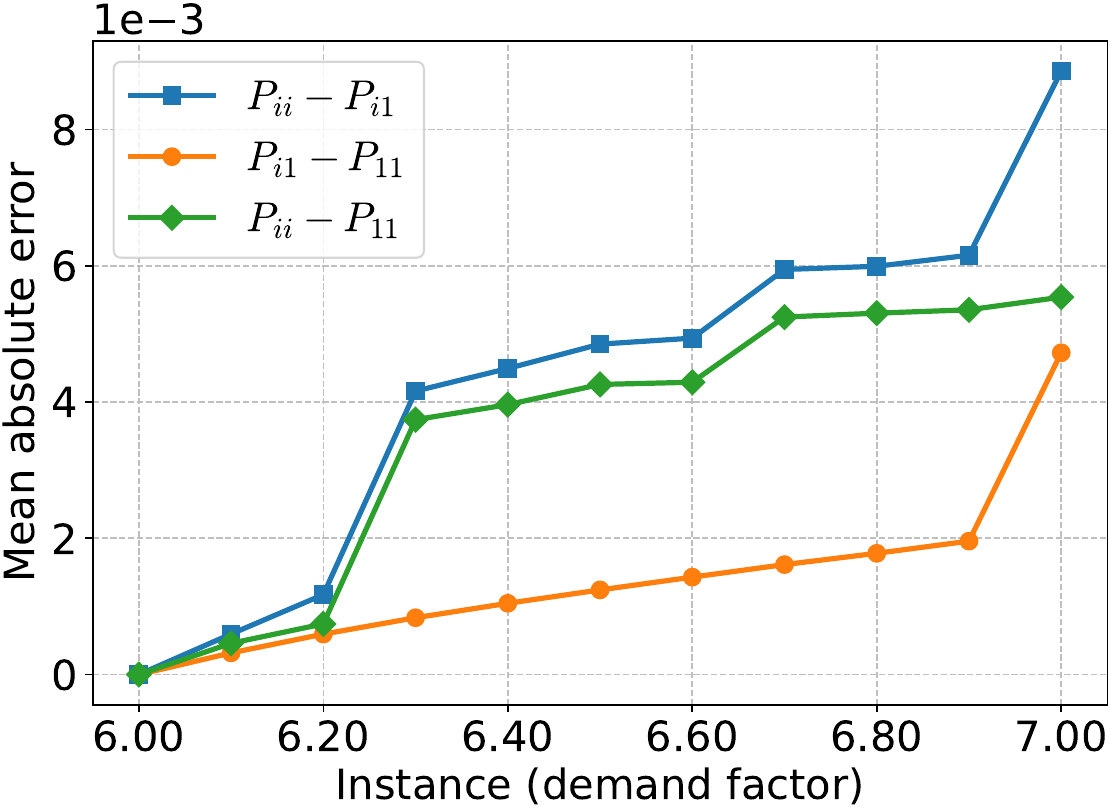}
        \caption{Probability error}
        \label{fig_pmat_err_dc_6.0}
    \end{subfigure}
    \hfill
    \begin{subfigure}[t]{0.24\linewidth}\centering
        \includegraphics[width=\textwidth]{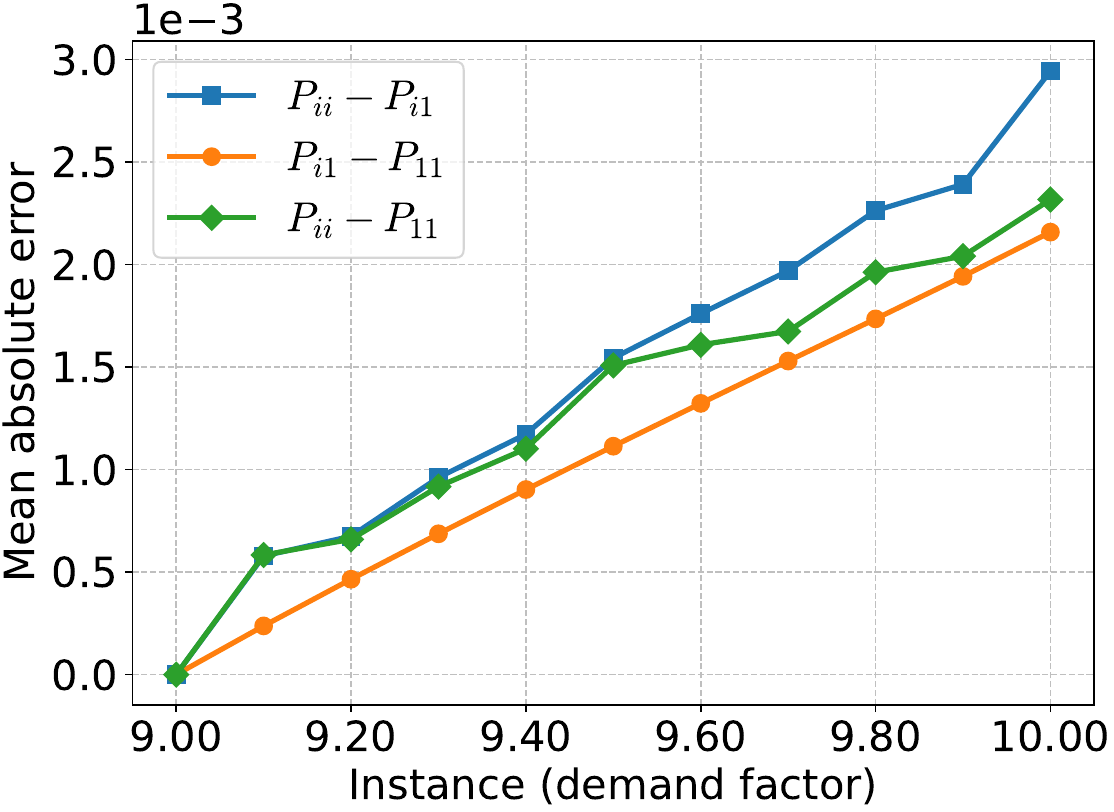}
        \caption{Probability error}
        \label{fig_pmat_err_dc_9.0}
    \end{subfigure}
    \caption{Mean absolute error for line failure distribution and diffusion probability matrix (IEEE300 dataset).}
    \label{fig_err_dc_1.0}
\end{figure*}

\subsection{Experiments on the Polish Grid}

In the domain of power systems, the power grid network IEEE300 (300 nodes, 516 edges) is a relatively large physical network corresponding to a metropolitan scale. To further evaluate the scalability of the compared methods, we expanded our experiments to test with a larger power system, the Polish power grid (2383 nodes, 2896 edges), obtained from the open-source MATPOWER framework\footnote{https://matpower.org}.

To obtain the dataset of cascading failures, the grid simulator is running up to 70h for generating 500K cascades per instance. Here we generated 10 instances with the demand factor ranging from 1.0 to 1.5. Note that when the increasing demand factor is above 1.5 the Polish power grid is over-saturated and cannot be balanced. Once more, following the evaluation design of~\cite{hines2016cascading}, only the most critical transmission lines w.r.t. the cascading failures size were considered for the evaluation (top 0.8\% lines).

From Fig.~\ref{fig_err_polish}, we can conclude that  \algoname\ has good predictive performance on this larger dataset, and achieves better results compared with BP. We would like to emphasize here that the errors of both approaches are relatively large mainly due to the filtered (critical) transmission lines on which the evaluation is done and to the rather limited sample size relative to this increased network scale.

\begin{figure*}[!t]
    \captionsetup{skip=3pt}
    \captionsetup[subfigure]{skip=3pt}
    \centering
    \begin{subfigure}[t]{0.24\linewidth}\centering
        \includegraphics[width=\textwidth]{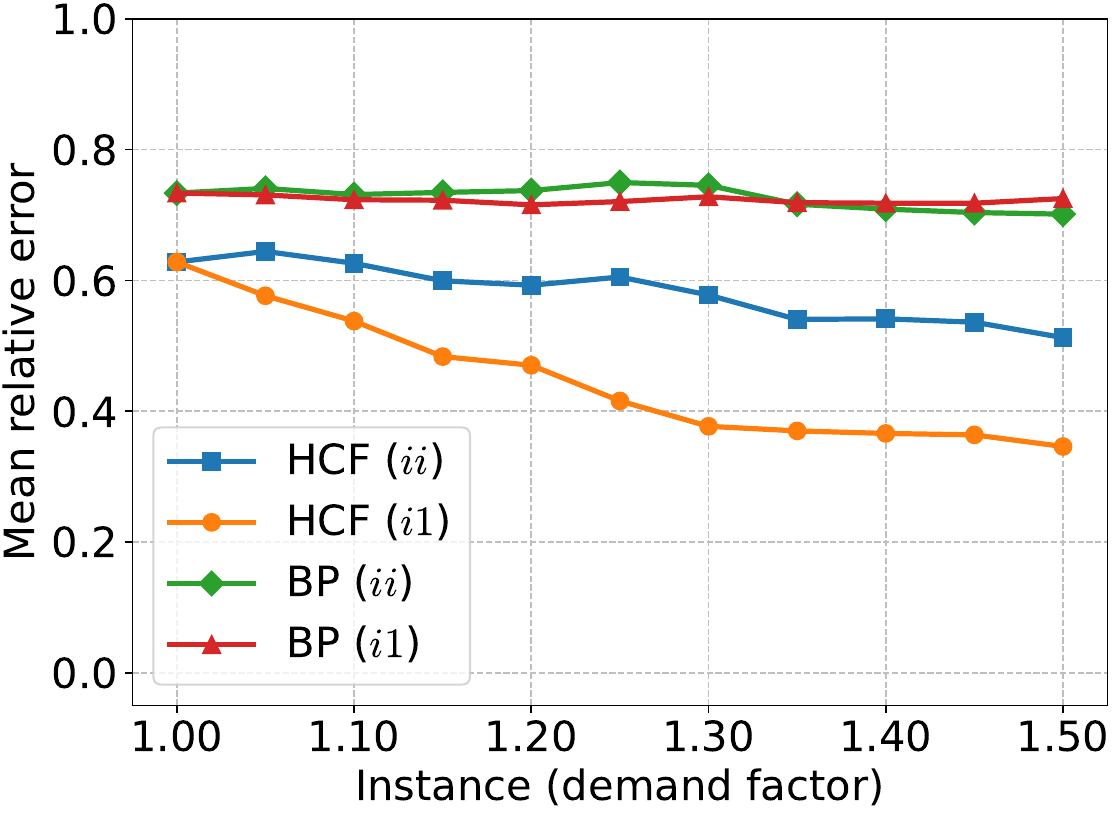}
        \caption{Relative distribution error}
        \label{fig_vcnt_err_r_polish}
    \end{subfigure}
    \hfill
    \begin{subfigure}[t]{0.24\linewidth}\centering
        \includegraphics[width=\textwidth]{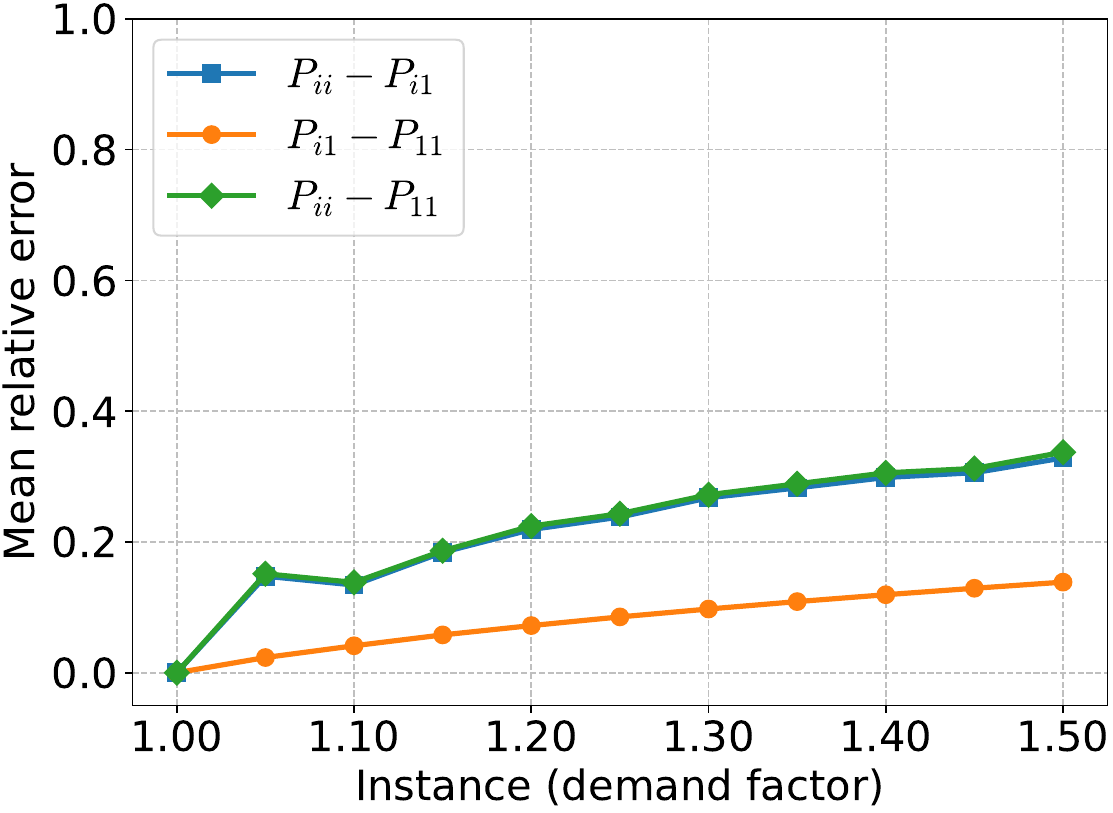}
        \caption{Relative probability error}
        \label{fig_pmat_err_r_polish}
    \end{subfigure}
    \hfill
    \begin{subfigure}[t]{0.24\linewidth}\centering
        \includegraphics[width=\textwidth]{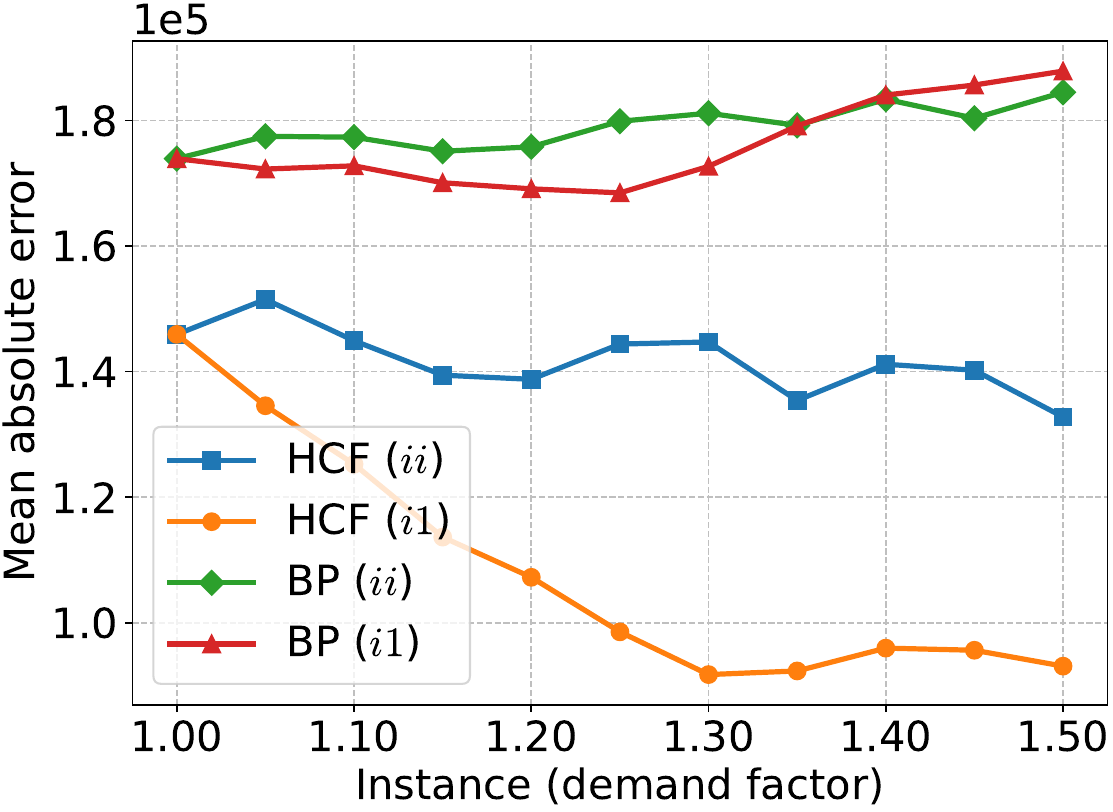}
        \caption{Absolute distribution error}
        \label{fig_vcnt_err_polish}
    \end{subfigure}
    \hfill
    \begin{subfigure}[t]{0.24\linewidth}\centering
        \includegraphics[width=\textwidth]{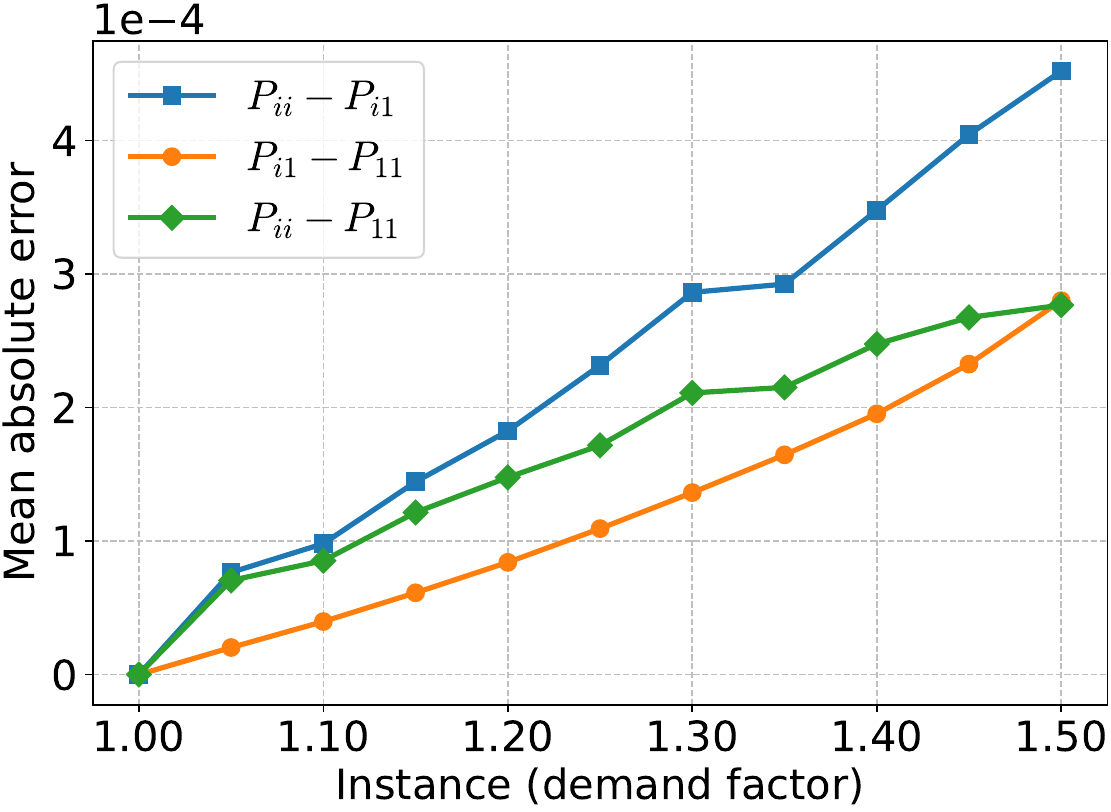}
        \caption{Absolute probability error}
        \label{fig_pmat_err_polish}
    \end{subfigure}
    \caption{Mean relative and absolute error for line failure distribution and diffusion probability matrix on Polish grid.}
    \label{fig_err_polish}
\end{figure*}

\vspace{-1mm}
\subsection{Experiments with SOTA method MK \cite{wu2021predicting}}

Despite the common limitation of existing algorithms in dealing with unseen configurations, for  a more comprehensive evaluation, we expand the experiments to include an additional method, MK~\cite{wu2021predicting} (Fig.~\ref{fig_err_r_mk}). It uses non-parametric regression to model the cascading processes. The performance of MK oscillates a lot, and for most instances, it exhibits on average poorer performance compared with the other methods, whenever the cascading failures distribution experiences significant changes as the demand factor changes. When the demand factor is within $[9,10]$ (see Fig.~\ref{fig_vcnt_err_r_mk_9.0}), cascading failures are large but the change in cascading failures distribution is small, thus MK performs better. Finally, training MK takes longer ($\sim$24h) and has a large memory footprint ($\sim$60GB) per instance, as MK entails solving a large-size constrained nonlinear optimization problem, and the testing time takes around 18 min per instance\footnote{We stress that we could not run MK on the Polish grid, due to its large computation time and memory requirements.}.

\begin{figure*}[!t]
    \captionsetup{skip=3pt}
    \captionsetup[subfigure]{skip=3pt}
    \centering
    \begin{subfigure}[t]{0.24\linewidth}\centering
        \includegraphics[width=\textwidth]{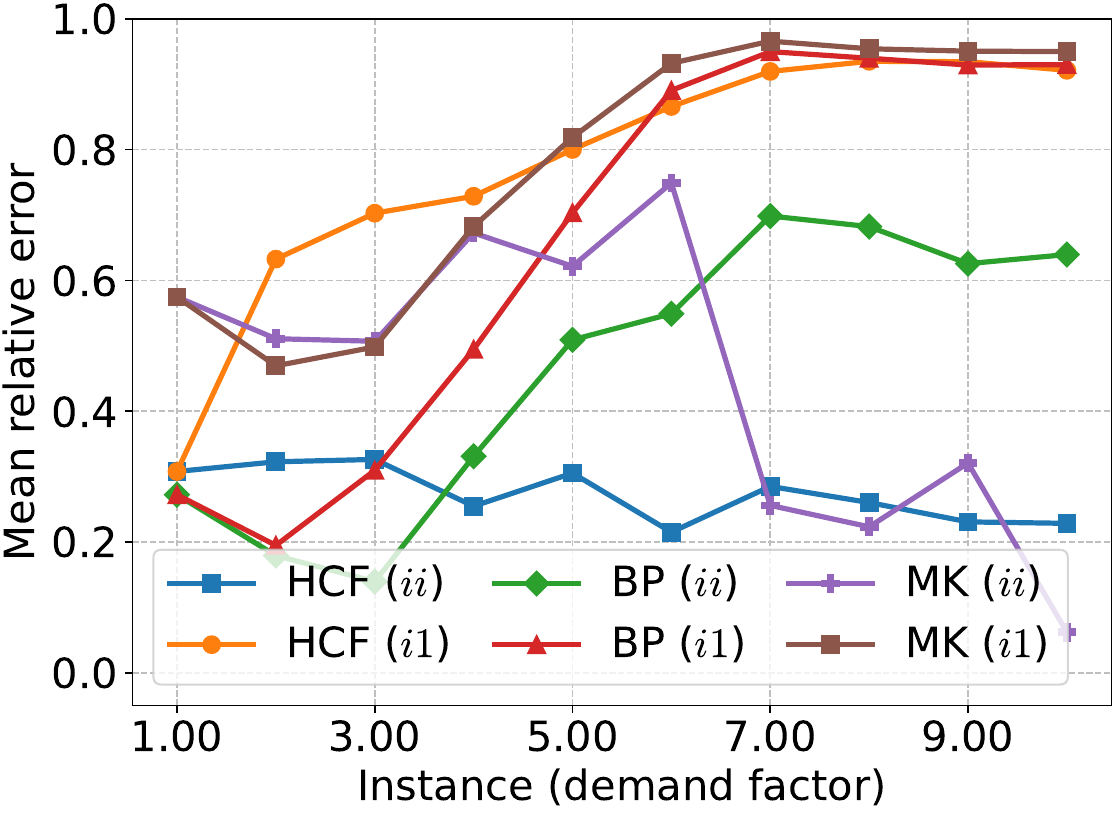}
        \caption{Distribution error}
        \label{fig_vcnt_err_r_mk_1.0}
    \end{subfigure}
    \hfill
    \begin{subfigure}[t]{0.24\linewidth}\centering
        \includegraphics[width=\textwidth]{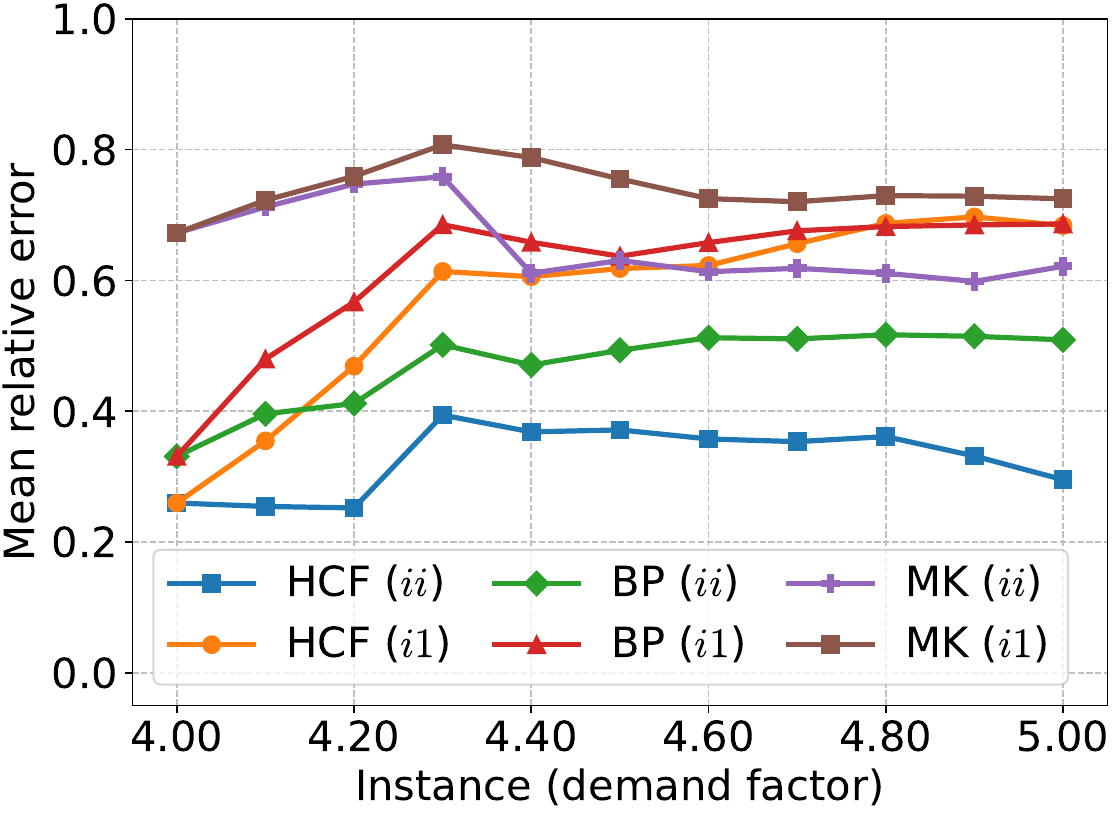}
        \caption{Distribution error}
        \label{fig_vcnt_err_r_mk_4.0}
    \end{subfigure}
    \hfill
    \begin{subfigure}[t]{0.24\linewidth}\centering
        \includegraphics[width=\textwidth]{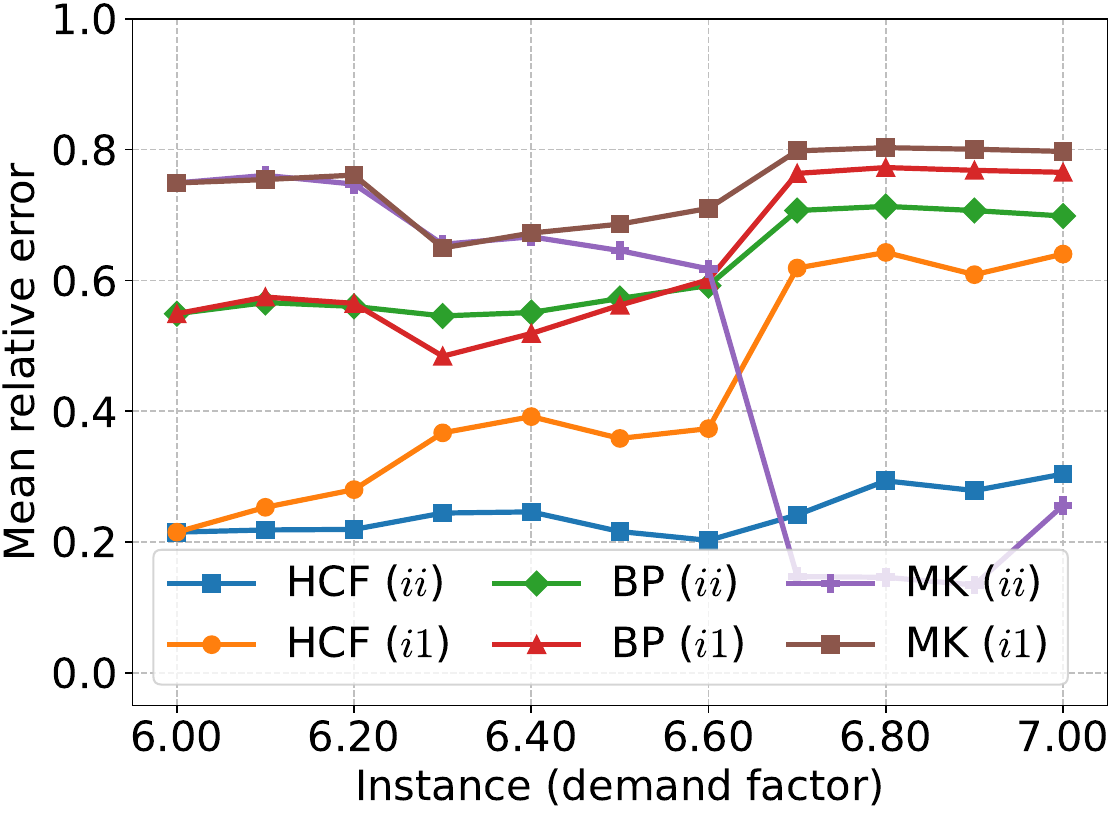}
        \caption{Distribution error}
        \label{fig_vcnt_err_r_mk_6.0}
    \end{subfigure}
    \hfill
    \begin{subfigure}[t]{0.24\linewidth}\centering
        \includegraphics[width=\textwidth]{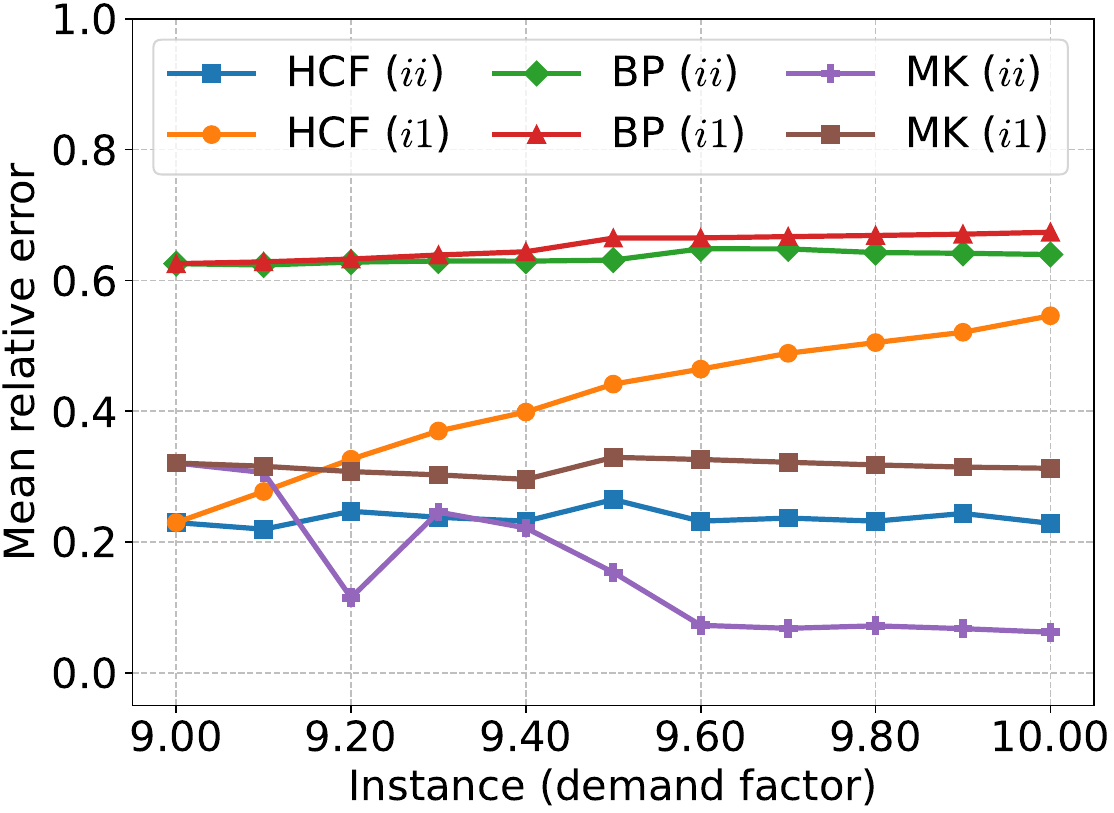}
        \caption{Distribution error}
        \label{fig_vcnt_err_r_mk_9.0}
    \end{subfigure}
    \caption{Mean relative error for line failure distribution and diffusion probability matrix on IEEE300 grid (compared with MK).}
    \label{fig_err_r_mk}
\end{figure*}

\begin{table}[!t]
\center
\captionsetup{skip=3pt}
\caption{Sensitivity w.r.t. hyperparameter dimension.}
\label{tab_sens}
\begin{tabularx}{\linewidth}{c|Y|Y|Y}
\toprule
Dimension $\backslash$ Instance & 1.0 & 5.0 & 10.0 \\
\midrule
$d=200$ & 0.33 & 0.38 & 0.37 \\
$d=250$ & 0.29 & 0.37 & 0.30 \\
$d=300$ & 0.27 & 0.36 & 0.29 \\
\bottomrule
\end{tabularx}
\end{table}

\vspace{-1mm}
\subsection{Sensitivity Analysis}

The main ingredient of our diffusion model is the hyperparameter vector, which is learned from data, where one controllable factor is its dimension. To study sensitivity to the hyperparameter's dimension, we evaluated in one experiment the changes of mean relative error when varying the dimension, as in Table~\ref{tab_sens}.

\section{Proofs}\label{proofs}

\subsection{Proof of Lemma~\ref{lemma_l_lip}}
\begin{proof}
    To study the Lipschitz continuity of the log-likelihood function, based on Eq.~\eqref{eq_dL}, we first check the gradient $\nabla_{\vec{\theta}} P(\vec{\theta}\,|\,\vec{x},\vec{s})$:
    \begin{align}
        \nabla_{\vec{\theta}} P(\vec{\theta}\,|\,\vec{x},\vec{s})
        &= \sum_{u\in\mathcal{V}_{\vec{s}}}\nabla_{\vec{\theta}}p_{uv} \prod_{u'\in\mathcal{V}_{\vec{s}}\backslash u}(1-p_{u'v})\\
        &= \left(\prod_{u\in\mathcal{V}_{\vec{s}}}(1-p_{uv})\right) \left(\sum_{u\in\mathcal{V}_{\vec{s}}}\frac{\nabla_{\vec{\theta}}p_{uv}}{1-p_{uv}}\right)\\
        &= (1-P(\vec{\theta}\,|\,\vec{x},\vec{s}))\sum_{u\in\mathcal{V}_{\vec{s}}}\frac{\nabla_{\vec{\theta}}p_{uv}}{1-p_{uv}}.
    \end{align}
    Let's define $\hat{\vec{x}}_{\vec{s}}(\vec{\theta})=\sum_{u\in\mathcal{V}_{\vec{s}}}\frac{\nabla_{\vec{\theta}}p_{uv}}{1-p_{uv}}$ (recall that $\vec{s}=(\mathcal{V}_{\vec{s}},v)$). Then,
    \begin{equation}
        \nabla_{\vec{\theta}} P(\vec{\theta}\,|\,\vec{x},\vec{s}) = (1-P(\vec{\theta}\,|\,\vec{x},\vec{s}))\hat{\vec{x}}_{\vec{s}}(\vec{\theta}).\label{eq_dP}
    \end{equation}

    For the sigmoid probability function $p_{uv}$, we have $\nabla_{\vec{\theta}}p_{uv}=(1-p_{uv})p_{uv}\vec{x}_{uv}$.
    Then,
    \begin{equation}
        \hat{\vec{x}}_{\vec{s}}(\vec{\theta})=\sum\nolimits_{u\in\mathcal{V}_{\vec{s}}}p_{uv}\vec{x}_{uv}.
    \end{equation}
    Then, we have
    \begin{align}
        \norm{\hat{\vec{x}}_{\vec{s}}(\vec{\theta})}_q
        &\leqslant\sum\nolimits_{u\in\mathcal{V}_{\vec{s}}}p_{uv}\max_{u\in\mathcal{V}_{\vec{s}}}\norm{\vec{x}_{uv}}_q\\
        &=\sum\nolimits_{u\in\mathcal{V}_{\vec{s}}}p_{uv}\sqrt[q]{d}.
    \end{align}
    For $z<1$, we have $z\leqslant -\log(1-z)$. Then,
    \begin{align}
        \sum\nolimits_{u\in\mathcal{V}_{\vec{s}}}p_{uv}
        &\leqslant-\sum\nolimits_{u\in\mathcal{V}_{\vec{s}}}\log(1-p_{uv})\\
        &=-\log(1-P(\vec{\theta}\,|\,\vec{x},\vec{s})).
    \end{align}
    Then, we have $\norm{\hat{\vec{x}}_{\vec{s}}(\vec{\theta})}_q\leqslant|\log(1-P(\vec{\theta}\,|\,\vec{x},\vec{s}))|\sqrt[q]{d}$.

    From Eq.~\eqref{eq_p_theta}, we have $P(\vec{\theta}\,|\,\vec{x},\vec{s})\in[\lambda,1-\lambda^{V_{\vec{s}}}]$, where $V_{\vec{s}}=|\mathcal{V}_{\vec{s}}|\geqslant1$ is the number of nodes activated at step $t$ in cascade $k$.

    Then, based on Eq.~\eqref{eq_l_loss}, $L(\vec{\theta}\,|\,\vec{x},\vec{s},y)\in[V_{\vec{s}}\log\lambda,\log(1-\lambda^{V_{\vec{s}}})]$ is bounded, and based on Eq.~\eqref{eq_dP}, the gradient $\nabla_{\vec{\theta}} L(\vec{\theta}\,|\,\vec{x},\vec{s},y)$ on $\vec{\theta}$ is derived as:
    \begin{align}
        \nabla_{\vec{\theta}} L(\vec{\theta}\,|\,\vec{x},\vec{s},y)
        &= \left(\frac{y}{P(\vec{\theta}\,|\,\vec{x},\vec{s})}-\frac{1-y}{1-P(\vec{\theta}\,|\,\vec{x},\vec{s})}\right)\nabla_{\vec{\theta}} P(\vec{\theta}\,|\,\vec{x},\vec{s})\\
        &= \left(\frac{y}{P(\vec{\theta}\,|\,\vec{x},\vec{s})}-1\right)\hat{\vec{x}}_{\vec{s}}(\vec{\theta}).
    \end{align}
    Then,
    \begin{align}
        \norm{\nabla_{\vec{\theta}} L(\vec{\theta}\,|\,\vec{x},\vec{s},y)}_q &= \left|\frac{y}{P(\vec{\theta}\,|\,\vec{x},\vec{s})}-1\right|\norm{\hat{\vec{x}}_{\vec{s}}(\vec{\theta})}_q\\
        &\leqslant\left|\frac{y}{P(\vec{\theta}\,|\,\vec{x},\vec{s})}-1\right| |\log(1-P(\vec{\theta}\,|\,\vec{x},\vec{s}))|\sqrt[q]{d}.
    \end{align}
    
    For $y=0$,
    \begin{equation}
        \norm{\nabla_{\vec{\theta}} L(\vec{\theta}\,|\,\vec{x},\vec{s},y)}_q \leqslant |\log(1-P(\vec{\theta}\,|\,\vec{x},\vec{s}))|\sqrt[q]{d} \leqslant V_{\vec{s}}\sqrt[q]{d}\log{\frac{1}{\lambda}}.
    \end{equation}
    
    For $y=1$,
    \begin{equation}
        \norm{\nabla_{\vec{\theta}} L(\vec{\theta}\,|\,\vec{x},\vec{s},y)}_q \leqslant \sqrt[q]{d}\frac{1-\lambda}{\lambda}\log\frac{1}{1-\lambda}. 
    \end{equation}
    
    Let $g(\lambda)=\log{\frac{1}{\lambda}}-\frac{1-\lambda}{\lambda}\log\frac{1}{1-\lambda},\lambda\in(0,0.5)$. $g(\lambda)$ is monotonically decreasing with $\lim_{\lambda\to 0.5}g(\lambda)=0$.
    Then, we have $\log{\frac{1}{\lambda}}>\frac{1-\lambda}{\lambda}\log\frac{1}{1-\lambda}$, and therefore,
    \[\norm{\nabla_{\vec{\theta}} L(\vec{\theta}\,|\,\vec{x},\vec{s},y)}_q \leqslant V_{\vec{s}}\sqrt[q]{d}\log{\frac{1}{\lambda}},\]
    which completes the proof.
\end{proof}

\subsection{Proof of Lemma~\ref{lemma_l_nonc}}
\begin{proof}
    To study the concavity of the log-likelihood function, we have to check the definiteness of the Hessian matrix $\nabla^2_{\vec{\theta}} L(\vec{\theta}\,|\,\vec{x},\vec{s},y)$.
    Based on $\nabla_{\vec{\theta}} L(\vec{\theta}\,|\,\vec{x},\vec{s},y)$ and $\nabla_{\vec{\theta}} P(\vec{\theta}\,|\,\vec{x},\vec{s})$, we derive
    \begin{align}
        \nabla^2_{\vec{\theta}} L(\vec{\theta}\,|\,\vec{x},\vec{s},y)
        &= (\frac{y}{P_{\vec{s}}}-1)\nabla_{\vec{\theta}}\hat{\vec{x}}_{\vec{s}} - \frac{y}{P_{\vec{s}}^2}\nabla_{\vec{\theta}}P_{\vec{s}}\,\hat{\vec{x}}^{\mathsf{T}}_{\vec{s}}\\
        &= (\frac{y}{P_{\vec{s}}}-1)\nabla_{\vec{\theta}}\hat{\vec{x}}_{\vec{s}} + \frac{y(P_{\vec{s}}-1)}{P_{\vec{s}}^2}\hat{\vec{x}}_{\vec{s}}\hat{\vec{x}}^{\mathsf{T}}_{\vec{s}}\\
        &=y\frac{1-P_{\vec{s}}}{P_{\vec{s}}^2}(P_{\vec{s}}\nabla_{\vec{\theta}}\hat{\vec{x}}_{\vec{s}}-\hat{\vec{x}}_{\vec{s}}\hat{\vec{x}}^{\mathsf{T}}_{\vec{s}})-(1-y)\nabla_{\vec{\theta}}\hat{\vec{x}}_{\vec{s}}\\
        &=y\frac{1-P_{\vec{s}}}{P_{\vec{s}}^2}\mathbf{Q}_{\vec{s}}-(1-y)\nabla_{\vec{\theta}}\hat{\vec{x}}_{\vec{s}},\label{eq_d2L}
    \end{align}
    where, for simplifying notations, $P_{\vec{s}}=P(\vec{\theta}\,|\,\vec{x},\vec{s})$, $\hat{\vec{x}}_{\vec{s}}=\hat{\vec{x}}_{\vec{s}}(\vec{\theta})$, and
    \begin{equation}
        \mathbf{Q}_{\vec{s}} = P_{\vec{s}}\nabla_{\vec{\theta}}\hat{\vec{x}}_{\vec{s}}-\hat{\vec{x}}_{\vec{s}}\hat{\vec{x}}^{\mathsf{T}}_{\vec{s}}.\label{eq_qs}
    \end{equation}
    In Eq.~\eqref{eq_qs},
    for the sigmoid function,
    \begin{equation}
        \nabla_{\vec{\theta}}\hat{\vec{x}}_{\vec{s}}=\sum\nolimits_{u\in\mathcal{V}_{\vec{s}}}p_{uv}(1-p_{uv})(\vec{x}_{uv}\vec{x}^{\mathsf{T}}_{uv}).\label{eq_dx_s}
    \end{equation}
    Then, $\nabla_{\vec{\theta}}\hat{\vec{x}}_{\vec{s}}$ is positive semi-definite by definition. Now we study the concavity of $L(\vec{\theta}\,|\,\vec{x},\vec{s},y)$ in terms of negative and positive samples.

    For a negative sample, $y=0$ and $\nabla^2_{\vec{\theta}}L(\vec{\theta}\,|\,\vec{x},\vec{s},y)=-\nabla_{\vec{\theta}}\hat{\vec{x}}_{\vec{s}}$ is negative semi-definite based on the definition, implying that $L(\vec{\theta}\,|\,\vec{x},\vec{s},y)$ is concave.

    For a positive sample, $y=1$ and $\nabla^2_{\vec{\theta}}L(\vec{\theta}\,|\,\vec{x},\vec{s},y)=\frac{1-P_{\vec{s}}}{P_{\vec{s}}^2}\mathbf{Q}_{\vec{s}}$.
    Now we study the definiteness of $\mathbf{Q}_{\vec{s}}$.

    Due to the linear independence of the set of features $\{\vec{x}_{uv}\,|\,\forall u\in\mathcal{V}_{\vec{s}}\}$, we have
    \begin{equation}
        \rank(\nabla_{\vec{\theta}}\hat{\vec{x}}_{\vec{s}})=|\mathcal{V}_{\vec{s}}|\geqslant
        \rank(\hat{\vec{x}}_{\vec{s}}\hat{\vec{x}}^{\mathsf{T}}_{\vec{s}})=1.
    \end{equation}
    When $\exists \vec{s}\in\mathcal{S}^+, |\mathcal{V}_{\vec{s}}|>1$,
    $\rank(\nabla_{\vec{\theta}}\hat{\vec{x}}_{\vec{s}})>\rank(\hat{\vec{x}}_{\vec{s}}\hat{\vec{x}}^{\mathsf{T}}_{\vec{s}})$.
    Then, $\exists\vec{z}\perp\hat{\vec{x}}_{\vec{s}}$ such that $\vec{z}^\mathsf{T}(\hat{\vec{x}}_{\vec{s}}\hat{\vec{x}}^{\mathsf{T}}_{\vec{s}})\vec{z}=0$ and $\vec{z}^\mathsf{T}(\nabla_{\vec{\theta}}\hat{\vec{x}}_{\vec{s}})\vec{z}\geqslant 0$. Then $\vec{z}^\mathsf{T}\mathbf{Q}_{\vec{s}}\vec{z}\geqslant 0$, $\mathbf{Q}_{\vec{s}}$ is not negative semi-definite, thus $L(\vec{\theta}\,|\,\vec{x},\vec{s},y)$ is not concave.

    For a special sample, i.e., when $|\mathcal{V}_{\vec{s}}|=1,\forall\vec{s}\in\mathcal{S}^+$, we study the definiteness of $\mathbf{Q}_{\vec{s}}$:

    we have $P_{\vec{s}}=1-\prod_{u\in\mathcal{V}_{\vec{s}}}\frac{e^{-\vec{\theta}^\mathsf{T}\vec{x}_{uv}}}{1+e^{-\vec{\theta}^\mathsf{T}\vec{x}_{uv}}}$. Then,
    $\mathbf{Q}_{\vec{s}}=\omega\vec{x}_{uv}\vec{x}^{\mathsf{T}}_{uv}$, where $\omega=-p^3_{uv}\leqslant 0$.

    Therefore, $\mathbf{Q}_{\vec{s}}$ is negative semi-definite, and $L(\vec{\theta}\,|\,\vec{x},\vec{s},y)$ is concave. The proof is completed.
\end{proof}

\subsection{Proof of Lemma~\ref{lemma_c_num}}
\begin{proof}
    Based on the definition of $\mathcal{F(S)}$, $\forall\vec{f},\vec{f}'\in\mathcal{F(S)}$ and the associated hypotheses $\vec{\theta},\vec{\theta}'\in\mathcal{H}$, we have
    \begin{equation}
        \norm{\vec{f}-\vec{f}'}_{\infty}=\max_{\vec{s},y\in\mathcal{S}}|L(\vec{\theta}\,|\,\vec{x},\vec{s},y)-L(\vec{\theta}'\,|\,\vec{x},\vec{s},y)|
    \end{equation}
    Based on the Lipschitz bound of $L(\vec{\theta}\,|\,\vec{x},\vec{s},y)$ in Lemma~\ref{lemma_l_lip}, we have:
    \begin{equation}
        \norm{\vec{f}-\vec{f}'}_{\infty}\leqslant(V\sqrt[q]{d}\log\frac{1}{\lambda})\norm{\vec{\theta}-\vec{\theta}'}_q.
    \end{equation}

    For hypothesis $\mathcal{H}=[-B,B]^d$, the $\ell_q$-norm $\varepsilon'$-cover can be created with $\mathcal{N}_q(\varepsilon',\mathcal{H})\leqslant\left(\frac{2B}{\varepsilon'}+1\right)^d$.
    Note that the radius $\varepsilon'$ is used to distinguish with the $\varepsilon$ in $\mathcal{N}_{\infty}(\varepsilon,\mathcal{F(S)})$.

    Then, for the $\varepsilon'$-covering of hypothesis $\mathcal{H}$, let $\vec{\theta}'$ be the center of one covering ball, $\norm{\vec{\theta}-\vec{\theta}'}_q\leqslant\varepsilon'$, correspondingly, we have $\norm{\vec{f}-\vec{f}'}_{\infty}\leqslant V\sqrt[q]{d}\varepsilon'\log\frac{1}{\lambda}$,
    which forms the $\ell_{\infty}$-norm covering of $\mathcal{F(S)}$ with radius $V\sqrt[q]{d}\varepsilon'\log\frac{1}{\lambda}$, and implies that
    \begin{equation}
        \mathcal{N}_{\infty}(V\sqrt[q]{d}\varepsilon'\log\frac{1}{\lambda},\mathcal{F(S)})\leqslant\mathcal{N}_q(\varepsilon',\mathcal{H}).
    \end{equation}

    Note that, when $q=\infty$, to construct the $\varepsilon'$-cover of the hypothesis $\mathcal{H}$, the space can be discretized ``evenly'' into $\mathcal{N}_{\infty}(\varepsilon',\mathcal{H})=\ceil{\frac{B}{\varepsilon'}}^d$ small $d$-dimensional hypercubes with edge length $2\varepsilon'$, which leads to a tighter bound.

    Then, based on above statements together with $\sqrt[q]{d}=1$, setting $\varepsilon'=\frac{\varepsilon}{V\log\frac{1}{\lambda}}$, $\mathcal{N}_{\infty}(\varepsilon,\mathcal{F(S)})$ can be derived, which completes the proof.
\end{proof}

\subsection{Proof of Lemma~\ref{lemma_r_bound}}
\begin{proof}
    Let $\mathcal{F}_\varepsilon(\vec{z})$ denote a subset of values $\vec{f}\in\mathcal{F(S)}$ $\varepsilon$-covered by $\vec{z}\in\mathcal{Z}$, where $\mathcal{Z}$ is a minimum cover set with $|\mathcal{Z}|=\mathcal{N}_{\infty}(\varepsilon,\mathcal{F(S)})$, then $\mathcal{F(S)}=\cup_{\vec{z}\in\mathcal{Z}}\mathcal{F}_\varepsilon(\vec{z})$.

    The Rademacher complexity of $\mathcal{F(S)}$ is defined as:
    \begin{equation}
        \mathcal{R(F(S))} = \frac{1}{m}\mathop{\mathbb{E}}_{\vec{\sigma}\sim\{\pm 1\}^m}\sup_{\vec{f}\in\mathcal{F(S)}}\innerp{\vec{\sigma},\vec{f}},
    \end{equation}
    where $\vec{\sigma}=(\sigma_i)_{i\in[m]}\in\{\pm 1\}^m$ is drawn i.i.d. from a random process with probability $Pr(\sigma_i=+1)=Pr(\sigma_i=-1)=0.5$.
    Then, based on the composition of $\mathcal{F(S)}$, we have
    \begin{align}
        m\mathcal{R(F(S))} &= \mathop{\mathbb{E}}_{\vec{\sigma}\sim\{\pm 1\}^m}\sup_{\vec{z}\in\mathcal{Z}}\sup_{\vec{f}\in\mathcal{F}_\varepsilon(\vec{z})}\innerp{\vec{\sigma},\vec{f}}\\
                        &= \mathop{\mathbb{E}}_{\vec{\sigma}\sim\{\pm 1\}^m}\sup_{\vec{z}\in\mathcal{Z}}\sup_{\vec{f}\in\mathcal{F}_\varepsilon(\vec{z})}\innerp{\vec{\sigma},\vec{z}}+\innerp{\vec{\sigma},\vec{f}-\vec{z}},\\
                        \shortintertext{based on subadditivity of $\mathop{\mathbb{E}}_{\vec{\sigma}}\sup_{\vec{z}}\sup_{\vec{f}}$:}
                        &\leqslant \mathop{\mathbb{E}}_{\vec{\sigma}\sim\{\pm 1\}^m}\sup_{\vec{z}\in\mathcal{Z}}\innerp{\vec{\sigma},\vec{z}} +\!\! \mathop{\mathbb{E}}_{\vec{\sigma}\sim\{\pm 1\}^m}\sup_{\vec{z}\in\mathcal{Z}}\sup_{\vec{f}\in\mathcal{F}_\varepsilon(\vec{z})}\innerp{\vec{\sigma},\vec{f}-\vec{z}}\\
                        &\leqslant m\mathcal{R(Z)} + \mathop{\mathbb{E}}_{\vec{\sigma}\sim\{\pm 1\}^m}\sup_{\vec{z}\in\mathcal{Z}}\sup_{\vec{f}\in\mathcal{F}_\varepsilon(\vec{z})}\norm{\vec{\sigma}}_{\infty}\norm{\vec{f}-\vec{z}}_1,\\
                        \shortintertext{based on Massart's Lemma:}
                        &\leqslant \max_{\vec{z}\in\mathcal{Z}}\norm{\vec{z}}_2\sqrt{2\log{|\mathcal{Z}|}} + m\varepsilon\\
                        &\leqslant \sqrt{m}V\log{\frac{1}{\lambda}}\sqrt{2\log{\mathcal{N}_{\infty}(\varepsilon,\mathcal{F(S)})}} + m\varepsilon.
    \end{align}
    The proof is completed by dividing $m$ from both sides of last inequality.
\end{proof}

\subsection{Proof of Lemma~\ref{lemma_s_complx}}
\begin{proof}
    Based on Rademacher complexity theory~\cite{shalev2014understanding},
    \begin{equation}
        \sup\nolimits_{\vec{\theta}\in\mathcal{H}}\mathcal{L}_{\mathcal{D}}(\vec{\theta}\,|\,\vec{x})-\mathcal{L}_{\mathcal{D}}(\hat{\vec{\theta}}\,|\,\vec{x})\leqslant 2\mathcal{R(F(S))} + 5V\log\frac{1}{\lambda}\sqrt{\frac{2\log{\frac{8}{\delta}}}{m}}.
    \end{equation}

    Based on Lemmas~\ref{lemma_c_num} and \ref{lemma_r_bound},
    \begin{align}
        \MoveEqLeft
        \sup\nolimits_{\vec{\theta}\in\mathcal{H}}\mathcal{L}_{\mathcal{D}}(\vec{\theta}\,|\,\vec{x})-\mathcal{L}_{\mathcal{D}}(\hat{\vec{\theta}}\,|\,\vec{x})\\
        &\leqslant \frac{V\log{\frac{1}{\lambda}}}{\sqrt{m}}\left(2\sqrt{2d\log\ceil*{\frac{B V\log\frac{1}{\lambda}}{\varepsilon'}}} + 5\sqrt{2\log{\frac{8}{\delta}}}\right) + 2\varepsilon'\\
        &\leqslant \frac{2V\log{\frac{1}{\lambda}}}{\sqrt{m}}\sqrt{4d\log\ceil*{\frac{B V\log\frac{1}{\lambda}}{\varepsilon'}} + 25\log{\frac{8}{\delta}}} + 2\varepsilon'.
    \end{align}
    For above inequality, let $\mathsf{RHS}\leqslant\varepsilon=3\varepsilon'$, rearranging the formulation, we have
    \begin{align}
        m&\geqslant\left(\frac{6V\log{\frac{1}{\lambda}}}{\varepsilon}\right)^2\left(4d\log\ceil*{\frac{3B V\log\frac{1}{\lambda}}{\varepsilon}} + 25\log{\frac{8}{\delta}}\right).
    \end{align}
    The proof is completed.
\end{proof}

\section{Additional Formal Results}\label{ap_add_formal}

\subsection{General Non-concavity}

 Lemma~\ref{lemma_l_nonc} analyzes the concavity of one-sample log-likelihood. Then, the expected log-likelihood $\mathcal{L}_{\mathcal{S}}(\vec{\theta}\,|\,\vec{x})$ depends on the dataset $\mathcal{S}$. Here we construct two examples showing that $\mathcal{L}_{\mathcal{S}}(\vec{\theta}\,|\,\vec{x})$ can be concave or non-concave on different datasets.
 i) $\mathcal{S}$ contains only negative or one-effective positive samples directly following Lemma~\ref{lemma_l_nonc}, then $\mathcal{L}_{\mathcal{S}}(\vec{\theta}\,|\,\vec{x})$ over $\mathcal{S}$ is concave.
 ii) $\mathcal{S}$ contains only one specific type of positive samples, written as $\{((\mathcal{V}\backslash u,\,u), 1), \forall u\in\mathcal{V}\}$, where $|\mathcal{V}|>1$. In this case, there is no implication of other negative samples. Then, based on Lemma~\ref{lemma_l_nonc}, $\mathcal{L}_{\mathcal{S}}(\vec{\theta}\,|\,\vec{x})$ over $\mathcal{S}$ is non-concave.

\subsection{Conditional Concavity}
We revisit Lemma~\ref{lemma_l_nonc} to study in depth the concavity of $\mathcal{L}_{\mathcal{S}}(\vec{\theta}\,|\,\vec{x})$ on a more general dataset instead of special cases, and achieve the following results.
\begin{lemma}\label{lemma_phi}
    Let $\mathbf{I}_{\vec{s}}(uv)$ indicate if an interaction $uv$ is covered by $\vec{s}$ or not, i.e., let $\vec{s}=(\mathcal{V}_{\vec{s}},v')$, $\mathbf{I}_{\vec{s}}(uv)=1$ if $u\in\mathcal{V}_{\vec{s}}\land v=v'$ and $0$ otherwise, and
\begin{equation}
    \phi_{uv} = \sum_{\vec{s}\in\mathcal{S}^+}\frac{1}{P_{\vec{s}}}\mathbf{I}_{\vec{s}}(uv)-\sum_{\vec{s}\in\mathcal{S}}\mathbf{I}_{\vec{s}}(uv),\label{eq_phi}
\end{equation}
where $\sum_{\vec{s}\in\mathcal{S}^+}\frac{1}{P_{\vec{s}}}\mathbf{I}_{\vec{s}}(uv)$ is the weighted size of positive samples $\mathcal{S}^+$ that cover $uv$, and $\sum_{\vec{s}\in\mathcal{S}}\mathbf{I}_{\vec{s}}(uv)$ the size of all samples that cover $uv$.
Let $\mathcal{C}^+=\{uv,\,\forall(\mathcal{V}_{\vec{s}},v)\in\mathcal{S}^+,u\in\mathcal{V}_{\vec{s}}\}$ represent the set of effective interactions covered by positive samples.
Then, if $\phi_{uv}\leqslant 0,\forall uv\in\mathcal{C}^+$, $\mathcal{L}_{\mathcal{S}}(\vec{\theta}\,|\,\vec{x})$ is concave.
\end{lemma}
\begin{proof}
    Given a general cascading failures dataset $\mathcal{S}$, we have
    \begin{align}
        \rank\left(\sum\nolimits_{\vec{s},y\in\mathcal{S}}\nabla_{\vec{\theta}}\hat{\vec{x}}_{\vec{s}}\right) &= \min(d,\,\mathsf{perm}(|\mathcal{V}|,2)),\\
        \rank\left(\sum\nolimits_{\vec{s},y\in\mathcal{S}}\hat{\vec{x}}_{\vec{s}}\hat{\vec{x}}^{\mathsf{T}}_{\vec{s}}\right) &= \rank([\hat{\vec{x}}_{\vec{s}}]_{\mathcal{S}}),
    \end{align}
    where $d=\mathsf{dim}(\vec{x})$ and $\mathsf{perm}(|\mathcal{V}|,2)=|\mathcal{V}|(|\mathcal{V}|-1)$ is the total number of feature vectors. For a large dataset $\mathcal{S}$, with high probability, we have $\rank([\hat{\vec{x}}_{\vec{s}}]_{\mathcal{S}})=d\ll\mathsf{perm}(|\mathcal{V}|,2)$, then,
    \begin{equation}
        \rank\left(\sum\nolimits_{\vec{s},y\in\mathcal{S}}\nabla_{\vec{\theta}}\hat{\vec{x}}_{\vec{s}}\right)
        =\rank\left(\sum\nolimits_{\vec{s},y\in\mathcal{S}}\hat{\vec{x}}_{\vec{s}}\hat{\vec{x}}^{\mathsf{T}}_{\vec{s}}\right).
    \end{equation}
    Then the rank analysis approach used in Lemma~\ref{lemma_l_nonc} is not applicable here.
    Now we further check the definiteness of the Hessian matrix:
    \begin{align}
        \MoveEqLeft
        \nabla^2_{\vec{\theta}}\mathcal{L}_{\mathcal{S}}(\vec{\theta}\,|\,\vec{x})
        =\frac{1}{|\mathcal{S}|}\left(\sum_{\vec{s}\in\mathcal{S}^+}\frac{1-P_{\vec{s}}}{P_{\vec{s}}^2}\mathbf{Q}_{\vec{s}}-\sum_{\vec{s}\in\mathcal{S}^-}\nabla_{\vec{\theta}}\hat{\vec{x}}_{\vec{s}}\right)\\
        &=\frac{1}{|\mathcal{S}|}\left(\sum_{\vec{s}\in\mathcal{S}^+}\frac{1}{P_{\vec{s}}}\nabla_{\vec{\theta}}\hat{\vec{x}}_{\vec{s}}-\sum_{\vec{s}\in\mathcal{S}}\nabla_{\vec{\theta}}\hat{\vec{x}}_{\vec{s}}-\sum_{\vec{s}\in\mathcal{S}^+}\frac{1-P_{\vec{s}}}{P_{\vec{s}}^2}\hat{\vec{x}}_{\vec{s}}\hat{\vec{x}}^{\mathsf{T}}_{\vec{s}}\right)\\
        &=\frac{1}{|\mathcal{S}|}\left(\mathbf{R}_{\mathcal{S}}-\sum_{\vec{s}\in\mathcal{S}^+}\frac{1-P_{\vec{s}}}{P_{\vec{s}}^2}\hat{\vec{x}}_{\vec{s}}\hat{\vec{x}}^{\mathsf{T}}_{\vec{s}}\right),\label{eq_d2L_r}
    \end{align}
    where
    \begin{equation}
        \mathbf{R}_{\mathcal{S}}=\sum_{\vec{s}\in\mathcal{S}^+}\frac{1}{P_{\vec{s}}}\nabla_{\vec{\theta}}\hat{\vec{x}}_{\vec{s}}-\sum_{\vec{s}\in\mathcal{S}}\nabla_{\vec{\theta}}\hat{\vec{x}}_{\vec{s}}.
    \end{equation}
    Based on Eq.~\eqref{eq_dx_s},
    \begin{equation}
        \nabla_{\vec{\theta}}\hat{\vec{x}}_{\vec{s}}=\sum\nolimits_{u\in\mathcal{V}_{\vec{s}}}\psi_{uv}(\vec{x}_{uv}\vec{x}^{\mathsf{T}}_{uv}),
    \end{equation}
    where
    $\psi_{uv}=p_{uv}(1-p_{uv})>0$ for sigmoid. Then,
    \begin{align}
        \mathbf{R}_{\mathcal{S}}
        &=\sum_{\vec{s}\in\mathcal{S}^+}\sum_{u\in\mathcal{V}_{\vec{s}}}\frac{1}{P_{\vec{s}}}\psi_{uv}(\vec{x}_{uv}\vec{x}^{\mathsf{T}}_{uv})-\sum_{\vec{s}\in\mathcal{S}}\sum_{u\in\mathcal{V}_{\vec{s}}}\psi_{uv}(\vec{x}_{uv}\vec{x}^{\mathsf{T}}_{uv})\\
        &=\sum_{uv\in\mathcal{V}^2,u\neq v}\left\{\psi_{uv}(\vec{x}_{uv}\vec{x}^{\mathsf{T}}_{uv})\left(\sum_{\vec{s}\in\mathcal{S}^+}\frac{1}{P_{\vec{s}}}\mathbf{I}_{\vec{s}}(uv)-\sum_{\vec{s}\in\mathcal{S}}\mathbf{I}_{\vec{s}}(uv)\right)\right\}\\
        &=\sum_{uv\in\mathcal{V}^2,u\neq v}\phi_{uv}\psi_{uv}(\vec{x}_{uv}\vec{x}^{\mathsf{T}}_{uv}).
    \end{align}
    Let $\mathcal{C}=\{uv,\,\forall uv\in\mathcal{V}^2,u\neq v\}$ represent the set of all interactions.
    Then, $\forall uv\in\mathcal{C}-\mathcal{C}^+,\,\vec{s}\in\mathcal{S}^+$, $\mathbf{I}_{\vec{s}}(uv)=0$, i.e., $uv$ is not covered by the positive samples, and $\phi_{uv}=-\sum_{\vec{s}\in\mathcal{S}}\mathbf{I}_{\vec{s}}(uv)\leqslant 0$. 
    
    Then, if $\phi_{uv}\leqslant 0,\,\forall uv\in\mathcal{C}^+$, $\mathbf{R}_{\mathcal{S}}$ will be negative semi-definite by definition, and $\nabla^2_{\vec{\theta}}\mathcal{L}_{\mathcal{S}}(\vec{\theta}\,|\,\vec{x})$ will be negative semi-definite.
    The proof is completed.
\end{proof}

Lemma~\ref{lemma_phi} provide a condition for the expected log-likelihood being concave. Then, given a dataset, it's always able to check beforehand the concavity and decide which algorithm to use. Note that Lemma~\ref{lemma_phi} evaluates the concavity by each time looking at the entire set or subset of dataset that has been selected for computing. Now we start to review Lemma~\ref{lemma_phi} on a general dataset.

    Let $N_{\mathcal{A}}(uv)$ denote the number of samples in $\mathcal{A}$ that cover $uv$. Then, we have $N_{\mathcal{S}}(uv)=\sum_{\vec{s}\in\mathcal{S}}\mathbf{I}_{\vec{s}}(uv)$ in Lemma~\ref{lemma_phi}, and $N_{\mathcal{S}^+}(uv)=\sum_{\vec{s}\in\mathcal{S}^+}\mathbf{I}_{\vec{s}}(uv)$.
    Let $\overline{P}^{-1}_{uv}=\frac{\sum_{\vec{s}\in\mathcal{S}^+}{P^{-1}_{\vec{s}}}\mathbf{I}_{\vec{s}}(uv)}{\sum_{\vec{s}\in\mathcal{S}^+}\mathbf{I}_{\vec{s}}(uv)}$.   
Based on Eq.~\eqref{eq_phi}, we have: 
\begin{align}
    \phi_{uv} &=\overline{P}^{-1}_{uv}N_{\mathcal{S}^+}(uv)-N_{\mathcal{S}}(uv)\\
              &=\frac{1}{N_{\mathcal{S}}(uv)}\left(\frac{\tilde{p}_{uv}}{\overline{P}_{uv}}-1\right),\label{eq_phi_n}
\end{align}
where $\tilde{p}_{uv}=\frac{N_{\mathcal{S}^+}(uv)}{N_{\mathcal{S}}(uv)}$ is an empirical estimation of $p_{uv}$.
By definition, we have
\begin{equation}
    \overline{P}_{uv}\geqslant\min_{\vec{s}\in\mathcal{S}^+:\mathbf{I}_{\vec{s}}(uv)=1}P_{\vec{s}}\geqslant p_{uv}.\label{ineq_p}
\end{equation}
Then, we have
\begin{equation}
    \phi_{uv}\leqslant\frac{1}{N_{\mathcal{S}}(uv)}\left(\frac{\tilde{p}_{uv}}{p_{uv}}-1\right).
\end{equation}

Without any details on a general dataset, given a positive sample $\vec{s}\in\mathcal{S}^+$ with $|\mathcal{V}_{\vec{s}}|>1$, the empirical estimation $\tilde{p}_{uv}$ of $p_{uv}$ is larger than $p_{uv}$ for some $uv$, i.e., $\tilde{p}_{uv}\geqslant p_{uv}$ due to the combination effect from $\mathcal{V}_{\vec{s}}$ and the over-counting effect of empirical estimation from $\frac{N_{\mathcal{S}^+}(uv)}{N_{\mathcal{S}}(uv)}$. In this case, $\text{RHS}\geqslant 0$, hence 
we cannot determine the sign of $\phi_{uv}$, thus the concavity of $\mathcal{L}_{\mathcal{S}}(\vec{\theta}\,|\,\vec{x})$.
However, we stress that $\text{RHS}$ is a quite loose 
estimated bound for $\phi_{uv}$, taking the minimum value $p_{uv}$.

    For a general dataset, maximizing $\mathcal{L}_{\mathcal{S}}(\vec{\theta}\,|\,\vec{x})$ will make the sample probability $P_{\vec{s}},\forall \vec{s}\in\mathcal{S}^+$ approach to 1, thus, $\overline{P}_{uv}$ also approach to 1, and $\phi_{uv}\approx N_{\mathcal{S}^+}(uv)-N_{\mathcal{S}}(uv)\leqslant 0$. Therefore, if $|\mathcal{S}|\gg|\mathcal{S}^+|$, which could be frequently seen in cascading failures scenario, with certain probability, $\mathcal{L}_{\mathcal{S}}(\vec{\theta}\,|\,\vec{x})$ tend to be concave.

\subsection{Discussion on General Dataset}
Here we analyze the property of a general dataset $\mathcal{S}$.
The cascading influence graph is a complete graph in this context, then every event will involve interactions from some nodes to the rest in $\mathcal{V}$ (see Section~\ref{sec-system-model}), i.e., given a sample $(\vec{s},y)\in\mathcal{S}$ where $\vec{s}=(\mathcal{V}_{\vec{s}},v)$, it implies that $\mathcal{S}$ also contains other positive or negative samples $\{(\vec{s}',y')\,|\,\vec{s}'=(\mathcal{V}_{\vec{s}},v'),v'\in\mathcal{V}\backslash v-\mathcal{V}_{\vec{s}}\}$.

Then, for a general dataset $\mathcal{S}$ with $|\mathcal{S}|\gg |\mathcal{V}|$, with very high probability, $\mathcal{S}$ can cover all the interactions among $|\mathcal{V}|$. This covering problem can be described as: given $\mathcal{V}$, each node has equal probability being taken; if every round randomly picking nodes from $\mathcal{V}$ then returning them, what is the probability for $|\mathcal{S}|$ rounds picking to ensure every node being taken at least one time?
The covering probability is then $P_{cover}=(1-(1-\frac{1}{|\mathcal{V}|})^{|\mathcal{S}|})^{|\mathcal{V}|}\geqslant 1-|\mathcal{V}|e^{-{|\mathcal{S}|}/{|\mathcal{V}|}}$. For instance, taking a large power network with $|\mathcal{V}|=1000$ if $|\mathcal{S}|=21|\mathcal{V}|$, $P_{cover}$ is larger than $1-10^{-6}$.

\end{document}